\documentclass[a4paper,11pt]{article}
\usepackage{a4wide}
\usepackage{hyperref}
\usepackage{graphicx}
\usepackage{amsmath,amssymb,amsthm}

\theoremstyle{definition}
\newtheorem{definition}{Definition}
\newtheorem{example}{Example}
\newtheorem{remark}{Remark}

\theoremstyle{plain}

\newtheorem{proposition}{Proposition}
\newtheorem{theorem}{Theorem}
\newtheorem{corollary}{Corollary}

\usepackage{stmaryrd}

\usepackage[linesnumbered,ruled,vlined,hanginginout]{algorithm2e}

\SetKwInput{KwInput}{Input}
\SetKwInput{KwOutput}{Output}

\newcommand{\Gr}{Gr\"obner }

\newcommand{\algref}[1]{{\tt #1}}

\DeclareMathOperator{\Sol}{Sol}

\DeclareMathOperator{\Mon}{Mon}

\DeclareMathOperator{\lm}{lm}
\DeclareMathOperator{\lc}{lc}

\DeclareMathOperator{\NF}{NF}

\DeclareMathOperator{\ld}{ld}
\DeclareMathOperator{\init}{init}

\DeclareMathOperator{\disc}{disc}
\DeclareMathOperator{\sep}{sep}

\newcommand{\setderiv}{\Delta}

\newcommand{\automorphism}{\sigma}
\newcommand{\ranking}{\succ}
\newcommand{\setauto}{\Sigma}

\newcommand{\algvar}{y}

\newcommand{\domain}{\mathcal{D}}
\newcommand{\differentialfield}{\mathcal{K}}
\newcommand{\differentialring}{\mathcal{R}}
\newcommand{\differentialideal}{\mathcal{I}}
\newcommand{\differentialranking}{\succ}
\newcommand{\differentialrankeq}{\succeq}

\newcommand{\differencefield}{\tilde{\mathcal{K}}}
\newcommand{\differencering}{\tilde{\mathcal{R}}}
\newcommand{\differenceideal}{\tilde{\mathcal{I}}}
\newcommand{\differenceranking}{\succ}

\newcommand{\differencemonomialordering}{\sqsupset}

\newcommand{\diffidealgen}[1]{[#1]}

\newcommand{\cM}{\mathcal{M}}
\newcommand{\cI}{\mathcal{I}}
\newcommand{\cR}{\mathcal{R}}

\newcommand{\Q}{\mathbb{Q}}
\newcommand{\C}{\mathbb{C}}
\newcommand{\Z}{\mathbb{Z}}

\newcommand{\N}{\mathbb{N}}

\newcommand{\R}{\mathbb{R}}

\begin{document}

\title{Strong Consistency and Thomas Decomposition of\\ Finite Difference Approximations to Systems of\\ Partial Differential Equations}

\author{Vladimir P.\ Gerdt\thanks{Laboratory of Information Technologies, Joint Institute for Nuclear Research, 6 Joliot-Curie Str., 141980 Dubna, Russian Federation  and Peoples' Friendship University of Russia (RUDN), 6 Miklukho-Maklaya Str, Moscow, 117198, Russian Federation, \texttt{gerdt@jinr.ru}}
\and Daniel Robertz\thanks{School of Engineering, Computing and Mathematics, University of Plymouth, 2-5 Kirkby Place, Drake Circus, Plymouth PL4 8AA, Devon, United Kingdom, \texttt{daniel.robertz@plymouth.ac.uk}}
\and Yuri A.\ Blinkov\thanks{Faculty of Mathematics and Mechanics, National Research Saratov State University, 83 Astrakhanskaya St, Saratov, 410012, Russian Federation, \texttt{blinkovua@info.sgu.ru}}
}


\date{}

\maketitle

\begin{abstract}
For a wide class of polynomially nonlinear systems of partial differential equations we suggest an algorithmic approach that combines differential and difference algebra to analyze s(trong)-consistency of finite difference approximations. Our approach is applicable to regular solution grids. For the grids of this type we give  a  new  definition  of  s-consistency  for  finite difference  approximations  which  generalizes  our  definition  given  earlier  for Cartesian  grids. The algorithmic verification of s-consistency presented in the paper is based  on  the  use  of  both differential and difference Thomas decomposition.  First, we apply the differential decomposition to the input system, resulting in a partition of its solution space. Then, to the  output  subsystem  that  contains  a  solution  of  interest we apply a difference analogue of the differential Thomas decomposition which allows to check the s-consistency. For linear and some quasi-linear differential systems one can also apply difference \Gr bases for the s-consistency analysis. We illustrate our methods and algorithms by a number of examples, which include Navier-Stokes equations for viscous incompressible flow.
\end{abstract}

%
%
\section{Introduction}\label{sec:introduction}

In the given paper we consider systems of partial differential equations (PDE):
\begin{equation}
f_1=\cdots=f_p=0,\qquad F:=\{f_1,\ldots,f_p\}\subset \differentialring\,.  \label{pde}
\end{equation}
Here $f_i$ $(i=1,\ldots, p)$ are elements of the {\em differential polynomial ring}
$\differentialring := \differentialfield\{ \mathbf{u} \}$, the ring of polynomials in the dependent variables $\mathbf{u}:=\{u^{(1)},\ldots,u^{(m)}\}$ ({\em differential indeterminates}) and their partial derivatives,
which are obtained by applying the power products of the pairwise commuting derivation operators $\setderiv := \{\partial_1,\ldots,\partial_n\}$ $(\partial_j\equiv\partial_{x_j})$. We shall assume that the coefficients of the polynomials are rational functions in $\mathbf{a}:=\{a_1,\ldots,a_l\}$, a finite number of parameters (constants), whose coefficients are rational numbers, i.e., $\differentialfield:={\Q}(\mathbf{a})$. One can also extend the last field to ${\Q}(\mathbf{a},\mathbf{x})$, where $\mathbf{x}:=\{x_1,\ldots,x_n\}$ is the set of independent variables.

Equations~\eqref{pde} arise in mathematical descriptions of many processes in natural sciences, e.g., in continuous mechanics and physics, whose dynamics evolve in space-time. Apart from very special cases, exact solutions to the governing PDE system are unknown and only numerical solutions can provide valuable information in the study of the process under investigation. For their numerical solution the differential equations in~\eqref{pde} have to be replaced by discrete counterparts. The most widely used methods of discretization and numerical solving are the method of finite elements, the method of finite volumes and the method of finite differences. The last method is historically the first~\cite{Samarskii'01} and is based on replacing differential equations by difference ones defined on a chosen solution grid. In order to construct a numerical solution, the devised finite difference approximation (FDA) to PDE is complemented with an appropriate discretization of initial or/and boundary condition(s) for the PDE. As this takes place, the quality of the numerical solution to PDE crucially depends on the quality of its FDA (difference scheme).

The main requirement for an FDA is the convergence of a numerical solution to a solution of PDE in a limit when the grid spacings tend to zero. However, except for a very limited class of problems (see~\cite{Str'04}, Thm.~1.5.1), convergence cannot be directly established. In practice, given an FDA, its consistency and stability are analyzed as the necessary conditions for convergence. Consistency implies reduction of the FDA to the original PDE when the grid spacings tend to zero, and stability provides boundedness of the error in the solution under small perturbation in the numerical data. It is pertinent to note that in the case of  nonlinear FDA (scheme) its theoretical stability analysis is highly conjectural and it is usually studied ``experimentally'' by division of the grid spacings in halves or by comparison with the exact solution if it is known.

One of the most challenging problems is to construct such difference approximations of equations~\eqref{pde} which preserve their basic algebraic properties in the discrete setting, e.g., the continuous identities and theorems of vector calculus, symmetries and conservation laws.
Such discretizations, which are called \emph{compatible} or \emph{mimetic}~\cite{Arnold'06,Lipnikov'2014,CastilloMiranda'13}, and sometimes \emph{structure preserving}~\cite{Christiansen'11}, are more likely to produce highly accurate and stable numerical results, as was observed in numerous computational experiments (cf.~\cite{JCP'14}). The most universal approach to determine invariant solutions of initial and boundary value problems for systems~\eqref{pde} and to derive their conservation laws is the Lie symmetry analysis \cite{Olver}. Certain counterparts of continuous symmetries for single differential equations were studied for finite difference schemes in \cite[Ch.~4]{Dorodnitsyn}.

In~\cite{GR'10,G'12} for systems~\eqref{pde} and Cartesian (i.e., rectilinear and equisized)
solution grids we introduced the novel concept  of {\em strong consistency}, or {\em s-con\-sis\-ten\-cy}, of FDA to PDE, which strengthens the concept of consistency. Loosely speaking, s-con\-sis\-ten\-cy of an FDA means not only approximation of \eqref{pde} by the FDA, but also approximation of every element in the radical differential ideal, generated by $\{f_1,\ldots,f_p\}$, by an element in the perfect difference ideal generated by the difference polynomials in the FDA. In the subsequent papers~\cite{ABGLS'13,ABGLS'17}, by computational experiments with two-di\-men\-sio\-nal incompressible Navier-Stokes equations, it was shown that FDA which are s-con\-sis\-tent have much better numerical behavior than FDA which are not. To verify s-con\-sis\-ten\-cy of linear FDA (to linear PDE) in \cite{GR'10} we used the algorithms and software for constructing differential and difference Janet/\Gr bases~\cite{Maple-Janet'03,GR'12}. The generalization to nonlinear PDE given in~\cite{G'12}, based on the concept of difference \Gr basis, is not algorithmic, because the difference polynomial ring is non-Noetherian~\cite{GLS'15,LaScala'15} and the basis may be infinite. In the conference paper~\cite{GR19} we suggested an algorithm for verification of s-con\-sis\-ten\-cy on Cartesian grids that is based on investigating the FDA by a difference analogue of \emph{differential Thomas decomposition}. In the special case of binomial perfect difference ideals another kind of decomposition was suggested in \cite{BinomialDifference}.

The notion of differential Thomas decomposition and its algorithmic construction stemmed from the Riquier-Janet theory~\cite{Riquier'10,Janet'29}. Wu Wen-tsun was the first who showed~\cite{Wu'90} that this theory can be used for algorithmic construction of algebraic \Gr bases.
Joseph M.~Thomas~\cite{Thomas'37,Thomas'62} generalized the Riquier-Janet theory to polynomially nonlinear systems and showed how to decompose algebraic and differential systems into triangular subsystems with disjoint solution sets. Thomas called these subsystems \emph{simple} since their structure alleviates their algebraic analysis. The first algorithmization and implementation in Maple of Thomas' approach for systems of algebraic and ordinary differential equations was achieved by Dongming Wang~\cite{Wang'01,Wang'04}. The further algorithmic development of Thomas decomposition for algebraic systems and its full algorithmization for PDE systems, incorporating the involutive algorithm for constructing Janet bases~\cite{G'05}, together with an implementation in Maple, was realized in \cite{BGLHR'12,Robertz'14,GLHR'18}.
Thomas decomposition provides \emph{regular differential chains}~\cite{Hubert'01}, which allow to test membership
to the radical differential ideal through differential Janet reduction.
Related methods are the Ro\-sen\-feld-Gr\"ob\-ner algorithm~\cite{BLOP'09} and the rif-algorithm \cite{ReidWittkopfBoulton}.
In contrast to regular differential chains generated by the Ro\-sen\-feld-Gr\"ob\-ner algorithm, the Thomas decomposition and the rif-algorithm
provide partitions of the solution sets.
Because of the last property these decompositions lend themselves to the s-con\-sis\-ten\-cy analysis.

However, the concept of s-consistency, as it was introduced in~\cite{GR'10,G'12,GR19}, is applicable to Cartesian grids only. Its generalization to more general regular grids, whose grid spacings may be pairwise different, requires certain modifications and extensions. These modifications and extensions are presented in the given paper and illustrated by a number of examples that include incompressible Navier-Stokes equations and overdetermined PDE systems. For some examples we performed not only symbolic but also numeric analysis.

The paper is organized as follows. Section~\ref{sec:differentialthomas} contains a description of differential Thomas decomposition  into simple differential systems which, in addition to equations, may include inequations. The illustrative Example~\ref{ex:differential} is given and the fundamental property of simple systems used in the s-con\-sis\-ten\-cy analysis
is formulated in Proposition~\ref{prop:differentialmembership}. In Section~\ref{sec:differenceapproximations} we consider finite difference approximations to the PDE system~\eqref{pde} on a regular grid~\eqref{grid} and define the differential and difference ideals generated by PDE and FDA, respectively.  The concept of difference \Gr basis together with related definitions and the simplest form of difference Buchberger algorithm are outlined in Section~\ref{sec:differencestandardbases}. Then, in Section~\ref{sec:consistency} we give the definition of \emph{$s$-consistency} of FDA to PDE for the grid~\eqref{grid}. In addition, we present the criterion of s-consistency in terms of difference \Gr bases.
As an example of application of this criterion,  we construct and analyze certain s-consistent FDA to the incompressible Navier-Stokes equations (Section~\ref{sec:NavierStokesEquations}). In Section~\ref{DifferenceThomasDecomposition} we define \emph{simple} and \emph{quasi-simple} difference systems and describe the algorithm of \emph{Thomas decomposition into quasi-simple difference systems}.
We prove correctness and termination of the algorithm and show how it provides the fully algorithmic check of s-con\-sis\-ten\-cy. Two examples (Examples~\ref{ex:Example1} and \ref{ex:Example2}) of quasi-linear PDE and different FDA are analyzed with respect to their s-consistency in Section~\ref{sec:examples}. Concluding remarks are given in Section~\ref{sec:conclusion}. The proof of admissibility of the difference monomial ordering used in Example~\ref{ex:Example2} is postponed to Appendix~\ref{sec:GBtermorder}.

%
%
\section{Differential Thomas decomposition}\label{sec:differentialthomas}

Let $\differentialfield$ be the field of (complex) meromorphic functions on some connected open subset $\domain$ of $\C^n$
with coordinates $x_1$, \ldots, $x_n$.
A \emph{system of polynomial partial differential equations and inequations} (or \emph{differential system} for brevity)
for $m$ unknown functions $u^{(1)}$, $u^{(2)}$, \ldots, $u^{(m)}$ of $x_1$, \ldots, $x_n$ is given by
\begin{equation}\label{eq:diffsystem}
f_1 = 0\,, \quad f_2 = 0\,, \quad \ldots \quad f_p = 0\,, \quad
g_1 \neq 0\,, \quad g_2 \neq 0\,, \quad \ldots \quad g_q \neq 0\,,
\end{equation}
where $p \in \Z_{\ge 0}$, $q \in \Z_{\ge 0}$, and
all $f_i$ and $g_j$ are elements of the differential polynomial ring $\differentialfield\{ \mathbf{u} \}$,
endowed with the set $\setderiv = \{ \partial_1, \ldots, \partial_n \}$ of commuting derivations.
Most commonly, a solution of (\ref{eq:diffsystem}) is an $m$-tuple $(\phi_1, \ldots, \phi_m)$ of
locally analytic functions on $\domain$ which satisfy every equation and inequation of (\ref{eq:diffsystem}).
Around any point $\mathbf{z}$ of the domain each function $\phi_i$ has an expansion as convergent power series
\[
\sum_{\mathbf{k} \in (\Z_{\ge 0})^n} c_{\mathbf{k}} \, \frac{(\mathbf{x}-\mathbf{z})^{\mathbf{k}}}{\mathbf{k}!}, \qquad
(\mathbf{x} - \mathbf{z})^{\mathbf{k}} = (x_1 - z_1)^{k_1} \cdots (x_n - z_n)^{k_n}, \qquad
\mathbf{k}! = k_1! \cdots k_n!\,,
\]
with certain coefficients $c_{\mathbf{k}} \in \C$.

Given a differential system (\ref{eq:diffsystem})
the determination of (even just formal) power series solutions around some point $\mathbf{z}$
is in general a non-trivial task, because integrability conditions need to be taken into account
and the system of simultaneous algebraic equations and inequations for the coefficients $c_{\mathbf{k}}$
requires splitting into different cases due to nonlinearity.

\begin{example}[cf.\ also Ex.~2.1.46 in \cite{Robertz'14}]\label{ex:differential}
For simplicity we choose $\mathbf{z} = (0, 0)$ for investigating
formal power series solutions of the overdetermined system of quasilinear PDE
\begin{equation}\label{eq:examplesystem}
\left\{ \begin{array}{rcccl}
f_1 & := & u_x - u^2 & = & 0\,,\\[1em]
f_2 & := & u_{y,y} - u^3 & = & 0\,,
\end{array} \right.
\qquad \quad \mbox{for } u \, = \, u(x ,y)\,,
\qquad \domain \, = \, \C^2\,.
\end{equation}
Each of the two differential equations by itself is equivalent to
\[
\left(
\sum_{(k_1, k_2) \in (\Z_{\ge 0})^2} c_{(k_1+1, k_2)} \, \frac{x^{k_1} y^{k_2}}{k_1! \, k_2!}
\right) - \left(
\sum_{(k_1, k_2) \in (\Z_{\ge 0})^2} c_{(k_1, k_2)} \, \frac{x^{k_1} y^{k_2}}{k_1! \, k_2!}
\right)^2 \, \, \, = \, \, \, 0
\]
and
\[
\left(
\sum_{(k_1, k_2) \in (\Z_{\ge 0})^2} c_{(k_1, k_2+2)} \, \frac{x^{k_1} y^{k_2}}{k_1! \, k_2!}
\right) - \left(
\sum_{(k_1, k_2) \in (\Z_{\ge 0})^2} c_{(k_1, k_2)} \, \frac{x^{k_1} y^{k_2}}{k_1! \, k_2!}
\right)^3 \, \, \, = \, \, \, 0\,,
\]
respectively, and hence equivalent to a system of algebraic recurrence equations
\[
\left\{ \begin{array}{rcl}
c_{1,0} & = & c_{0,0}^2\,,\\[0.2em]
c_{1,1} & = & 2 \, c_{0,0} \, c_{0,1}\,,\\[0.2em]
c_{2,0} & = & 2 \, c_{0,0} \, c_{1,0}\,,\\[0.2em]
c_{2,1} & = & 2 \, c_{0,0} \, c_{1,1} + 2 \, c_{0,1} \, c_{1,0}\,,\\[0.2em]
        & \vdots &
\end{array} \right.
\qquad
\left\{ \begin{array}{rcl}
c_{0,2} & = & c_{0,0}^3\,,\\[0.2em]
c_{1,2} & = & 3 \, c_{0,0}^2 \, c_{1,0}\,,\\[0.2em]
c_{0,3} & = & 3 \, c_{0,0}^2 \, c_{0,1}\,,\\[0.2em]
c_{1,3} & = & 3 \, c_{0,0}^2 \, c_{1,1} + 6 \, c_{0,0} \, c_{0,1} \, c_{1,0}\,,\\[0.2em]
        & \vdots &
\end{array} \right.
\]
respectively. However, the integrability condition
\[
\begin{array}{rcl}
0 \, \, = \, \, \partial_y^2 f_1 - \partial_x f_2 & \, = \, & 3 \, u^2 \, u_x - 2 \, u \, u_{y,y} - 2 \, u_y^2
\qquad \qquad (\mbox{mod } f_1 = 0, f_2 = 0)\\[0.5em]
& \, = \, & 3 \, u^4 - 2 \, u^4 - 2 \, u_y^2 \, \, = \, \, (u^2 - \sqrt{2} \, u_y) \, (u^2 + \sqrt{2} \, u_y)
\end{array}
\]
reveals that further conditions on $c_{(k_1, k_2)}$ are implied when
$f_1 = 0$, $f_2 = 0$ is considered as a system.
Taking the above factorization into account,
we obtain
\[
\left\{ \begin{array}{rcl}
c_{0,1} & = & c_{0,0}^2 / \sqrt{2}\,,\\[0.2em]
c_{1,1} & = & \sqrt{2} \, c_{0,0} \, c_{1,0}\,,\\[0.2em]
c_{0,2} & = & \sqrt{2} \, c_{0,0} \, c_{0,1}\,,\\[0.2em]

c_{1,2} & = & \sqrt{2} \, (c_{0,0} \, c_{1,1} + c_{0,1} \, c_{1,0})\,,\\[0.2em]
        & \vdots &
\end{array} \right.
\quad \vee \quad
\left\{ \begin{array}{rcl}
c_{0,1} & = & -c_{0,0}^2 / \sqrt{2}\,,\\[0.2em]
c_{1,1} & = & -\sqrt{2} \, c_{0,0} \, c_{1,0}\,,\\[0.2em]
c_{0,2} & = & -\sqrt{2} \, c_{0,0} \, c_{0,1}\,,\\[0.2em]
c_{1,2} & = & -\sqrt{2} \, (c_{0,0} \, c_{1,1} + c_{0,1} \, c_{1,0})\,,\\[0.2em]
        & \vdots &
\end{array} \right.
\]
The method of Thomas decomposition does not require polynomial factorization.
If the above factorization is ignored, the newly-discovered consequence $2 \, u_y^2 - u^4 = 0$
translates into algebraic conditions on the Taylor coefficients $c_{(k_1,k_2)}$ as follows:
\begin{equation}\label{eq:algebraicsystem}
\left\{ \begin{array}{rcl}
2 \, \underline{c_{0,1}}^2 - c_{0,0}^4 & = & 0\,,\\[0.2em]
4 \, c_{0,1} \, \underline{c_{1,1}} - 4 \, c_{0,0}^3 \, c_{1,0} & = & 0\,,\\[0.2em]
8 \, c_{0,1} \, \underline{c_{0,2}} - 4 \, c_{0,0}^3 \, c_{0,1} & = & 0\,,\\[0.2em]
2 \, (c_{0,1} \, \underline{c_{2,1}} + c_{1,1}^2) - 2 \, c_{0,0}^2 \, (c_{0,0} \, c_{2,0} + 3 \, c_{1,0}^2) & = & 0\,,\\[0.2em]
& \vdots &
\end{array} \right.
\end{equation}
In this example the process of finding integrability conditions is complete because further
cross-derivatives reduce to zero modulo the previous equations. The given system does not impose
any conditions on the Taylor coefficient $c_{0,0}$, whose value can be chosen arbitrarily, and the
possible values of all other Taylor coefficients are determined by the above algebraic equations.
Taking the total order of differentiation into account, a systematic way of solving these algebraic equations is
to solve each equation for the underlined variable. In order to ensure both square-freeness of the
first polynomial equation for $c_{0,1}$ in (\ref{eq:algebraicsystem}) and solvability of all
subsequent equations, a case distinction whether $u(x, y)$ is the zero function or not is also made.
Therefore, a Thomas decomposition of system (\ref{eq:examplesystem}) is
\[
\left\{ \begin{array}{rcl}
u_x - u^2 & = & 0\,,\\[0.5em]
2 \, u_y^2 - u^4 & = & 0\,,\\[0.5em]
u & \neq & 0\,,
\end{array} \right.
\qquad \vee \qquad
\left\{ \begin{array}{rcl}
& \phantom{=} & \\[0.1em]
u & = & 0\,.\\[0.1em]
& \phantom{=} &
\end{array} \right.
\]
If the additional effort in factorizing the integrability condition $2 \, u_y^2 - u^4 = 0$ is spent, a Thomas
decomposition of the same system is also given by
\[
\left\{ \begin{array}{rcl}
u_x - u^2 & = & 0\,,\\[1em]
u_y - u^2 / \sqrt{2} & = & 0\,,
\end{array} \right.
\qquad \vee \qquad
\left\{ \begin{array}{rcl}
u_x - u^2 & = & 0\,,\\[1em]
u_y + u^2 / \sqrt{2} & = & 0\,.
\end{array} \right.
\]
Note that even when no polynomial factorization is performed, a Thomas decomposition
of a PDE system is not uniquely determined in general.
\end{example}

Computing a Thomas decomposition of a differential system is a finite process
which constructs a generating set of all integrability conditions systematically
and performs case splittings, if necessary, so as to obtain a generating set of
recurrence relations for $c_{\mathbf{k}}$ around a generic center of expansion.
This process is steered by a total order $\differentialrankeq$ on the set of symbols
representing derivatives of unknown functions:
\begin{equation}\label{MonDelta}
\Mon(\setderiv) \mathbf{u} \, \, := \, \,
\{ \, \partial^{\mathbf{k}} u^{(\alpha)} = \partial_1^{k_1} \cdots \partial_n^{k_n} u^{(\alpha)} \mid
1 \le \alpha \le m, \, \mathbf{k} \in (\Z_{\ge 0})^n \, \}\,.
\end{equation}
(We shall mainly be working with the strict total order $\differentialranking$ associated with $\differentialrankeq$.)

\begin{definition}\label{DifferentialRanking}
Let $\differentialring = \differentialfield\{ \mathbf{u} \}$ be the differential polynomial ring
and $f \in \differentialring \setminus \differentialfield$.
\begin{enumerate}
\item A \emph{ranking} $\differentialranking$ on 
$\differentialring$ is a total order on $\Mon(\setderiv) \mathbf{u}$
such that for all $1 \le \alpha \le m$ and all $\mathbf{k} \neq \mathbf{0}$
we have $\partial^{\mathbf{k}} u^{(\alpha)} \differentialranking u^{(\alpha)}$,
and such that $\partial^{\mathbf{k}_1} u^{(\alpha)} \differentialranking \partial^{\mathbf{k}_2} u^{(\alpha)}$
implies $\partial^{\mathbf{k}_1 + \mathbf{k}'} u^{(\alpha)} \differentialranking \partial^{\mathbf{k}_2 + \mathbf{k}'} u^{(\alpha)}$
for all $\mathbf{k}' \in (\Z_{\ge 0})^n$.
A ranking $\differentialranking$ is said to be \emph{orderly} if $|\mathbf{k}_1| > |\mathbf{k}_2|$ implies
$\partial^{\mathbf{k}_1} u^{(\alpha)} \differentialranking \partial^{\mathbf{k}_2} u^{(\alpha)}$ for any $\alpha$.
\item The \emph{leader} $\ld(f)$ of the differential polynomial $f$
with respect to a ranking $\differentialranking$ is the highest ranked derivative in $\Mon(\setderiv) \mathbf{u}$
that effectively occurs in $f$.
\item The coefficient of the highest power of $\ld(f)$ in $f$ is the \emph{initial} of $f$,
denoted by $\init(f)$. It is itself a differential polynomial in derivatives that are
ranked lower than $\ld(f)$ with respect to $\differentialranking$.
\item The \emph{discriminant} $\disc(f)$ is the discriminant of $f$ as a polynomial in $\ld(f)$.
\item The \emph{separant} of $f$ is the differential polynomial $\sep(f) := \partial f / \partial \ld(f)$.
\end{enumerate}
\end{definition}

\begin{example}
If $x$, $y$ are the independent variables, $u$ the dependent variable,
and $f = u_y \, u_{x,y}^2 + u^5 \in \differentialfield\{ u \}$,
then, with respect to any orderly ranking $\differentialranking$ on $\differentialfield\{ u \}$,
we have $\ld(f) = u_{x,y}$ and $\init(f) = u_y$ and $\sep(f) = 2 \, u_y \, u_{x,y}$.
Note that, generally, the separant of $f$ is the initial of any proper derivative of $f$,
e.g., $\partial_x f = 2 \, u_y \, u_{x,y} \, u_{x,x,y} + u_{x,y}^3 + 5 \, u^4 \, u_x$
has leader $u_{x,x,y}$ and initial $2 \, u_y \, u_{x,y}$.
\end{example}

The determination of all integrability conditions of a system of
polynomially nonlinear PDE is facilitated by a combination of Euclid's algorithm
with case distinctions and completion to involution as performed
by Janet's algorithm. Before recalling the latter ingredient we outline
the former aspect.
In what follows we assume that a ranking $\differentialranking$ on $\differentialring$ is fixed.

Note that any linear combination with coefficients in $\differentialring$ of
(the left hand sides of)  equations $f_1 = 0$, \ldots, $f_p = 0$ and their partial derivatives in a
differential system is a consequence of that system, and these consequences
form a differential ideal of $\differentialring$.
Every differential polynomial $f \in \differentialring \setminus \differentialfield$ is considered as
a univariate polynomial in $\ld(f)$ whose coefficients are themselves
univariate polynomials in their leaders. In this way an \emph{algebraic} and a
\emph{differential reduction} are defined for all pairs $(f_1, f_2) \in (\differentialring \setminus \differentialfield)^2$,
producing a differential polynomial $f_3$ that is either in $\differentialfield$ or has a leader
that is ranked lower than $\ld(f_2)$ with respect to $\differentialranking$.
\begin{enumerate}
\renewcommand{\theenumi}{\alph{enumi})}
\renewcommand{\labelenumi}{\theenumi}
\item\label{algebraic_pseudoreduction}
If $\ld(f_1) = \ld(f_2) =: v$ and $d_1 := \deg_{v}(f_1) \ge d_2 := \deg_{v}(f_2)$, then let
\[
f_3 \, \, = \, \, c_1 f_1 - c_2 \, v^{d_1-d_2} \, f_2\,,
\]
where $c_1$ is a suitable power of $\init(f_2)$ and $c_2 \in \differentialring$ such that the $d_1$-th power of $v$ cancels in $f_3$.
\item\label{differential_pseudoreduction}
If $\ld(f_1) = \partial^{\mathbf{k}} \ld(f_2) =: v$ for $\partial^{\mathbf{k}} \in \Mon(\setderiv) \, \mathbf{u}$,
$\mathbf{k} \neq \mathbf{0}$, and $d := \deg_{v}(f_1)$, then let
\[
f_3 \, \, = \, \, c_1 \, f_1 - c_2 \, v^{d-1} \, \partial^{\mathbf{k}} f_2\,,
\]
where $c_1 = \sep(f_2)$ and $c_2 \in \differentialring$ such that the $d$-th power of $v$ cancels in $f_3$.
\end{enumerate}
Note that $f_3$ is an element of the differential ideal containing $f_1$ and $f_2$ in any case.

If $f_1 = 0$, $f_2 = 0$ are two equations in a differential system, then replacing $f_1 = 0$ by $f_3 = 0$
is supposed to not alter the solution set of the system. This is ensured if the differential polynomial $c_1$
does not vanish on the solution set of the system. Note that $c_1$ is chosen as a power of $\init(f_2)$ or
$\sep(f_2)$. If Euclid's algorithm considers separately the cases obtained by adding the inequation $\init(f_2) \neq 0$
(resp.\ $\sep(f_2) \neq 0$) or the equation $\init(f_2) = 0$ (resp.\ $\sep(f_2) = 0$) to the system,
the above replacement of $f_1 = 0$ by $f_3 = 0$ is justified with the imposed inequation,
and the solution sets corresponding to the
branches of computation define a partition of the solution set of the original system.

Ignoring that the indeterminates represent unknown functions of a PDE system, Euclid's algorithm
deals with a system $S$ of algebraic equations, say in, $\algvar_1$, \ldots, $\algvar_r$, totally
ordered by the fixed ranking. The solution set $\Sol(S)$ in $\C^r$ of that algebraic system is investigated
with respect to a sequence of projections from $\C^r$ to affine subspaces which corresponds to the
ordering, say, $\algvar_1 \differentialranking \algvar_2 \differentialranking \ldots \differentialranking \algvar_r$,
of the indeterminates:
\[
\begin{array}{rclrcr}
\pi_1\colon \C^r & \longrightarrow & \C^{r-1}\colon &
(a_1, a_2, \ldots, a_r) & \longmapsto & (a_2, a_3, a_4, \ldots, a_r)\,,\\[0.5em]
\pi_2\colon \C^r & \longrightarrow & \C^{r-2}\colon &
(a_1, a_2, \ldots, a_r) & \longmapsto & (a_3, a_4, \ldots, a_r)\,,\\[0.5em]
& \vdots & & & \vdots \\[0.5em]
\pi_{r-1}\colon \C^r & \longrightarrow & \C\colon &
(a_1, a_2, \ldots, a_r) & \longmapsto & a_r\,.
\end{array}
\]
Euclid's algorithm, performing case distinctions with regard to the vanishing of initials $\init(f)$
and discriminants $\disc(f)$
of (non-constant) polynomials $f$, produces a finite collection of algebraic systems having the following property.

\begin{definition}\label{de:algebraicsimple}
Let $S = \{ \, f_1 = 0, \, \ldots, \, f_p = 0, \, g_1 \neq 0, \, \ldots, \, g_q \neq 0 \, \}$ be
an algebraic system, i.e., $f_i$, $g_j \in \differentialfield[\algvar_1, \ldots, \algvar_r]$.
Then $S$ is said to be \emph{simple} if the following four conditions are satisfied.
\begin{enumerate}
\item None of $f_1$, \ldots, $f_p$, $g_1$, \ldots, $g_q$ is constant.\label{de:algebraicsimple_1}
\item The leaders of $f_1$, \ldots, $f_p$, $g_1$, \ldots, $g_q$ are pairwise distinct.\label{de:algebraicsimple_2}
\item For every $h \in \{ f_1, \ldots, f_p, g_1, \ldots, g_q \}$,
if $\ld(h) = \algvar_k$,
then the equation $\init(h) = 0$
has no solution in $\pi_k(\Sol(S))$.\label{de:algebraicsimple_3}
\item For every $h \in \{ f_1, \ldots, f_p, g_1, \ldots, g_q \}$,
if $\ld(h) = \algvar_k$,
then the equation $\disc(h) = 0$
has no solution in $\pi_k(\Sol(S))$.\label{de:algebraicsimple_4}
\end{enumerate}
(Note that in \ref{de:algebraicsimple_3}.\ and \ref{de:algebraicsimple_4}.\
we have $\init(h), \disc(h) \in \differentialfield[\algvar_{k+1}, \ldots, \algvar_r]$.)
\end{definition}

\begin{definition}\label{de:algebraicquasisimple}
An algebraic system $S$ as in Definition~\ref{de:algebraicsimple} is
said to be \emph{quasi-simple} if conditions~\ref{de:algebraicsimple_1}.--\ref{de:algebraicsimple_3}.\
(but not necessarily \ref{de:algebraicsimple_4}.) are satisfied.
(Such systems are also called \emph{regular}, cf.\ \cite[p.~107]{Wang'01}, \cite{Hubert'01}, \cite{Kalkbrener}, \cite{LemaireMorenoMazaPanXie}.)
\end{definition}

Our strategy for handling integrability conditions builds on \emph{Janet division}.
Note first that the leader of the derivative of an equation $f = 0$ is the corresponding derivative of $\ld(f)$.
Hence, for each $\alpha \in \{ 1, \ldots, m \}$, the monomials $\partial^{\mathbf{k}}$, $\mathbf{k} \in (\Z_{\ge 0})^n$,
for which $\partial^{\mathbf{k}} u^{(\alpha)}$ is the leader of a consequence
of a differential system, form a set of monomials that is closed under multiplication by $\partial_1$, \ldots, $\partial_n$.

Suppose a set of monomials is closed under multiplication by the elements of a certain subset $\mu$
of $\setderiv = \{ \partial_1, \ldots, \partial_n \}$. If that set of monomials
consists of all such multiples of a single monomial, then we call the set a \emph{cone}.
Let $M$ be a finite set of monomials.
Janet division assigns to each $m \in M$ a set $\mu(m, M) \subseteq \setderiv$ 
of \emph{multiplicative variables} so as to decompose the set of all multiples of $M$ into disjoint cones.
Denoting by $\Mon(\mu)$ the set of all monomials in the elements of $\mu$, we have
\[
\bigcup_{m \in M} \Mon(\setderiv) \, m \, \, \supseteq \, \, \biguplus_{m \in M} \Mon(\mu(m, M)) \, m\,.
\]
In case of equality the set $M$ is said to be \emph{Janet complete}.

In our context we call the multiplicative variables \emph{admissible derivations}.

\begin{definition}\label{de:admissiblederivations}
Let $M$ be a finite set of monomials in $\partial_1$, \ldots, $\partial_n$.
For $1 \le j \le n$ we let $\partial_j$ be an \emph{admissible derivation} for
$\partial_1^{i_1} \cdots \partial_n^{i_n} \in M$ if and only if
\[
i_j \, \, = \, \, \max \{ \, k_j \mid \partial_1^{k_1} \cdots \partial_n^{k_n} \in M \mbox{ with }
k_1 = i_1, \, k_2 = i_2, \, \ldots, \, k_{j-1} = i_{j-1} \, \}\,.
\]
\end{definition}

\begin{example}
Let $M = \{ \, \partial_1^2 \partial_2, \, \partial_1^2 \partial_3, \, \partial_2^2 \partial_3, \, \partial_2 \partial_3^2 \, \}$.
These four monomials are assigned the sets of
admissible derivations $\{ \partial_1, \partial_2, \partial_3 \}$,
$\{ \partial_1, \partial_3 \}$, $\{ \partial_2, \partial_3 \}$ and $\{ \partial_3 \}$, respectively.
\end{example}

We extend Janet division as well as the notion of Janet completeness from finite sets of monomials 
to finite sets $\{ f_1, \ldots, f_p \}$ of differential polynomials in $\differentialring \setminus \differentialfield$
by assigning $f_i$ the set of admissible derivations $\mu_i := \mu(\theta_i, \{ \theta_1, \ldots, \theta_p \})$,
where $\theta_i \in \Mon(\{ \partial_1, \ldots, \partial_n \})$ is such that $\ld(f_i) = \theta_i \, u^{(\alpha_i)}$
for a certain $\alpha_i$.

By restricting the differential reduction process introduced in \ref{differential_pseudoreduction} above
to reduction steps for which $\partial^{\mathbf{k}}$ is a monomial in \emph{admissible} derivations for $f_2$,
we obtain the Janet reduction process. The remainder of a differential polynomial $f$ modulo $\{ \, f_1, \ldots, f_p \, \}$,
or modulo $T = \{ \, (f_1, \mu_1), \ldots, (f_p, \mu_p) \, \}$, is called the \emph{Janet normal form} of $f$
modulo $T$ and is denoted by $\NF(f, T, \ranking)$.

\begin{definition}\label{de:differentialpassive}
Let $T = \{ \, (f_1, \mu_1), \ldots, (f_p, \mu_p) \, \}$ be Janet complete.
Then the differential system $\{ \, f_1 = 0, \ldots, f_p = 0 \, \}$, or $T$, is said to be \emph{passive} if
\[
\NF(\partial f_i, T, \ranking) \, = \, 0 \qquad \mbox{for all} \quad
\partial \in \overline{\mu}_i = \{ \partial_1, \ldots, \partial_n \} \setminus \mu_i\,, \quad i = 1, \ldots, p\,.
\]
\end{definition}

A suitable combination of Euclid's algorithm with case distinctions and differential reductions
of differential polynomials that are obtained by applying non-ad\-mis\-si\-ble derivations defines
a process that returns a Thomas decomposition in finitely many steps \cite[Thm.~2.2.57]{Robertz'14}, \cite[Sect.~3.4]{BGLHR'12},
namely, a finite collection of
differential systems, whose solution sets partition the solution set of the original
differential system, and such that each output system has the following property.

\begin{definition}\label{de:differentialsimple}
A differential system $S$ as in (\ref{eq:diffsystem})
is said to be \emph{simple} if the following three conditions hold.
\begin{enumerate}
\item $S$ is simple as an algebraic system
(in the finitely many indeterminates occurring in it, ordered by
the ranking $\differentialranking$; cf.\ Definition~\ref{de:algebraicsimple}).\label{de:differentialsimple_a}
\item $\{ \, f_1 = 0, \, \ldots, \, f_p = 0 \, \}$ is passive (cf.\ Definition~\ref{de:differentialpassive}).\label{de:differentialsimple_b}
\item The left hand sides $g_1$, \ldots, $g_q$
are Janet reduced (i.e., in Janet normal form)
modulo the equations $\{ \, f_1 = 0, \, \ldots, \, f_p = 0 \, \}$.\label{de:differentialsimple_c}
\end{enumerate}
\end{definition}

A simple differential system $S$ allows to decide, by differential reduction,
whether or not a given equation $f = 0$,
where $f \in \differentialring$, is a consequence of $S$.

\begin{proposition}[\cite{Robertz'14}, Prop.\ 2.2.50]\label{prop:differentialmembership}
Let $S$ be a simple differential system, defined over the differential polynomial ring $\differentialring$,
and let $E$ be the differential ideal of $\differentialring$ which is generated by $f_1$, \ldots, $f_p$. Moreover,
let $q$ be the product of the initials and separants of all $f_1$, \ldots, $f_p$.
Then the differential ideal
\[
E : q^{\infty} \, \, := \, \, \{ \, f \in \differentialring \mid q^r \, f \in E \mbox{ for some } r \in \Z_{\ge 0} \, \}
\]
is radical. Given $f \in \differentialring$, we have $f \in E : q^{\infty}$ if and only if
the Janet normal form of $f$ modulo $\{ \, f_1, \ldots, f_p \, \}$ is zero.
\end{proposition}

%
%
\section{Difference approximations to PDE systems}\label{sec:differenceapproximations}

To approximate the differential system \eqref{pde} by a difference system we
shall consider a regular computational grid (mesh) as the set of
points
\begin{equation}\label{grid}
\{\,(z_1 + k_1h_1,\ldots, z_n + k_nh_n)\mid k_i\in \Z \,\}\,,
\end{equation}
where $(z_1, \ldots, z_n) \in \R^n$ and $0 < h_i\in \R$ are fixed.
\begin{definition}\label{GridFunction}
A vector function $\tilde{\mathbf{u}}=\{\tilde{u}^{(1)},\ldots,\tilde{u}^{(m)}\}$ which assigns to each grid node a value is called {\em grid vector function}. We shall denote 
such function by
\begin{equation*}
\mathbf{\tilde{u}}_{k_1,\ldots,k_n} \, \,:= \, \, \mathbf{\tilde{u}}(z_1 + k_1h_1,\ldots, z_n + k_nh_n)\,.
\end{equation*}
\end{definition}

From now on we shall consider $h_1$, \ldots, $h_n$ as parameters and denote by
\begin{equation*}
\mathbf{h}:=\{h_1,\ldots,h_n\}\ \ \text{and}\ \ \Mon(\mathbf{h}):= \Big\{\prod_{i=1}^n h_i^{\mu_i} \, \Big| \, \mu\in (\Z_{\geq 0})^n\Big\}
\end{equation*}
the set of {\em mesh steps (grid spacings)} and  the {\em monoid of monomials generated by the elements in $\mathbf{h}$}, respectively. The total degree of an element $m\in \Mon(\mathbf{h})$ will be denoted by $\deg(m)$.

We assume that coefficients of the differential polynomials in $F$ (cf.\ (\ref{pde}))
do not vanish in the grid points. The coefficients on the grid as rational functions in $\{\mathbf{a},\mathbf{h}\}$ are elements of the difference field~\cite{Levin'08} with \emph{differences} $\{\automorphism_1,\ldots,\automorphism_n,\automorphism_1^{-1},\ldots,\automorphism_n^{-1}\}$ acting on a grid function $\tilde{u}^{(\alpha)}_{k_1,\ldots,k_n}$ as the shift operators
\begin{equation}
\automorphism_j^{\pm 1}
\tilde{u}^{(\alpha)}_{k_1,\ldots,k_j,\ldots,k_n}=\tilde{u}^{(\alpha)}_{k_1,\ldots,k_j\pm 1,\ldots,k_n}\,,\quad
\alpha \in \{1,\ldots,m\}\,,\quad j\in \{1,\ldots,n\}\,.
\label{rs-operators}
\end{equation}
Let $\Mon(\setauto)$ be the free commutative semigroup generated by $\setauto = \{\automorphism_1, \ldots, \automorphism_n\}$,
\begin{equation}\label{MonoidOfShifts}
\Mon(\setauto) \, \, := \, \, \{\,\automorphism_1^{i_1}\cdots\automorphism_n^{i_n}\mid i_1,\ldots,i_n\in \Z_{\geq 0}\,\}\,,
\end{equation}
$\differencefield = \Q(\mathbf{a},\mathbf{h})$, and $\differencering$  the ring of {\em difference polynomials} over $\differencefield$.
The elements in $\differencering$ are polynomials in the \emph{difference indeterminates} $\tilde{u}^{(\alpha)}$ ($\alpha=1,\ldots,m$) and in their shifts $\automorphism_1^{i_1}\cdots \automorphism_n^{i_n} \tilde{u}^{(\alpha)}$, $i_1, \ldots, i_n \in \Z$,
with coefficients in $\differencefield$.

\begin{remark}\label{TaylorExpansionElementsSigma}
Since the shift operators $\sigma_j$ admit the (formal) power series expansion
\[
\sigma_j \, \, = \, \, \sum_{k\geq 0}\frac{h_j^k}{k!}\partial^k_j\,,\qquad
\sigma_j^{-1} \, \, = \, \, \sum_{k\geq 0}\frac{(-1)^kh_j^k}{k!}\partial^k_j\,,
\quad
\]
a difference polynomial $f(\tilde{\mathbf{u}})\in \differencering$ admits the Taylor expansion around a grid point~\eqref{grid}.
\end{remark}

The standard technique to obtain FDA to the PDE system \eqref{pde} is to replace the derivatives occurring in \eqref{pde} by finite differences. In order to use the method of difference \Gr bases (Section~\ref{sec:differencestandardbases}) or/and difference Thomas decomposition (Section~\ref{DifferenceThomasDecomposition}) one has to apply appropriate power products of the right-shift operators (\ref{rs-operators}) to remove negative
shifts in indices which may be introduced from expressions like
\begin{eqnarray*}
&& \partial_j u^{(\alpha)} \, \, = \, \, \frac{\tilde{u}^{(\alpha)}_{k_1,\ldots,k_j,\ldots,k_n}-\automorphism^{-1}_j \tilde{u}^{(\alpha)}_{k_1,\ldots,k_{j},
\ldots,k_n}}{h_j}+\mathcal{O}(h_j)\,, \label{BackwardlDifference} \\
&& \partial_j u^{(\alpha)} \, \, = \, \, \frac{\sigma_j \tilde{u}^{(\alpha)}_{k_1,\ldots,k_j,\ldots,k_n}-\automorphism^{-1}_j \tilde{u}^{(\alpha)}_{k_1,\ldots,k_j,
\ldots,k_n}}{2h_j}+\mathcal{O}(h_j^2)\,. \label{CentralDifference}
\end{eqnarray*}
In the sequel we shall consider discretization of \eqref{pde} as a finite set of difference polynomials
\begin{equation}
\tilde{f}_1=\cdots=\tilde{f}_p=0\,,\qquad \tilde{F}:=\{\tilde{f}_1,\ldots,\tilde{f}_p\}\subset \differencering\,. \label{fda}
\end{equation}

\begin{definition}\label{DiffIdeal}
The \emph{differential (resp. difference) ideal generated by a polynomial set $F$ (resp. $\tilde{F}$)}, denoted by $\differentialideal := [F]$ (resp.\ $\differenceideal := [\tilde{F}]$), is the smallest subset of $\differentialring$ (resp.\ $\differencering$) containing $F$ (resp.\ $\tilde{F}$) and satisfying
\[
(\,\forall \partial_i\in \{\partial_1,\ldots,\partial_n\})\ (\,\forall a,b\in {\cI}\,)\ (\,\forall c\in \differentialring\,)\quad [\,a+b\in \differentialideal\,,\ a\cdot c\in \differentialideal,\ \partial_i  a\in \differentialideal\,]
\]
and, respectively,
\[
(\,\forall \automorphism_i\in \{\automorphism_1,\ldots,\automorphism_n\})\ (\,\forall \tilde{a},\tilde{b}\in \differenceideal\,)\ (\,\forall \tilde{c}\in \differencering\,)\quad [\,\tilde{a}+\tilde{b}\in \differenceideal\,,\ \tilde{a}\cdot \tilde{c}\in \differenceideal,\ \automorphism_i \tilde{a}\in \tilde{\cI}\,]\,.
\]
\end{definition}

Let $\differentialideal \subset \differentialring$ be a differential ideal. Then the set
\begin{equation}
\sqrt{\differentialideal} \, \, := \, \, \{\,p\in \differentialring \mid p^k \in \differentialideal,\ k\in {\N}_{>0}\, \} \label{radical}
\end{equation}
is a differential ideal.

If $\differentialideal = \sqrt{\differentialideal}$, then $\differentialideal$ is called \emph{radical} or \emph{perfect differential ideal}. Given $F \subset \differentialring$, the \emph{radical differential ideal generated by ${F}$}, denoted by $\llbracket {F}\rrbracket$, is the smallest radical differential ideal of $\differentialring$ containing ${F}$.

In the difference case, the radical $\sqrt{\differenceideal}$ of $\differenceideal$ is defined similarly to Eq.\,\eqref{radical}. However, the notion of perfect difference ideal is significantly distinct from that of perfect differential ideal in differential algebra~\cite{Ritt'50}.

\begin{definition}\label{PerfectIdeal}
The \emph{perfect difference ideal~\cite{Levin'08} generated by a set $\tilde{F} \subset \differencering$}, denoted by $\llbracket \tilde{F}\rrbracket$, is the smallest difference ideal of $\differencering$ containing $\tilde{F}$ and such that for any $\tilde{f}\in \differencering$, $\theta_1,\ldots,\theta_r\in \Mon(\setauto)$ and $k_1,\ldots,k_r \in \Z_{\ge 0}$ we have
\begin{equation*}\label{shuffling}
(\theta_1 \tilde{f})^{k_1}\cdots (\theta_r \tilde{f})^{k_r}\in \llbracket \tilde{F}\rrbracket \quad \Longrightarrow \quad \tilde{f}\in \llbracket \tilde{F}\rrbracket \,.
\end{equation*}
\end{definition}

It is clear that $[\tilde{F}]\subseteq \sqrt{\tilde{F}} \subseteq \llbracket \tilde{F}\rrbracket$. In difference algebra perfect ideals are analogues of radical ideals in commutative~\cite{CLO'07} and differential algebra~\cite{Ritt'50,Hubert'01}. In particular, the difference Hilbert’s Nullstellensatz is formulated in terms of perfect difference ideals (cf.\,\cite{{CohnDifferenceAlgebra}}, Ch.~4, Thm.~4 and\,\cite{Levin'08},\,Thm.~2.6.5). For this reason we give the following definition.

\begin{definition}\label{DiffCons}
We shall say that a differential (resp.\ difference) polynomial $f\in \differentialring$ (resp. $\tilde{f}\in \differencering$) is
a \emph{differential-algebraic} (resp.\ \emph{difference-algebraic}) \emph{consequence} of \eqref{pde} (resp.\ of \eqref{fda})
if $f$ (resp.\ $\tilde{f}$) is an element of the perfect differential (resp.\ difference) ideal
generated by \eqref{pde} (resp.\ \eqref{fda}).
\end{definition}

Some recent results on the relation between the difference Hilbert’s Nullstellensatz and solvability are presented in~\cite{OPS'19,PSW'19}.

%
%
\section{Difference \Gr Bases}\label{sec:differencestandardbases}

The notion of difference \Gr basis was introduced and studied in~\cite{G'12,LaScala'15,GLS'15}. It is a difference analogue of the notion of differential standard basis introduced in~\cite{Ollivier'90}, where a finite standard basis is called {\em \Gr basis}. In this paper we prefer to use the approach to difference \Gr bases suggested in~\cite{G'12}.

\begin{definition}\label{DifferenceRanking}
A \emph{ranking} on $\differencering$ is defined in the same way
as in Definition~\ref{DifferentialRanking} by replacing the action
of $\partial_i$ by the action of $\automorphism_i$
and $\setderiv$ by $\setauto$.
\end{definition}

\begin{definition}\label{DifferenceMonomialOrdering}
A total ordering $\differencemonomialordering$ on the set of \emph{difference monomials}
\[
{\cM} \, \, := \, \, \left\{ \,(\theta_1 \tilde{u}^{(1)})^{i_1}\cdots (\theta_m \tilde{u}^{(m)})^{i_m} \, \Big| \, \theta_j\in \setauto,\ i_j\in \Z_{\geq 0},\ 1\leq j\leq m \, \right\}
\]
is an \emph{admissible (difference) monomial ordering} if it extends a ranking and satisfies
\begin{eqnarray*}
&(a) & (\forall\, \tilde{t}\in {\cM}\setminus \{1\})\ [\tilde{t} \differencemonomialordering 1]\,,\\
&(b) & (\,\forall\, \theta\in \setauto)\ (\,\forall\, \tilde{t}, \tilde{v},\tilde{w}\in {\cM})\ [\ \tilde{v} \differencemonomialordering \tilde{w} \Longleftrightarrow \tilde{t}\cdot \theta\circ \tilde{v} \differencemonomialordering  \tilde{t}\cdot \theta\circ \tilde{w}\,]\,.
\end{eqnarray*}
\end{definition}

For examples of admissible monomial orderings we refer to Appendix~\ref{sec:GBtermorder}.

\medskip

Given an admissible ordering $\differencemonomialordering$, every difference polynomial $\tilde{f}$ has the \emph{leading monomial} $\lm(\tilde{f})\in {\cM}$ with \emph{leading coefficient} $\lc(\tilde{f})$. In what follows every difference polynomial
is to be \emph{normalized (i.e., monic)} by division by its leading coefficient.

\begin{definition}\label{divisibility}
If for $v,w\in {\cM}$ the equality $w=t\cdot \theta\circ v$ holds with $\theta\in \Mon(\setauto)$ and $t\in {\cM}$ we shall say that $v$ {\em divides} $w$ and write $v\mid w$. It is easy to see that this divisibility relation yields a partial order.
\end{definition}

\begin{definition}\label{def_SB}
Given a difference ideal $\differenceideal$ and an admissible monomial ordering $\differencemonomialordering$, a subset $\tilde{G}\subset \differenceideal$ is a \emph{(difference) \Gr basis} for $\differenceideal$ if $[\tilde{G}]=\differenceideal$ and
\begin{equation*}
(\,\forall\, \tilde{f}\in {\cI}\,) (\,\exists\, \tilde{g}\in \tilde{G}\,)\ \ [\,\lm(\tilde{g})\mid \lm(\tilde{f})\,]\,.
\label{SB}
\end{equation*}
\end{definition}

\begin{definition}
A polynomial $\tilde{p}\in \differencering$ is said to be \emph{head reducible modulo $\tilde{q}\in \differencering$ to $\tilde{r}$} if $\tilde{r}=\tilde{p}-m\cdot\theta\circ \tilde{q}$ and $m\in {\cM}$, $\theta\in \Mon(\setauto)$ are such that $\lm(\tilde{p})=m\cdot\theta\circ \lm(\tilde{q})$. In this case transformation from
$\tilde{p}$ to $\tilde{r}$ is an \emph{elementary reduction}, denoted by ${\tilde{p}}\xrightarrow[\tilde{q}]{} \tilde{r}$. Given a set $\tilde{F}\subset \differencering$, $\tilde{p}$ is \emph{head reducible modulo $\tilde{F}$} $($denotation: ${\tilde{p}}\xrightarrow[\tilde{F}]{})$  if there is $\tilde{f}\in \tilde{F}$ such that $\tilde{p}$ is head reducible modulo $\tilde{f}$.
A polynomial $\tilde{p}$ {\em is head reducible to $\tilde{r}$ modulo $\tilde{F}$} if there is a finite chain of elementary reductions
\begin{equation}
\tilde{p}\xrightarrow[\tilde{F}]{}\tilde{p}_1\xrightarrow[\tilde{F}]{} \tilde{p}_2\xrightarrow[\tilde{F}]{}\cdots \xrightarrow[\tilde{F}]{}\tilde{r}\,.
\label{red_chain}
\end{equation}
If no monomial in $\tilde{r}$ from \eqref{red_chain} is head reducible modulo $\tilde{F}$, then \emph{$\tilde{r}$ is in head normal form modulo $\tilde{F}$} and we write $\tilde{r}=\mathrm{HNF}(\tilde{p},\tilde{F}, \differencemonomialordering)$. Similarly, one can define \emph{tail reduction} and \emph{(full) normal form} (denotation: $\mathrm{NF}(\tilde{p},\tilde{F}, \differencemonomialordering)$ . A polynomial set $\tilde{F}$ with more than one element is \emph{(head) interreduced} if
\begin{equation}
(\,\forall \tilde{f}\in \tilde{F}\,)\ [\,\tilde{f}=\mathrm{(H)NF}(\tilde{f},\tilde{F}\setminus \{\tilde{f}\}, \differencemonomialordering)\,]\,. \label{interreduce}
\end{equation}
\label{reduction}
\end{definition}
Admissibility of $\differencemonomialordering$, as in commutative algebra,  provides termination of the chain~(\ref{red_chain}) for any $\tilde{p}$ and $\tilde{F}$. Then $\mathrm{(H)NF}(\tilde{p},\tilde{F},\sqsupset)$ can be computed by the difference version of a multivariate polynomial division algorithm~\cite{BW'93,CLO'07}. If $\tilde{G}$ is a \Gr basis of $[\tilde{G}]$, then from Definitions~\ref{def_SB} and \ref{reduction} it follows
\[
\tilde{f}\in [\tilde{G}] \quad \Longleftrightarrow \quad \mathrm{(H)NF}(\tilde{f},\tilde{G}, \differencemonomialordering)=0\,.
\]
Thus, if an ideal has a finite Gr\"{o}bner basis, then its construction solves the ideal membership problem in the same way as in commutative~\cite{BW'93,CLO'07} and differential~\cite{Ollivier'90,Zobnin'05} algebra. The algorithmic characterization of difference \Gr bases and their construction in difference polynomial rings employ difference $S$-polynomials.

\begin{definition}
Given an admissible ordering and monic difference polynomials $\tilde{p}$ and $\tilde{q}$, a polynomial $S(\tilde{p},\tilde{q}):=m_1\cdot \theta_1\circ \tilde{p}-m_2\cdot \theta_2\circ \tilde{q}$
is called \emph{$S$-polynomial} associated to $\tilde{p}$ and $\tilde{q}$ if
$ m_1\cdot \theta_1\circ  \lm(\tilde{p})=m_2\cdot \theta_2\circ \lm(\tilde{q})$
with co-prime $m_1\cdot \theta_1$ and $m_2\cdot \theta_2$
(for $\tilde{p}=\tilde{q}$ we shall say that the $S$-polynomial is associated with $\tilde{p}$).\label{S-polynomial}
\end{definition}

\begin{proposition}\label{B-criterion}
Given a difference ideal $\differenceideal \subset \differencering$ and an admissible monomial ordering $\differencemonomialordering$, a set of polynomials $\tilde{G}\subset \differenceideal$ is a Gr\"obner basis of $\differenceideal$ if and only if
\begin{equation}\label{PassivityConditions}
\mathrm{(H)NF}(S(\tilde{p},\tilde{q}),\tilde{G}, \differencemonomialordering) \, \, = \, \, 0
\end{equation}
for
all $S$-polynomials associated with polynomials in $\tilde{G}$.
\end{proposition}

\begin{proof}
This follows from Definitions \ref{def_SB}, \ref{reduction} and \ref{S-polynomial} in line with the standard proof of the analogous theorem for \Gr bases in commutative algebra~\cite{BW'93,CLO'07} and with the proof of a similar theorem for standard bases in differential algebra~\cite{Ollivier'90}.\qed
\end{proof}

\begin{definition}\label{df:PassivityConditions}
Given a system~\eqref{fda} of difference equations, the conditions~\eqref{PassivityConditions} with $\tilde{p},\tilde{q}\in \tilde{F}$ are said to be the {\em (Gr\"{o}bner) passivity conditions} for the system~\eqref{fda}.
\end{definition}

Let $\differenceideal = [\tilde{F}]$ be a difference ideal generated by a finite set $\tilde{F}\subset \differencering$ of difference polynomials with non-negative shifts. Then for a fixed admissible monomial ordering the algorithm~\algref{DifferenceGr\"obnerBasis} given below, if it terminates, returns a \Gr basis $\tilde{G}$ of $\differenceideal$.  The subalgorithm \algref{Interreduce} invoked in line~9 performs mutual (head) interreduction of the elements in $\tilde{H}$ and returns a set satisfying (\ref{interreduce}).

Algorithm~\algref{DifferenceGr\"obnerBasis} is a difference analogue of the simplest version of Buchberger's algorithm (cf.~\cite{BW'93,CLO'07,Ollivier'90}). Its correctness is provided by Theorem~\ref{B-criterion}. The algorithm always terminates when the input polynomials are linear. If this is not the case, the algorithm may not terminate.  This means that the {\bf repeat-until} loop (lines 2--8) may be infinite as in the differential case~\cite{Ollivier'90,Zobnin'05}. One can improve the algorithm by taking into account Buchberger's criteria to avoid some useless zero reductions in line~5. The difference criteria are similar to the differential ones~\cite{Ollivier'90}.

\begin{algorithm}\label{StandardBasis}
\DontPrintSemicolon
\KwInput{$\tilde{F}\subset \differencering \setminus \{0\}$, a finite set of non-zero polynomials; $\differencemonomialordering$, a monomial ordering}
\KwOutput{$\tilde{G}$, a (head) interreduced \Gr basis of $[\tilde{F}]$}
$\tilde{G} \gets \tilde{F}$\;
\Repeat{$\tilde{G} = \tilde{H}$}
{
  $\tilde{H} \gets \tilde{G}$\;
  \For{$S$-polynomials $\tilde{s}$ associated with elements in $\tilde{H}$}
  {
    $\tilde{g} \gets \mathrm{(H)NF}(\tilde{s},\tilde{H}, \differencemonomialordering)$\;
    \If{$\tilde{g}\neq 0$}
    {
      $\tilde{G} \gets \tilde{G}\cup \{\tilde{g}\}$\;
    }
  }
}
$\tilde{G} \gets \algref{Interreduce}(\tilde{G})$\;
\Return $\tilde{G}$
\caption{\algref{DifferenceGr\"obnerBasis}}
\end{algorithm}


%
%
\section{Consistency}\label{sec:consistency}

Let the PDE system~\eqref{pde} and its finite difference discretization~\eqref{fda} on the regular grid~\eqref{grid} be given. 

\begin{definition}\label{Implication}
We shall say that a {\em difference equation} $\tilde{f}(\mathbf{\tilde{u}})=0$, $\tilde{f}\in \differencering$, implies the set of differential equations
\begin{equation*}\label{ImplSet}
{\mathcal{F}} \, \, := \, \, \{\,{r}_k(\mathbf{u})=0\,\mid r_k\in \differentialring\,,\ k=1,\ldots ,K,\ K\in \N_{>0}\,\}
\end{equation*}
and we write
$\tilde{f}\rhd {\mathcal{F}}$ for $K>1$ or  $\tilde{f}\rhd r_1$ for $K=1$, if
Taylor expansion of $\tilde{f}$ about a grid point,
after clearing denominators containing the elements in $\mathbf{h}=(h_1,\ldots,h_n)$
by multiplying by an appropriate $q(\mathbf{h})$, yields
\begin{equation}\label{ContinuousLimit}
q(\mathbf{h}) \cdot \tilde{f}(\mathbf{\tilde{u}}) \, \, = \, \, \sum_{k=1}^K \,m_k(\mathbf{h})\cdot r_k(\mathbf{u}) + \mathcal{O}(d+1)\,,
\end{equation}
where $(\forall k)\ [\,m_k(\mathbf{h})\in \Mon(\mathbf{h}),\ \deg(m_k(\mathbf{h}))=d\,]$ for some $d\in \Z_{>0}$ and $\mathcal{O}(d+1)$ denotes terms whose total degree in $h_1,\ldots,h_n$ is larger than $d$.
\end{definition}

\begin{definition}\label{def-wcons}
A difference system~\eqref{fda} is {\em weakly consistent} or {\em w-consistent} with PDE \eqref{pde} if 
\begin{equation}\label{w-consistency}
(\,\forall\,j \in \{\,1,\ldots,p\,\}\,) \ \  [\,\tilde{f}_j\rhd f_j\,]\,.
\end{equation}
\end{definition}

It is clear that if one considers a single PDE $f(\mathbf{u})=0$ and a difference equation  $\tilde{f}(\mathbf{\tilde{u}})=0$ with implication $\tilde{f}\rhd {f}$, i.e.,
\begin{equation}\label{SingleImplication}
q(\mathbf{h}) \cdot \tilde{f}(\mathbf{\tilde{u}}) \, \, = \, \, m(\mathbf{h})\cdot f(\mathbf{u})+\mathcal{O}(\deg(m)+1)\,,\quad m\in \Mon(\mathbf{h})\,,
\end{equation}
then it is always possible to redefine $\tilde{f}' := q(\mathbf{h}) \tilde{f}/m(\mathbf{h})$, and there is a limit (cf.\ Proposition~\ref{prop-s-cons}) $\mathbf{h}\rightarrow 0$, i.e. $(\forall i)[\ h_i\rightarrow 0\ ]$, such that
\[
\tilde{f}' \xrightarrow[\mathbf{h}\rightarrow 0]{} f + \mathcal{O}(\mathbf{h})\,,
\]
taking Remark~\ref{TaylorExpansionElementsSigma} into account.

The condition~\eqref{w-consistency} means that Eqs.~\eqref{fda} \emph{approximate} Eqs.~\eqref{pde} and by this reason we call Eqs.~\eqref{fda}
\emph{finite difference approximation} (FDA) to Eqs.~\eqref{pde}.

\begin{remark}\label{ContLim}
Given $\tilde{f}(\tilde{\mathbf{u}})$, computation of $f(\mathbf{u})$ is straightforward and has been implemented as routine {\sc ContinuousLimit} in the Maple package {\sc LDA} \cite{GR'12}  ({\underline{L}inear \underline{D}ifference \underline{A}lgebra}).

\end{remark}

\begin{example}({\cite{Str'04}, Ex.\,1.4.2}) We consider the one-way wave equation
\begin{equation}\label{one-way-wave-equation}
f:=u_t+a\,u_x=0\,,\qquad u=u(t,x)\,,\quad a = \text{const}\,,
\end{equation}
where $a$ is a constant, $t$ represents time, and $x$ represents the spatial variable. For the classical Lax-Friedrichs discretization the difference form of Eq.~\eqref{one-way-wave-equation} for the grid with $t_{n+1}-t_n=h_1$,\ \ $x_{m+1}-x_m=h_2$ is given by
\begin{equation}\label{LFscheme}
\tilde{f} \, \, := \, \, \frac{\tilde{u}^{n+1}_m-\frac{1}{2}\left(\tilde{u}^{n}_{m+1}+\tilde{u}^{n}_{m-1}\right)}{h_1}+
a \, \frac{\tilde{u}^n_{m+1}-\tilde{u}^n_{m-1}}{2h_2} \, \, = \, \, 0\,,
\end{equation}
where $\tilde{u}^n_m:=\tilde{u}(nh_1,mh_2)$ is a smooth grid function. The Taylor expansion of~\eqref{LFscheme} around the point $(t=nh_1,x=mh_2)$ reads
\[
 \tilde{f} \xrightarrow[h_1,h_2\rightarrow 0]{} (u_t+a\,u_x)\,h_1 + \frac{1}{2}h_1^2\,u_{tt}-\frac{1}{2}h_2^2\,u_{xx}+\frac{1}{6}a\,h_1h_2^2\,u_{xxx}+\cdots\,.
\]
So $\tilde{f}\rhd f$ as $h_1,h_2\rightarrow 0$, and $\tilde{f}/h_1 \rightarrow f$ as $h_1\rightarrow 0$  and $h_2^2/h_1\rightarrow 0$.
\end{example}

Now we formulate the property of strong consistency which, if it holds, links the radical differential ideal generated by the differential system with the perfect difference ideal generated by a difference approximation to the system.

\begin{definition}
An FDA \eqref{fda} to a PDE system \eqref{pde} with $p>1$ is \emph{strongly consistent} or \emph{s-consistent} if
\begin{equation}
(\,\forall \tilde{f}\in \llbracket \tilde{F} \rrbracket\,)\  (\,\exists\, {\mathcal{F}}\subset \llbracket F \rrbracket\,)\  \ [\,\tilde{f}\rhd {\mathcal{F}}\,]\,. \label{s-cond}
\end{equation}
\label{def-scon}
\end{definition}

In \cite[Definition~12]{G'12} we defined s-consistent FDA to PDE for Cartesian grids $h_1=h_2=\cdots =h_n=h$ as the ones satisfying the condition
\begin{equation*}
(\,\forall \tilde{f}\in \llbracket \tilde{F} \rrbracket\,)\  (\,\exists
f\in \llbracket F \rrbracket \,)\ [\,\tilde{f}\rhd f\,] \label{s-cond-Carthesian}
\end{equation*}
that corresponds to the Taylor expansion
\begin{equation*}\label{ContinuousLimit-Carthesian}
\tilde{f}(\mathbf{\tilde{u}}) \, \, = \, \, h^kf(\mathbf{u}) + \mathcal{O}(h^{k+1})
\end{equation*}
of the form~\eqref{SingleImplication} and to the consistency condition $\tilde{f}_1\xrightarrow[h\rightarrow 0]{} f $ with $\tilde{f}_1=\tilde{f}/h^k$.

The s-consistency property~\eqref{s-cond} implies the existence of a limit $\mathbf{h}\rightarrow 0$ such that in this limit every difference-algebraic consequence of~\eqref{fda} (cf.\ De\-fi\-nition~\ref{DiffCons}), after clearing denominators containing elements in~$\mathbf{h}$, is a diffe\-ren\-tial-al\-ge\-bra\-ic consequence of~\eqref{pde}.

\begin{proposition}\label{prop-s-cons}
Let FDA~\eqref{fda} to PDE system~\eqref{pde} be defined on a regular solution grid~\eqref{grid}. If the perfect difference ideal generated by the set $\tilde{F}$ satisfies the condition~\eqref{s-cond} of s-consistency, then there is a limit $\mathbf{h}\rightarrow 0$ such that
\begin{equation}
(\,\forall \tilde{f}\in \llbracket \tilde{F} \rrbracket\,)\  (\,\exists
f\in \llbracket F \rrbracket,\,\mu = {\cal O}(\mathbf{h}),\,\,d\in \Z_{\geq 0}  \,)\ \left[\,\frac{\tilde{f}}{\mu^d}\xrightarrow[\mathbf{h}\rightarrow 0]{} f \,\right]. \label{s-cond-consequence}
\end{equation}
\end{proposition}

\begin{proof} Let $\tilde{f}(\mathbf{\tilde{u}})\in \llbracket \tilde{F} \rrbracket$. Then the equation $\tilde{f}(\mathbf{\tilde{u}})=0$ is a difference-algebraic consequence of system \eqref{fda}. Without loss of generality one may assume that the denominators in $\tilde{f}$  containing elements in $\mathbf{h}$ have been cleared. Now we consider the following limit $\mathbf{h}\rightarrow 0$:
\begin{equation*}\label{Limit:mu}
h_i:=a_i\,\mu\,,\quad 0<a_i\in \Q\,,\quad \mu \in \R_{> 0}\,,\quad i=1,\ldots,n\,,\quad \mu\rightarrow 0\,.
\end{equation*}
Then, from the Taylor expansion~\eqref{ContinuousLimit} we obtain
\[
 \tilde{f}(\mathbf{\tilde{u}}) \, \, = \, \, \mu^d \left(\sum_{k=1}^K \,m_k(\mathbf{a})\cdot r_k(\mathbf{u}) + \mathcal{O}(\mu)\right),
\]
where $\quad \mathbf{a}:=(a_1,\ldots, a_n) \in \Q^n$\ \ and\ \ $\sum_{k=1}^K \,m_k(\mathbf{a})\cdot r_k(\mathbf{u})\in \llbracket F\rrbracket$.\qed
\end{proof}

Let $F$ be a PDE system~\eqref{pde} and $\tilde{F}$ be a w-consistent FDA~\eqref{fda}. In practice, FDA (scheme) can be obtained from PDE by approximation of the partial derivatives occurring in PDE with appropriate finite differences. Another way of discretization is to apply the method suggested in~\cite{GBM'06}. In either case, to verify the s-consistency condition~\eqref{s-cond}, one has to reformulate this condition to make it algorithmic.

The first step in this direction is to use the following statement.
\begin{theorem}\label{GB:s-constistency}
A difference approximation (\ref{fda}) to a differential system (\ref{pde}) is s-consistent if and only if a \Gr basis $\tilde{G}\subset \differencering$ of the difference
ideal $[\tilde{F}]$ satisfies
\begin{equation}
(\,\forall \tilde{g}\in \tilde{G}\,)\ (\,\exists \,G\subset \llbracket F \rrbracket\,)\  [\,\tilde{g}\rhd G\,]\,. \label{cons-gb}
\end{equation}
\end{theorem}

\begin{proof}
For a Cartesian grid with $h_1=h_2=\cdots=h_n$ the proof was given in \cite[Thm.~3]{G'12}. Its extension to the grid~\eqref{grid} is straightforward.\qed
\end{proof}

If the difference \Gr basis $\tilde{G}$ is finite and we can construct it in finitely many steps, then we can compute the Taylor expansion~\eqref{ContinuousLimit} of every element $\tilde{f}(\mathbf{\tilde{u}})\in~\tilde{G}$ and obtain
\begin{equation*}\label{Map:ContLim}
  \frac{\tilde{f}}{\mu^d} \xrightarrow[\mu\rightarrow 0]{} f=\sum_{k=1}^K \,m_k(\mathbf{q})\cdot r_k(\mathbf{u})\,.
\end{equation*}

Furthermore, to check the radical ideal membership $f\in \llbracket F \rrbracket$, one can compute algorithmically a differential Thomas decomposition (cf.\ Section~\ref{sec:differentialthomas}, \cite{GLHR'18,BGLHR'12,Robertz'14}) of the PDE~\eqref{pde}
(with respect to any ranking $\differentialranking$)
into the subsystems $S_1,\ldots,S_t$ with disjoint solution sets and use the relation
\begin{equation}\label{DiffIdealMembership}
f\in \llbracket F \rrbracket \quad \iff \quad
\text{NF}(f,S_1,\differentialranking) \, = \, \ldots \, = \, \NF(f,S_t,\differentialranking) \, = \, 0\,,
\end{equation}
where $\text{NF}(f,S_i,\differentialranking)$ denotes the Janet normal form of $f$ modulo $S_i$,
as defined above (before Definition~\ref{de:differentialpassive}, cf.\ also \cite[Prop.\,2.2.50]{Robertz'14}).

Since the difference polynomial ring $\differencering$ is non-Noetherian~\cite{GLS'15}, the computation of a \Gr basis $\tilde{G}$ is not algorithmic. Hence, Algorithm~\ref{StandardBasis} may not terminate (cf.\ also \cite{G'12,GLS'15}).
However, instead one can apply the difference triangular decomposition
as developed in Section~\ref{DifferenceThomasDecomposition}.

If the input PDE system is linear, then the condition~\eqref{cons-gb} is algorithmic and can be verified by computing \Gr bases of the ideals generated by $F$ and $\tilde{F}$~\cite{GR'10}. All related computations can be done by using the relevant routines of the Maple packages {\sc LDA}~\cite{GR'12} and {\sc Janet}~\cite{Maple-Janet'03}.

\section{Incompressible Navier-Stokes equations}
\label{sec:NavierStokesEquations}
The Navier-Stokes equations for a three-dimensional incompressible flow of constant viscosity can be written as
\begin{equation}\label{NSE}
\left\lbrace
\begin{array}{rl}
f_0 := &\ \partial_1u+\partial_2 v+ \partial_3 w =0\,, \hfill \\[4pt]
f_1 := &\ \partial_t u+u\partial_1u+v\partial_2u +w\partial_3u +\partial_1p -\frac{1}{\mathrm{Re}}\nabla^2 u=0\,,\hfill \\[4pt]
f_2 := &\ \partial_t v+u\partial_1v+v\partial_2v +w\partial_3v +\partial_2p -\frac{1}{\mathrm{Re}}\nabla^2 v=0\,,\hfill \\[4pt]
f_3 := &\ \partial_t w + u\partial_1 w + v \partial_2w +w \partial_3 w +\partial_3 p -\frac{1}{\mathrm{Re}}\nabla^2 w=0\,.\hfill \\
\end{array}
\right.
\end{equation}
Here $\mathbf{u}(\mathbf{x},t)$  is the velocity vector  $\mathbf{u}:=(u,v,w)$, $\mathbf{x}:=(x_1,x_2,x_3)$ is the vector of Cartesian coordinates, $p(\mathbf{x},t)$ is the pressure, $\partial_i:=\partial_{x_i}$ $(i\in\{1,2,3\})$, $\nabla^2$ is the Laplace operator $\nabla^2:=\partial_1^2+\partial_2^2+\partial_3^2$ and $\mathrm{Re}$ is the Reynolds number.

\begin{remark}\label{rem:symmetry}
The PDE system~\eqref{NSE} consist of the \emph{continuity equation}
or \emph{incompressibility condition}  $f_0=0$  that represents  conservation of mass and the {\em momentum equations}  $f_i=0$ $(i\in\{1,2,3\})$ that represent conservation of momentum. As a consequence of these fundamental conservation laws, the system is invariant under a permutation of the coordinates provided the corresponding permutation is applied to the components of $\mathbf{u}$.
\end{remark}

One can rewrite system~\eqref{NSE} equivalently as
\begin{equation}\label{NSE1}
\left\lbrace
\begin{array}{rl}
F_0 := &\ \partial_1u+\partial_2 v+ \partial_3 w =0\,, \hfill \\[4pt]
F_1 := &\ \partial_t u+\partial_1(u^2)+\partial_2(uv) +\partial_3(uw) +\partial_1p -\frac{1}{\mathrm{Re}}\nabla^2 u=0\,,\hfill \\[4pt]
F_2 := &\ \partial_t v+\partial_1(uv)+\partial_2(v^2) +\partial_3(vw) +\partial_2p -\frac{1}{\mathrm{Re}}\nabla^2 v=0\,,\hfill \\[4pt]
F_3 := &\ \partial_t w + \partial_1 (uw) + \partial_2(vw) +  \partial_3 (w^2) +\partial_3 p -\frac{1}{\mathrm{Re}}\nabla^2 w=0\,,\hfill \\
\end{array}
\right.
\end{equation}
where the nonlinear parts in the momentum equations are in {\em divergence form or conservative form}, and modulo the continuity equation the systems~\eqref{NSE} and~\eqref{NSE1} coincide.

For the Navier-Stokes equations in the form \eqref{NSE}, one can conveniently use vector notation, which has the advantage of brevity, and rewrite these equations as
\begin{equation}
\nabla\cdot \mathbf{u}=0\,, \quad
\partial_t \mathbf{u} +(\mathbf{u}\cdot \nabla)\, \mathbf{u}
+\nabla p -\frac{1}{\mathrm{Re}}\nabla^2\, \mathbf{u}=0\,, \label{VNSE}
\end{equation}
where $\nabla:=(\partial_1,\partial_2,\partial_3)$ is the nabla operator.

Let $\differentialranking$ be the ranking that compares first the monomials in the partial derivations $\partial_t$, $\partial_1$, $\partial_2$, $\partial_3$  (cf.\,Eq.\,\eqref{MonDelta})
with respect to the lexicographic ordering and then, in the case of equal differential monomials,
compares differential indeterminates (dependent variables) as
\begin{equation}\label{eq:NavierStokesRanking}
\partial_t \differentialranking \partial_1 \differentialranking \partial_2 \differentialranking \partial_3\quad \mathrm{and}\quad p \differentialranking u \differentialranking v \differentialranking w\,.
\end{equation}
The (non-admissible) prolongation $\partial_t F_0=\nabla\cdot \partial_t \mathbf{u}=0$ of the continuity equation and its reduction modulo the vector momentum equation yields the
\emph{pressure Poisson equation}
\begin{equation}\label{PPE}
\nabla^2p + \nabla\cdot (\mathbf{u}\cdot \nabla)\, \mathbf{u}=0\,,
\end{equation}
which is the \emph{integrability condition}~(cf.~\cite{Seiler'10},\,p.~50) to Eqs.~\eqref{VNSE} and to~\eqref{NSE1} as well. Clearly, the differential system~\eqref{VNSE}, \eqref{PPE} is passive and simple (cf.~Definition~\ref{de:differentialsimple}). We mention that the arbitrariness of analytic solutions to the incompressible Navier-Stokes equations can also be represented by the
differential counting polynomial $\infty^{\ell^3+\frac{11}{2} \ell^2 + \frac{17}{2} \ell + 4}$ \cite[Example~4.7]{LangeHegermann}.

Eq.~\eqref{PPE} can be expressed in terms of the continuity and momentum equations as
\begin{equation}\label{IntCond}
F_4 := \partial_1 F_1+\partial_2 F_2+\partial_3 F_3+\frac{1}{\mathrm{Re}} \left( \partial_1^2 F_0 + \partial_2^2 F_0 + \partial_3^2 F_0 \right) - \partial_t F_0\,.
\end{equation}
It is significant that both Eqs.~\eqref{PPE} and \eqref{IntCond} preserve permutational symmetry in line with Remark~\ref{rem:symmetry}.

Now we consider the following class of FDA to~\eqref{NSE} defined on the four-di\-men\-sio\-nal grid~\eqref{grid}
\begin{equation}\label{DNSE}
\mathbf{D}\cdot \mathbf{\tilde{u}}=0\,,\quad  D_t\mathbf{\tilde{u}}+(\mathbf{\tilde{u}}\cdot \mathbf{D})\, \mathbf{\tilde{u}}
+\mathbf{D}\,\tilde{p} -\frac{1}{\mathrm{Re}}\,\tilde{\Delta}_3\, \mathbf{\tilde{u}}=0\,,
\end{equation}
where $D_t$ approximates $\partial_t$, $\mathbf{D}=(D_1,D_2,D_3)$ approximates $\nabla$ and ${\tilde{\Delta}}_3$ approximates $\nabla^2$. It is clear that system~\eqref{DNSE} is w-consistent with
Eqs.~\eqref{VNSE}.

As an example of such finite difference approximations on the grid~\eqref{grid}, one can consider the following one
\begin{equation}\label{discretizations}
D_t=\frac{\automorphism_t - 1}{\tau}\,,\quad  D_{i}=\frac{\automorphism_{i}-\automorphism_{i}^{-1}}{2h_i}\,,\quad  \tilde{\Delta}_3=\sum_{j=1}^3\frac{\automorphism_i-2+\automorphism_i^{-1}}{h_j^2}\,,
\end{equation}
where $i\in \{1,2,3\}$ and $\tau,h_i\in \R_{>0}$.

If one considers a difference analogue of Eq.~\eqref{eq:NavierStokesRanking} satisfying
\begin{equation}\label{eq:NavierStokesDifferenceRanking}
\automorphism_t \succ \automorphism_1 \succ \automorphism_2 \succ \automorphism_3 \quad \mathrm{and}\quad \tilde{p} \succ \tilde{u} \succ \tilde{v} \succ \tilde{w}\,,
\end{equation}
then completion of 
\eqref{DNSE} to a passive form by Algorithm~\ref{StandardBasis} is equivalent to enlargement of this system with the difference integrability condition
\begin{equation}\label{DPPE}
(\mathbf{D}\cdot \mathbf{D})\,\tilde{p}+\mathbf{D}\cdot (\mathbf{\tilde{u}}\cdot \mathbf{D})\,\mathbf{\tilde{u}} - \frac{1}{\mathrm{Re}}\mathbf{D}\cdot \tilde{\Delta}_3\,\mathbf{\tilde{u}}=0\,.
\end{equation}
Eq.~\eqref{DPPE} approximates Eq.~\eqref{PPE} and is obtained, in full analogy with the differential case, by the prolongation $\mathbf{D}\cdot D_t\tilde{\mathbf{u}}=0$ of the discrete continuity equation in \eqref{DNSE} and its reduction modulo the discrete vector momentum equation.

\begin{remark}
\label{rem:divergence-free}
Because of equality
$
\mathbf{D}\cdot \tilde{\Delta}\,\mathbf{\tilde{u}}=\tilde{\Delta} \, \mathbf{D}\cdot \mathbf{\tilde{u}}\,,
$
the last term in Eq.~\eqref{DPPE} can be omitted if one considers a solution satisfying the discrete continuity equation $\mathbf{D}\cdot \tilde{\mathbf{u}}=0$. In such a case instead of Eq.~\eqref{DPPE} one can
use
\begin{equation}\label{DPPERED}
(\mathbf{D}\cdot \mathbf{D})\,\tilde{p}+\mathbf{D}\cdot (\mathbf{\tilde{u}}\cdot \mathbf{D})\,\mathbf{\tilde{u}}=0\,.
\end{equation}
\end{remark}

The left-hand sides of Eqs.~\eqref{DNSE} and \eqref{DPPERED} form a difference \Gr basis of the ideal they generate in $\Q(\mathrm{Re,h})\{\mathbf{\tilde{u}},\tilde{p}\}$. Hence, by Theorem~\ref{GB:s-constistency}, FDA \eqref{DNSE}, \eqref{DPPERED} is s-consistent with Eqs.~\eqref{NSE}, \eqref{PPE}.

Similarly, we can approximate Eqs.~\eqref{NSE1} as follows:
\begin{equation}\label{DNSE1}
\left\lbrace
\begin{array}{rl}
\tilde{F}_0 := &\ D_1\,\tilde{u}+D_2\, \tilde{v}+ D_3\, \tilde{w} =0\,, \hfill \\[4pt]
\tilde{F}_1 := &\ D_t\, \tilde{u}+D_1\,\tilde{u}^2+D_2\,(\tilde{u}\,\tilde{v}) +D_3\,(\tilde{u}\,\tilde{w}) +D_1\,\tilde{p} -\frac{1}{\mathrm{Re}}\,\tilde{\Delta}_3\, \tilde{u}=0\,,\hfill \\[4pt]
\tilde{F}_2 := &\ D_t \,\tilde{v} + D_1\,(\tilde{u}\,\tilde{v})+ D_2\,(\tilde{v}^2) + D_3\,(\tilde{v}\,\tilde{w}) +D_2\,\tilde{p} -\frac{1}{\mathrm{Re}}\,\tilde{\Delta}_3\, \tilde{v}=0\,,\hfill \\[4pt]
\tilde{F}_3 := &\ D_t \tilde{w} + D_1 (\tilde{u}\tilde{w}) + D_2(\tilde{v}\tilde{w}) +  D_3 (\tilde{w}^2) +D_3 \tilde{p} -\frac{1}{\mathrm{Re}}\,\tilde{\Delta}_3\, \tilde{w}=0\,,\hfill \\
\end{array}
\right.
\end{equation}
and complete this system to a passive form by performing the head reduction of $\mathbf{D}\cdot D_t\tilde{\mathbf{u}}=0$ modulo the momentum equations in~\eqref{DNSE1}.
As a result, we obtain another discretization of Eq.~\eqref{IntCond}:
\begin{equation}\label{DIntCond}
\begin{array}{rl}
\tilde{F}_4 := &\big(\mathbf{D}\cdot \mathbf{D}\big)\,\tilde{p}+\big(D_1\cdot D_1\big)\,\tilde{u}^2+ \big(D_2\cdot D_2\big)\,\tilde{v}^2 + \big(D_3\cdot D_3\big)\,\tilde{w}^2 \\[4pt]
 &+2\,\big[\big(D_1\cdot D_2\big) (\tilde{u}\,\tilde{v}) + \big(D_1\cdot D_3\big) (\tilde{u}\,\tilde{w})+\big(D_2\cdot D_3\big) (\tilde{v}\,\tilde{w})\big] \\[4pt]
 & -\frac{1}{\mathrm{Re}}\,\big[ \big(D_1\cdot \tilde{\Delta}_3\big)\,\tilde{u}+\big(D_2\cdot \tilde{\Delta}_3\big)\,\tilde{v} + \big(D_3\cdot \tilde{\Delta}_3\big)\,\tilde{w}\big]=0\,.
\end{array}
\end{equation}
We emphasize that the right-hand sides of Eqs.\,\eqref{DNSE1} and \eqref{DIntCond} are tale redundant modulo $\tilde{F}_0$. However, we prefer to use this redundant form since it inherits the permutational symmetry of~Eqs.\,\eqref{NSE1} and their divergence (conservative) form.
\begin{proposition}\label{prop:s-consistency}
FDA $\tilde{F}:=\{\tilde{F}_0,\tilde{F}_1,\tilde{F}_2,\tilde{F}_3,\tilde{F}_4\}$ to~$F:=\{F_0,F_1,F_2,F_3,F_4\}$ is s-consistent.
\end{proposition}
\begin{proof}
By inspection of the leading terms, it is easy to see that $\tilde{F}$ is a head reduced difference \Gr basis of $[\tilde{F}]=\llbracket \tilde{F} \rrbracket$, and  Theorem~\ref{GB:s-constistency} implies the s-consistency. \qed
\end{proof}

To compare different FDA to Eqs.~\eqref{NSE}, we compare their numerical behavior with the exact non-stationary two-dimensional solution \cite{KM'85} originally derived by Taylor~\cite{Taylor'23} in his  study of decaying vortex flow. This solution is widely used as a benchmark for numerical solving of  Navier-Stokes equations (see, for example,~\cite{NPP:2015}) and we have already used it in~\cite{ABGLS'13} and \cite{ABGLS'17}.
\begin{eqnarray}\label{exactsol}
\left\lbrace
\begin{array}{l}
u=-e^{-\frac{2t}{\mathrm{Re}}}\cos(x)\sin(y)\,, \\[4pt]
 v=e^{-\frac{2t}{\mathrm{Re}}}\sin(x)\cos(y)\,, \\[4pt]
 p=-\frac{1}{4}e^{-\frac{4t}{\mathrm{Re}}}(\cos(2x)+\cos(2y))\,.
\end{array}
\right.
\end{eqnarray}
We consider here four difference approximations to the two-dimensional form of Eqs.\,\eqref{VNSE} and \eqref{PPE} with the grid functions 
\begin{equation}\label{2D:GridFunction}
  \tilde{\mathbf{u}}^{n}_{j,k} := \tilde{\mathbf{u}}(jh,kh,n\tau)\,,\quad  \tilde{p}^{n}_{j,k} := \tilde{p}(jh,kh,n\tau)\,,\quad (j,k,n)\in \Z^3\,,
\end{equation}
and the following approximations of partial derivatives
\begin{equation}\label{2Ddiscretizations}
D_t=\frac{\automorphism_t - 1}{\tau}\,,\quad  D_{i}=\frac{\automorphism_{i}-\automorphism_{i}^{-1}}{2h}\,,\quad  \tilde{\Delta}_2=\frac{\automorphism_1+\automorphism_2 -4 + \automorphism_1^{-1}+\automorphism_2^{-1}}{h^2}\,,
\end{equation}
where $i\in \{1,2\}$ and $\tau,h\in \R_{>0}$.

\begin{description}
\item{FDA1~\cite{ABGLS'17}}
\begin{equation}\label{FDA1}
\left\lbrace
\begin{array}{rl}
\tilde{F}^{(1)}_0 := &\ D_1\,\tilde{u}+D_2\, \tilde{v}=0\,, \hfill \\[4pt]
\tilde{F}^{(1)}_1 := &\ D_t\, \tilde{u}+D_1\,\tilde{u}^2+D_2\,(\tilde{u}\,\tilde{v}) +D_1\,\tilde{p} -\frac{1}{\mathrm{Re}}\,\tilde{\Delta}_2\, \tilde{u}=0\,,\hfill \\[4pt]
\tilde{F}^{(1)}_2 := &\ D_t \,\tilde{v} + D_1\,(\tilde{u}\,\tilde{v})+ D_2\,(\tilde{v}^2) + D_2\,\tilde{p} -\frac{1}{\mathrm{Re}}\,\tilde{\Delta}_2\, \tilde{v}=0\,,\hfill \\[4pt]
\tilde{F}^{(1)}_3 := &\big(D_1\cdot D_1+D_2\cdot D_2\big)\,\tilde{p}+\big(D_1\cdot D_1\big)\,\tilde{u}^2+ \big(D_2\cdot D_2\big)\,\tilde{v}^2 \\[4pt]
 &+2\,\big(D_1\cdot D_2\big) (\tilde{u}\,\tilde{v}) -\frac{1}{\mathrm{Re}}\,\big[ \big(D_1\cdot \tilde{\Delta}_2\big)\,\tilde{u}+\big(D_2\cdot \tilde{\Delta}_2\big)\,\tilde{v} \big]=0\,.
\end{array}
\right.
\end{equation}
\item{FDA2~\cite{ABGLS'17}}
\begin{equation}\label{FDA2}
\left\lbrace
\begin{array}{rl}
\tilde{F}^{(2)}_0 := &\ D_1\,\tilde{u}+D_2\,\tilde{v}=0\,, \hfill \\[4pt]
\tilde{F}^{(2)}_1 := &\ D_t\, \tilde{u}+\tilde{u}\,D_1\,\tilde{u}+\tilde{v}\,D_2\,\tilde{u} + D_1\,\tilde{p} -\frac{1}{\mathrm{Re}}\,\tilde{\Delta}_2\,\tilde{u}=0\,,\hfill \\[4pt]
\tilde{F}^{(2)}_2 := &\ D_t\tilde{v}+\tilde{u}\,D_1\,\tilde{v}+\tilde{v}\,D_2\,\tilde{v} +D_2\,\tilde{p} -\frac{1}{\mathrm{Re}}\tilde{\Delta}_2\,\tilde{v}=0\,,\hfill \\[4pt]
\tilde{F}^{(2)}_3 := &\ \tilde{\Delta}_2\, \tilde{p} + (D_1\,\tilde{u})^2 + 2\, (D_1\,\tilde{v})(D_2\,\tilde{u}) + (D_2\,\tilde{v})^2=0\,.\hfill \\
\end{array}
\right.
\end{equation}
\item{FDA3}
\begin{equation}\label{FDA3}
\left\lbrace
\begin{array}{rl}
\tilde{F}^{(3)}_0 := &\ D_1\,\tilde{u}+D_2\, \tilde{v}=0\,, \hfill \\[4pt]
\tilde{F}^{(3)}_1 := &\ D_t\, \tilde{u}+D_1\,\tilde{u}^2+D_2\,(\tilde{u}\,\tilde{v}) +D_1\,\tilde{p} -\frac{1}{\mathrm{Re}}\,\tilde{\Delta}_2\, \tilde{u}=0\,,\hfill \\[4pt]
\tilde{F}^{(3)}_2 := &\ D_t \,\tilde{v} + D_1\,(\tilde{u}\,\tilde{v})+ D_2\,(\tilde{v}^2) + D_2\,\tilde{p} -\frac{1}{\mathrm{Re}}\,\tilde{\Delta}_2\, \tilde{v}=0\,,\hfill \\[4pt]
\tilde{F}^{(3)}_3 := &\big(D_1\cdot D_1+D_2\cdot D_2\big)\,\tilde{p}+\big(D_1\cdot D_1\big)\,\tilde{u}^2+ \big(D_2\cdot D_2\big)\,\tilde{v}^2 \\[4pt]
 &+2\,\big(D_1\cdot D_2\big) (\tilde{u}\,\tilde{v})=0\,.
\end{array}
\right.
\end{equation}
\item{FDA4}
\begin{equation}\label{FDA4}
\left\lbrace
\begin{array}{rl}
\tilde{F}^{(4)}_0 := &\ D_1\,\tilde{u}+D_2\,\tilde{v}=0\,, \hfill \\[4pt]
\tilde{F}^{(4)}_1 := &\ D_t\, \tilde{u}+\tilde{u}\,D_1\,\tilde{u}+\tilde{v}\,D_2\,\tilde{u} + D_1\,\tilde{p} -\frac{1}{\mathrm{Re}}\,\tilde{\Delta}_2\,\tilde{u}=0\,,\hfill \\[4pt]
\tilde{F}^{(4)}_2 := &\ D_t\tilde{v}+\tilde{u}\,D_1\,\tilde{v}+\tilde{v}\,D_2\,\tilde{v} +D_2\,\tilde{p} -\frac{1}{\mathrm{Re}}\tilde{\Delta}_2\,\tilde{v}=0\,,\hfill \\[4pt]
\tilde{F}^{(4)}_3 := &\big(D_1\cdot D_1+D_2\cdot D_2\big)\,\tilde{p} + D_1\,\big(\tilde{u}\,D_1\,\tilde{u}\big)+ D_1\big(\tilde{v}\,D_2\,\tilde{u}\big)\hfill \\[4pt]
           & + D_2\,\big(\tilde{v}\,D_2\,\tilde{v}\big) -\frac{1}{\mathrm{Re}}\,\big[ \big(D_1\cdot \tilde{\Delta}_2\big)\,\tilde{u}+\big(D_2\cdot \tilde{\Delta}_2\big)\,\tilde{v} \big] =0\,.
\end{array}
\right.
\end{equation}
\end{description}

All these FDA are explicit as difference schemes, w-consistent and they inherit permutational symmetry of the Navier-Stokes equations. The first approximation, FDA1, given by Eqs.~\eqref{FDA1}, was constructed in the paper~\cite{ABGLS'17} where its s-con\-sis\-ten\-cy was established. The difference approximation FDA4, given by Eqs.~\eqref{FDA4} with the last term omitted, $\frac{1}{\mathrm{Re}}\,\mathbf{D}\cdot \tilde{\Delta}_2 \,\tilde{\mathbf{u}}$, which is reduced to zero modulo $\tilde{F}^{(4)}_0$ (cf.\ Remark~\ref{rem:divergence-free}), was derived in~\cite{GB'09} by the method suggested in~\cite{GBM'06}.

\begin{remark}\label{rem:NeedsContEq}
The equation $\tilde{F}^{(2)}_3=0$ in FDA2 provides a compact finite difference discretization of the Poisson pressure equation \eqref{PPE}. In this case
\begin{equation}\label{2D:ContEq}
(D_1\,\tilde{u})^2 + 2\, (D_1\,\tilde{v})(D_2\,\tilde{u}) + (D_2\,\tilde{v})^2 \xrightarrow[h\rightarrow 0]{} u_x^2+2v_xu_y+v_y^2
\end{equation}
whereas
\begin{equation}\label{2D:ContEq1}
     \nabla\cdot (\mathbf{u}\cdot \nabla)\, \mathbf{u}= u_x^2+2v_xu_y+v_y^2+uu_{xx}+vv_{yy}+vu_{xy}+uv_{xy}\,,
\end{equation}
and the right-hand sides of Eqs.~\eqref{2D:ContEq} and \eqref{2D:ContEq1} are equal modulo the continuity equation $u_x+v_y$ = 0.
\end{remark}

However, FDA2 is s-inconsistent, since $\tilde{F}^{(2)}_3 \not\in \llbracket \tilde{\cI}\rrbracket$, where $\llbracket \tilde{\cI}\rrbracket\subset \tilde{\cR}$ the ideal generated by $\{\tilde{F}^{(2)}_0,\tilde{F}^{(2)}_1,\tilde{F}^{(2)}_2\}$ by virtue of inequality
\[
\mathbf{D}\cdot \mathbf{D}=\frac{\automorphism_1^2+\automorphism_2^2-4+\automorphism_1^{-2}+\automorphism_2^{-2}}{4h^2}\neq \tilde{\Delta}_2\,,
\]
and there are consequences of Eqs.~\eqref{FDA2} implying the differential equations (cf.~Definition~\ref{Implication}) which are not consequences of \eqref{VNSE}. Below we explicitly demonstrate this for the linearized version of Eqs.~\eqref{NSE} called Stokes equations.

The approximation FDA3 given by Eqs.~\eqref{FDA3} differs from FDA1 in the structure of discrete Poisson pressure equation. In contrast to equation $\tilde{F}^{(3)}_3=0$ in Eqs.~\eqref{FDA3}, the equation $\tilde{F}^{(1)}_3=0$ in \eqref{FDA1} contains extra part $\frac{1}{\mathrm{Re}}\,\mathbf{D}\cdot \tilde{\Delta}_2 \,\tilde{\mathbf{u}}$ and can be omitted if $\tilde{u}$ satisfies the discrete continuity equation (see  Remark~\ref{rem:divergence-free}).

In the difference approximation FDA4, Eqs.~\eqref{FDA4}, the discrete pressure Poisson equation $\tilde{F}^{(4)}_3=0$, as opposed to that in FDA2, provides s-consistency of FDA4.

We compare these four schemes by using the following absolute/relative error formula
\begin{equation}
e_{g}^n =\max_{j,k} \frac{|g_{j,k}^n-g(x_j,y_k,t_n)|}{1+|g(x_j,y_k,t_n)|}\,,
\label{error}
\end{equation}
where $g\in\{u,v,p\}$ and $g(x,y,t)$ belongs to the exact solution~\eqref{exactsol}.

The governing differential system~\eqref{VNSE}, \eqref{PPE} is mixed elliptic-parabolic: the momentum equations are parabolic and the pressure Poisson equation is elliptic. It should be particularly emphasized that in our construction of a numerical solution for the initial-value problem with initial data taken from Eqs.\,\eqref{exactsol} at $t=0$, we use the discrete  momentum equations to determine velocities and the discrete pressure Poisson equation to determine pressure. In other words, we use the classical pressure-Poisson formulation of the Navier-Stokes equations (cf.~\cite{Rempfer'06}, Sect.~3.2) to solve the initial-value problem numerically. However, as this takes place, to construct numerical solution of the above given difference approximations we do not exploit the discrete di\-ver\-gence-free constraint $\mathbf{D}\cdot \tilde{\mathbf{u}}=0$ and use (cf.~\cite{ABGLS'13,ABGLS'17}) and use it only for verification of the obtained results.

We compute numerical solutions in the domain $[0,2\pi]\times[0,2\pi]\times[0,6]$ with the Reynolds number $\mathrm{Re}=100$. Figures~\ref{fig1}--\ref{fig3} contain the computed error for three different choices of $h$ (error in $u$ and $v$ coincides). We let $\tau=0.05 \times h$.

The results of our computational experiments shown in Figures~\ref{fig1} and \ref{fig2}. Except FDA3, which is unstable (Fig.~\ref{fig2}, top) by the lack of mass conservation (violation of the incompressibility condition), the other approximations clearly demonstrate the second order of convergence with respect to $h$, during which FDA1 far exceeds the others in accuracy. In addition, based on the numerical velocities obtained for the last FDA, the continuity equation is accurate to $10^{-10}$ (for $h=0.025$) what is incomparably better (Fig.~\ref{fig3}) than the obtained accuracies of other schemes for this equation. The error in the numerical continuity equation for the momentum grid functions in~\eqref{2D:GridFunction} was computed as the matrix Frobenius norm (cf.~\cite{GolubVanLoan:2013}, p.~71) 
\begin{equation}\label{error:divergence}
 ||\mathbf{D}\cdot \mathbf{\tilde{u}}||_{\mathcal{F}}:=\Big(\sum_{j,k} |D_1\tilde{u}^n_{j,k}+D_2\tilde{v}^n_{j,k}|^2\Big)^{\frac{1}{2}}\,.
\end{equation}

The superiority of FDA1 over the others FDAs is due to incorporation of  s-consistency, conservativity (divergence form) of the nonlinear terms in momentum equations and presence of the last term in the pressure Poisson equation~\eqref{DPPE}. In its turn, FDA2 provides lower accuracy in comparison with FDA4 because of its s-inconsistency. On the other hand, FDA2  the more stable than FDA3, since the former unlike the latter was constructed with application of the incompressibility condition (see Remark~\ref{rem:NeedsContEq}). If one correlates FDA1 with FDA4, then the last one does not have a conservation law (divergence) form and by this reason its accuracy is not so good.

Nearly all known finite difference approaches to solving incompressible Navier-Stokes equations in terms of ('primitive') variables $\{\tilde{\mathbf{u}},\tilde{p}\}$ started from the \emph{(fractional step) projection method} based on presentation of the vector momentum equation in~\eqref{VNSE} as a Helmholtz-Hodge decomposition~\cite{Brown:2001,GMS:2006,Rempfer'06} and on the use of the continuity equation (incompressibility condition) for correction of the velocity vector $\tilde{\mathbf{u}}$ on every time step. Our numerical experiments with the scheme FDA1 show, contrastingly, that it is sufficient to attain the fulfilment of the continuity equation by initial and/or boundary conditions.
\begin{figure}
\begin{center}
\includegraphics[width=.9\textwidth]{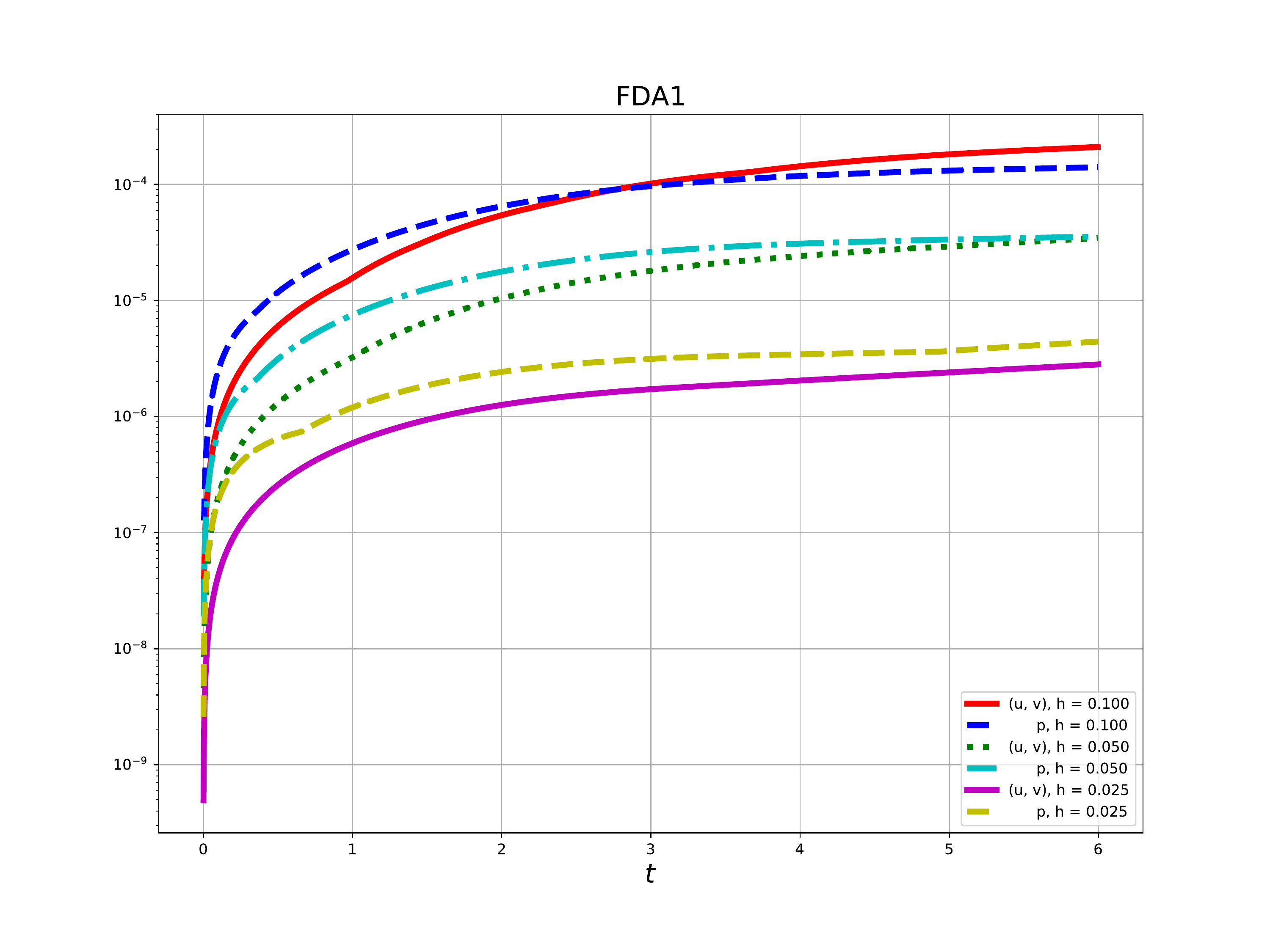} \\[-5mm]
\includegraphics[width=.9\textwidth]{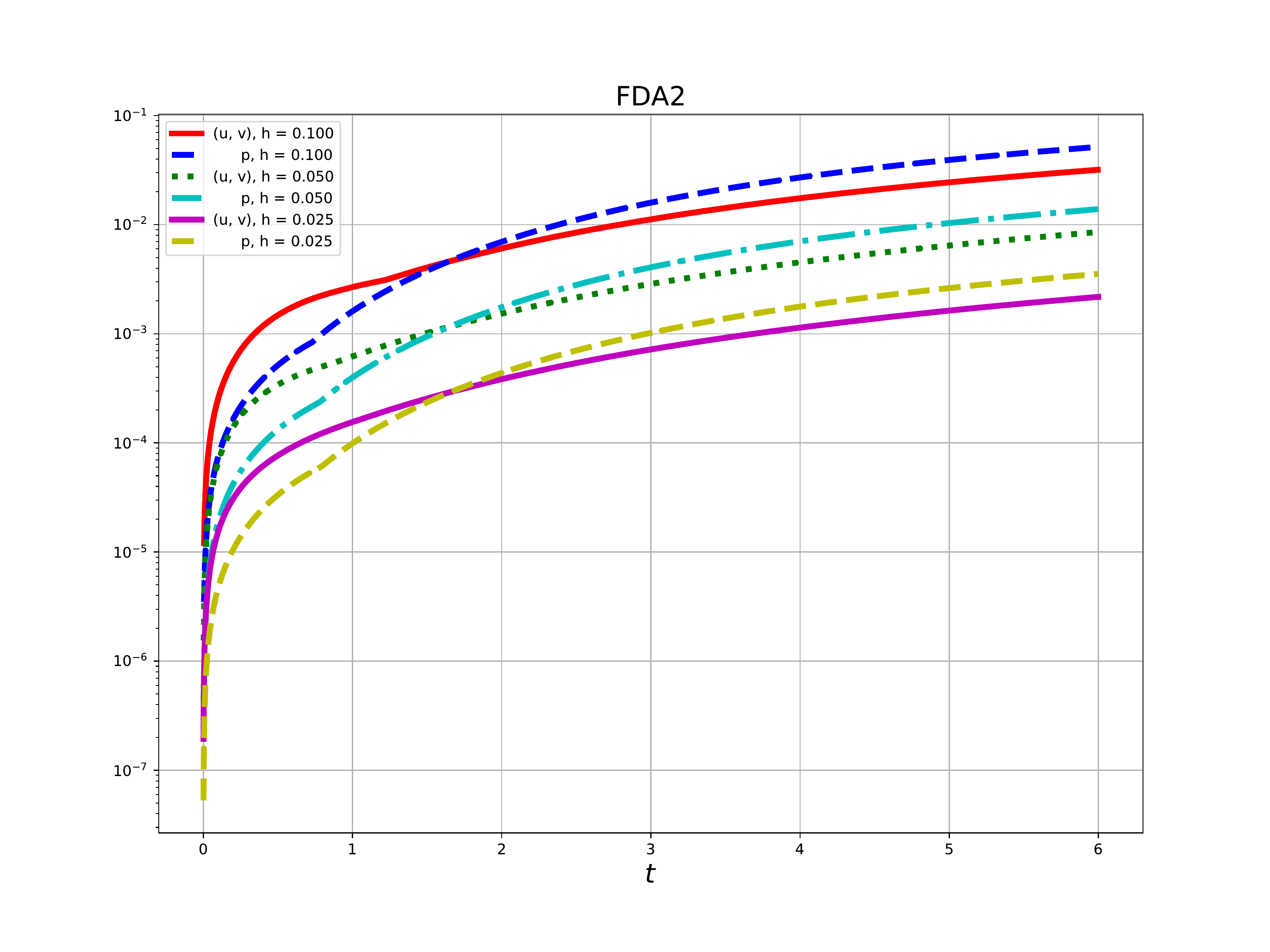} \\[-5mm]
\end{center}
\vspace*{-8mm}
\caption{Taylor decaying problem: error with different $h$ in the computed solution with FDA1 scheme and second-order standard discretizations FDA2\label{fig1}}
\end{figure}

\begin{figure}
\begin{center}
\includegraphics[width=.9\textwidth]{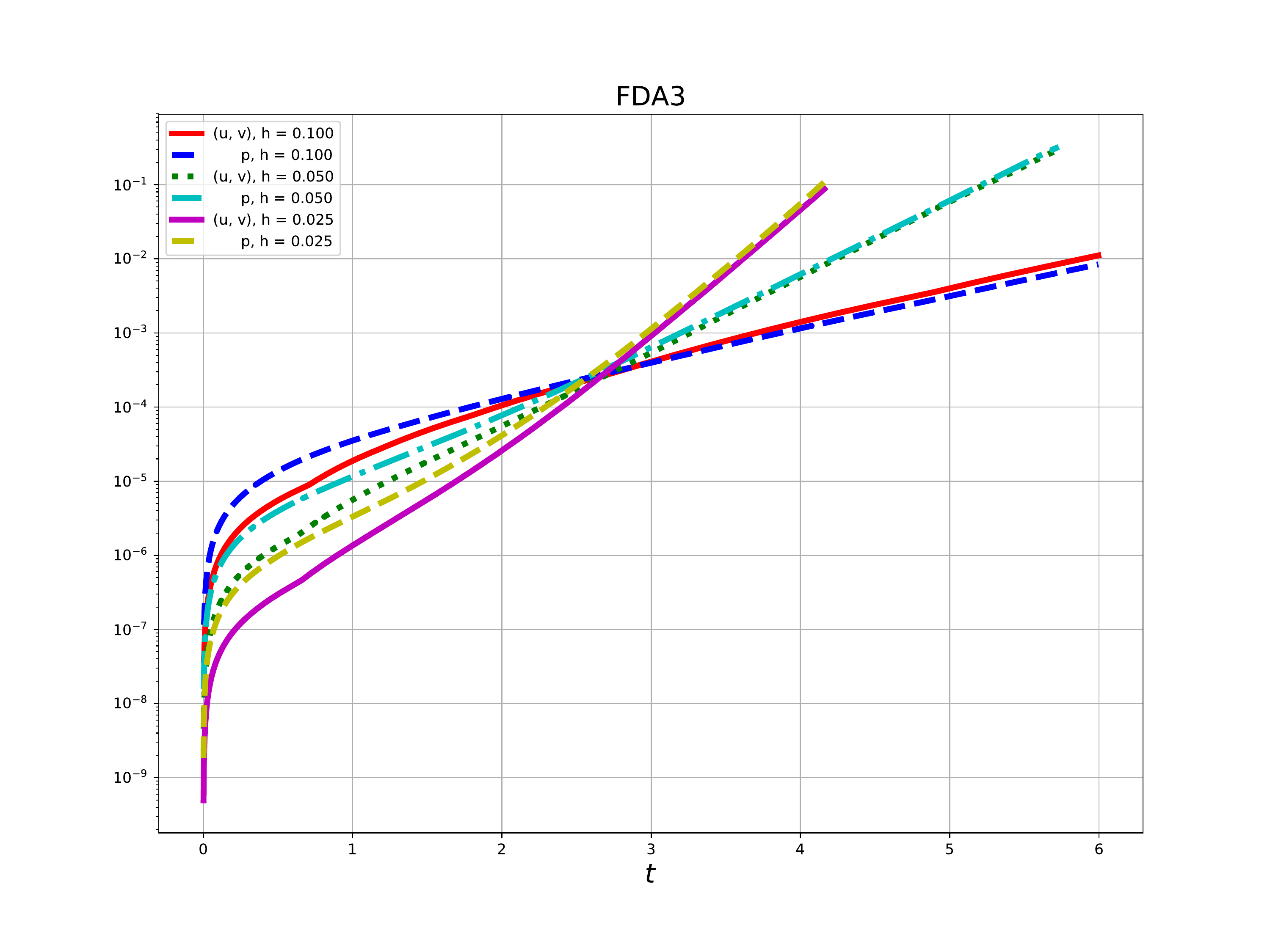} \\[-5mm]
\includegraphics[width=.9\textwidth]{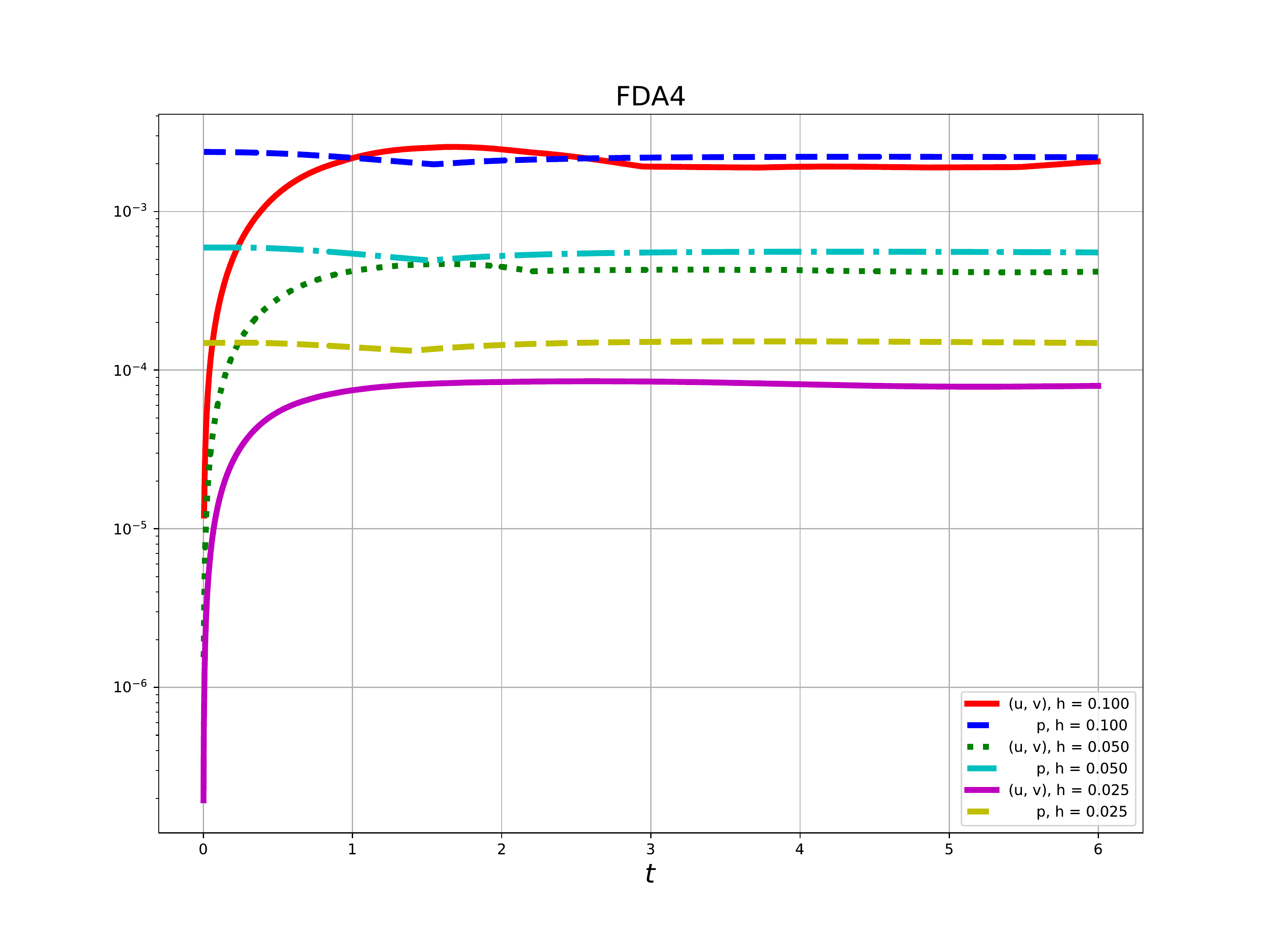} \\[-5mm]
\end{center}
\vspace*{-8mm}
\caption{Taylor decaying problem: error with different $h$ in the computed solution with FDA3 and FDA4 scheme\label{fig2}}
\end{figure}

\begin{figure}
\begin{center}
\includegraphics[width=.9\textwidth]{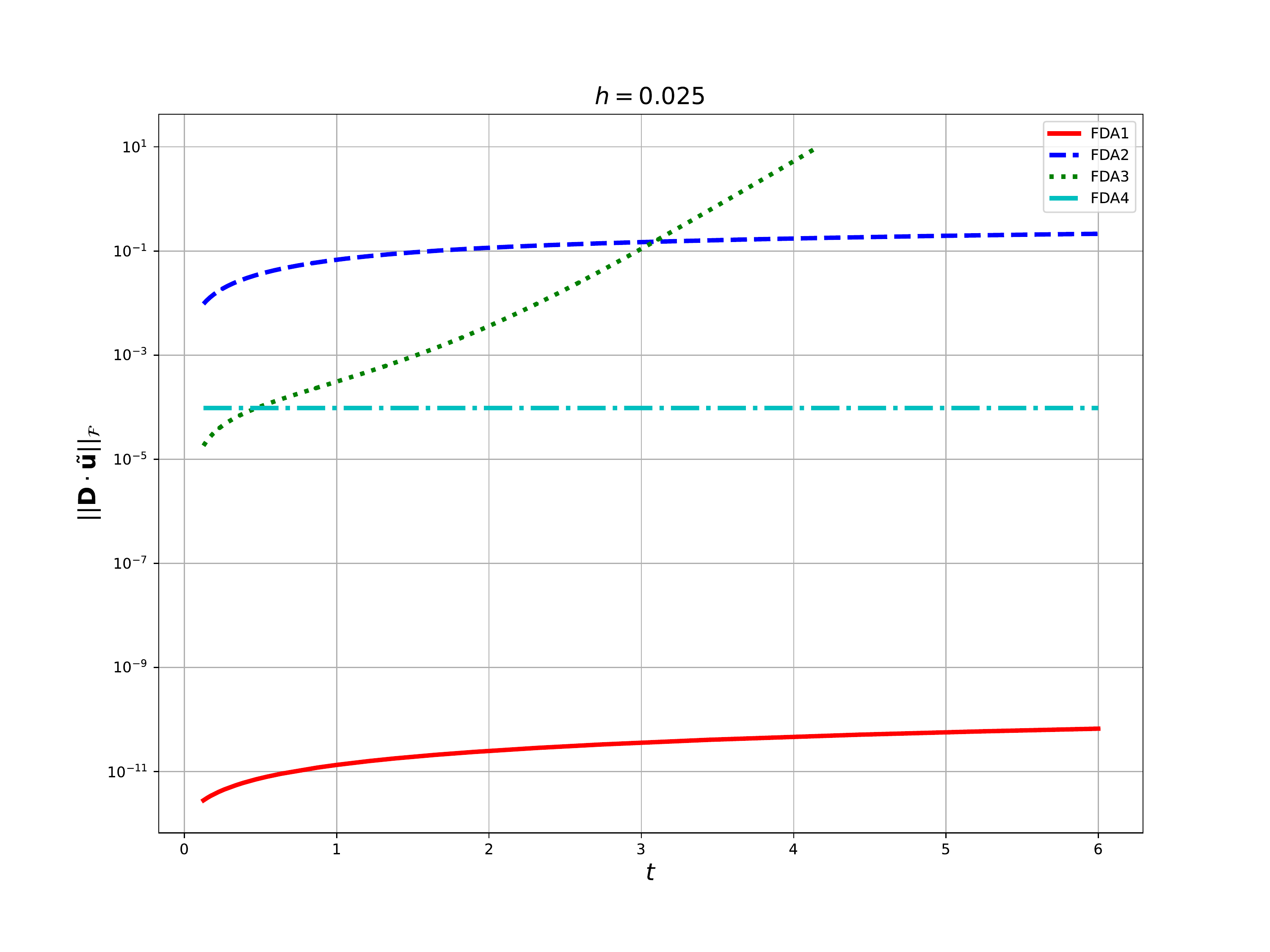} \\[-5mm]
\end{center}
\vspace*{-8mm}
\caption{Taylor decaying problem: computed value of error~\eqref{error:divergence}  for all four schemes\label{fig3}}
\end{figure}


To illustrate the fact that s-inconsistency has an adverse effect on the solution space of the discretized equations, we consider the incompressible Stokes equations flow given in vector notation by
\begin{equation}
\partial_t \mathbf{u} + \nabla  p -\frac{1}{\mathrm{Re}}\Delta\, \mathbf{u} \, \, = \, \, 0\,,\quad \nabla\cdot \mathbf{u} \, \, = \, \, 0\,. \label{SE}
\end{equation}
The linear PDE system \eqref{SE} approximates \eqref{VNSE} when $\mathrm{Re}$ is small (cf.\,\cite{Milne-Tompson}, Ch.\,22$\cdot 11$). Complemented by the gravity force $\rho \,\mathbf{g}$ ($\rho$ is the liquid density, $\mathbf{g}$ is the acceleration due to gravity) in the right-hand side of the first equation, Eqs.~\eqref{SE} have numerous applications in the description of fluid displacement processes that take place in porous media (see~\cite{Erdmann'06} and the references therein) and that are related to
both chemical and physical phenomena.

For Eqs.~\eqref{SE} the pressure Poisson equation~\eqref{PPE} becomes the
\emph{pressure Laplace equation}
\begin{equation}\label{PLE}
\Delta\, p \, \, = \, \, 0\,,
\end{equation}
and let the PDE system~\eqref{SE}--\eqref{PLE} be discretized as follows
\begin{equation}\label{DSE1}
D_t\mathbf{\tilde{u}}
+\mathbf{D}\,\tilde{p} -\frac{1}{\mathrm{Re}}\,{\tilde{\Delta}}_3\, \mathbf{\tilde{u}}   =   0\,,\ \ 
\mathbf{D}\cdot \mathbf{\tilde{u}}   =   0\,,\ \ {\tilde{\Delta}}_3\,\tilde{p}  =   0\,,
\end{equation}
where the difference approximations of partial derivatives given by Eqs.\,\eqref{discretizations}.

The passive form of Eqs.~\eqref{DSE1} reads
\begin{equation*}\label{PassiveDSE1}
D_t\mathbf{\tilde{u}}
+\mathbf{D}\,\tilde{p} -\frac{1}{\mathrm{Re}}\,{\tilde{\Delta}}_3 \, \mathbf{\tilde{u}}=0\,,\ \ \mathbf{D}\cdot \mathbf{\tilde{u}}=0\,,\ \ (\mathbf{D}\cdot \mathbf{D}) \, \tilde{p} = 0\,,\ \ {\tilde{\Delta}}_3\,\tilde{p}=0\,,\ \ \tilde{q}=0\,.
\end{equation*}
Here the difference polynomial $\tilde{q}$ is the reduced $S$-polynomial of the two preceding equations and, in accordance with Definition~\ref{Implication}, it implies
\begin{equation}\label{ImpliedEquations}0
   \tilde{q}:=S\left((\mathbf{D}\cdot \mathbf{D}) \, \tilde{p}, \, {\tilde{\Delta}}_3\,\tilde{p}\right)\ \rhd\    \partial_i^2\,\partial_j^2\, p\,,\quad i,j\in \{1,2,3\}\,.
\end{equation}

It is clear that FDA~\eqref{DSE1} is w-consistent with Eqs.~\eqref{SE} and \eqref{PLE}. However, none of  the differential equations occurring in \eqref{ImpliedEquations} is a consequence of~Eqs.~\eqref{SE}. This can be explicitly verified with the Maple packages {\sc LDA}~\cite{GR'12} and {\sc Janet}~\cite{Maple-Janet'03} by computing the related normal forms (cf.~\eqref{DiffIdealMembership}). Therefore, FDA~\eqref{DSE1} is s-in\-con\-sistent with the PDE system\,\eqref{SE}, \eqref{PLE}
in accordance with Theorem~\ref{GB:s-constistency}.

The difference polynomials in Eqs.~\eqref{DSE1} generate the perfect difference ideal whose element $\tilde{q}$, as indicated in~Eq.\,\eqref{ImpliedEquations}, in the continuous limit implies additional equations for the pressure. These equations together with the Laplace equation~\eqref{PLE}
restrict the pressure component to the exact solution
\begin{eqnarray*}
 &p:=&c_0+c_1x_1+c_2x_2+c_3x_3+c_4x_1^2+c_5x_1x_2+c_6x_2^2+c_7x_2x_3+c_8x_3^2\\
 &&   +\,c_9x_1x_3 + c_{10}x_1x_2x_3\,,\quad \mbox{where} \quad c_4+c_6+c_8=0\,,
\end{eqnarray*}
where $c_i$ $(0\leq i\leq 10)$ are arbitrary functions of $t$ satisfying the above constraint. It is clear that, in comparison with the s-con\-sis\-tent approximations, the s-in\-con\-sis\-tent one~\eqref{DSE1} significantly decreases the domain of solutions to the governing differential system that can be successfully constructed by numerical methods.

%
%
\section{Difference Thomas decomposition and s-consistency check}
\label{DifferenceThomasDecomposition}

Let $\tilde{S}$ be a \emph{system  of polynomial partial difference equations and inequations}
\begin{equation}\label{eq:polypartialdifference}
\tilde{f}_1 = 0\,, \ \ldots\,, \ \tilde{f}_s = 0\,, \ \tilde{f}_{s+1} \neq 0\,, \ \ldots\,, \ \tilde{f}_{s+t} \neq 0 \qquad (s, t \in \Z_{\ge 0})\,.
\end{equation}
Here $\tilde{f}_1$, \ldots, $\tilde{f}_{s+t}$ are elements
of the difference polynomial ring $\differencering$
in $\tilde{u}^{(1)}$, \ldots, $\tilde{u}^{(m)}$
with commuting automorphisms
$\setauto = \{ \automorphism_1, \ldots, \automorphism_n \}$.
A ranking $\differenceranking$ on $\differencering$ is fixed, so that
leaders, initials and discriminants of non-constant difference polynomials are defined
as in Definition~\ref{DifferentialRanking}.

\medskip

In this section we develop a difference analogue of the Thomas decomposition method for
differential systems (cf.\ Section~\ref{sec:differentialthomas}). The resulting
main Algorithm~\ref{alg:differencedecomp} is based, in particular, on Algorithm~\ref{alg:autoreducenonlin}
for auto-reduction of a difference system, which precedes the assignment of admissible automorphisms
by Janet division, 
and on Algorithm~\ref{alg:janetreducenonlin} computing Janet normal forms
modulo a Janet complete difference system, in order to check passivity.

Algorithm~\ref{alg:autoreducenonlin} 
performs reductions on
a finite system of difference equations (given by $\tilde{L}$), if possible, removing zero remainders
from the system. If a reduction occurs that results in a non-zero remainder, the original polynomial
is replaced by this remainder and the algorithm stops. Since in that case further splitting of the
system may be necessary to ensure non-va\-ni\-shing of initials, this situation is indicated by a flag
returned to the main Algorithm~\ref{alg:differencedecomp}.

The following notation is used in what follows.
For a difference system $\tilde{S}$ as above
let $\tilde{S}^{=}$ (resp.\ $\tilde{S}^{\neq}$) be the set
$\{ \tilde{f}_1, \ldots, \tilde{f}_s \}$ (resp.\ $\{ \tilde{f}_{s+1}, \ldots, \tilde{f}_{s+t} \}$).
Let $\tilde{E}$ be a difference ideal of $\differencering$
and let $\emptyset \neq \tilde{Q} \subseteq \differencering$ be multiplicatively closed and closed under $\automorphism_1$, \ldots, $\automorphism_n$. Then define
\[
\tilde{E} : \tilde{Q} \, \, := \, \, \{ \, \tilde{f} \in \differencering \mid \tilde{q} \, \tilde{f} \in \tilde{E} \mbox{ for some } \tilde{q} \in \tilde{Q} \, \}\,.
\]
Moreover, for $U \subseteq \Mon(\setauto) \, \tilde{\mathbf{u}}$ and $v \in \Mon(\setauto) \, \tilde{\mathbf{u}}$ we define
$U : v := \{ \, \theta \in \Mon(\setauto) \mid \theta \, v \in U \, \}$.

\begin{algorithm}
\label{AlgorithmAuto-reduce}
\DontPrintSemicolon
\KwInput{$\tilde{L} \subset \differencering \setminus \differencefield$ finite and a ranking $\differenceranking$ on $\differencering$ such that $\tilde{L} = \tilde{S}^{=}$ for some finite difference system $\tilde{S}$ which is quasi-simple as an algebraic system (in the finitely many indeterminates $\tilde{u}^{(\alpha)}_J$ occurring in it, totally ordered by $\differenceranking$)
}
\KwOutput{$a \in \{ \text{\bf true}, \text{\bf false} \}$ and $\tilde{L}' \subset \differencering \setminus \differencefield$ finite such that
\begin{equation*}\label{eq:LstrichL}
\diffidealgen{\tilde{L}'} : \tilde{Q} \, \, = \, \, \diffidealgen{\tilde{L}} : \tilde{Q}\,,
\end{equation*}
where $\tilde{Q}$ is the smallest multiplicatively closed subset of $\differencering$ containing all $\init(\theta \tilde{f})$, where $\tilde{f} \in \tilde{L}$
and $\theta \in \ld(\tilde{L} \setminus \{ \tilde{f} \}) : \ld(\tilde{f})$, and which is closed under $\automorphism_1$, \ldots, $\automorphism_n$,
and, in case $a = \text{\bf true}$, there exist no $\tilde{f}_1$, $\tilde{f}_2 \in \tilde{L}'$, $\tilde{f}_1 \neq \tilde{f}_2$, such that we have $v := \ld(\tilde{f}_1) = \theta \ld(\tilde{f}_2)$ for some $\theta \in \Mon(\setauto)$ and $\deg_{v}(\tilde{f}_1) \ge \deg_{v}(\theta \tilde{f}_2)$}
$\tilde{L}' \gets \tilde{L}$\;
\While{$\exists \, \tilde{f}_1, \tilde{f}_2 \in \tilde{L}', \, \tilde{f}_1 \neq \tilde{f}_2$ and $\theta \in \Mon(\setauto)$ such that we have $v := \ld(\tilde{f}_1) = \theta \ld(\tilde{f}_2)$ and $\deg_{v}(\tilde{f}_1) \ge \deg_{v}(\theta \tilde{f}_2)$}
{
  $\tilde{L}' \gets \tilde{L}' \setminus \{ \tilde{f}_1 \}$; \, $v \gets \ld(\tilde{f}_1)$\;
  $\tilde{r} \! \gets \! \init(\theta \tilde{f}_2) \, \tilde{f}_1 - \init(\tilde{f}_1) \, v^d \, \theta \tilde{f}_2$, $d \! := \! \deg_{v}(\tilde{f}_1) - \deg_{v}(\theta \tilde{f}_2)$\;
  \If{$\tilde{r} \neq 0$}{
    \Return $(\text{\bf false}, \tilde{L}' \cup \{ \tilde{r} \})$\;
  }
}
\Return $(\text{\bf true}, \tilde{L}')$\;
\caption{\algref{Auto-reduce} for difference algebra\label{alg:autoreducenonlin}}
\end{algorithm}

\medskip

Since leaders are dealt with in decreasing order with respect to $\differenceranking$,
and no ranking admits infinitely decreasing chains (cf.\ \cite[Ch.~0, Sect.~17, Lemma~15]{Kolchin}),
Algorithm~\ref{alg:autoreducenonlin} \emph{terminates}.
Its \emph{correctness} follows from the definition of $\tilde{E} : \tilde{Q}$.

\medskip

Before presenting Algorithm~\ref{alg:janetreducenonlin} which computes the Janet normal form
of a difference polynomial, we adapt the discussion of Janet division preceding Definition~\ref{de:admissiblederivations}
to the difference case.

\medskip

Janet division associates (with respect to a total ordering of $\setauto$)
to each $\tilde{f}_i = 0$ with $\ld(\tilde{f}_i) = \theta_i \tilde{u}^{(\alpha)}$
the set $\mu_i := \mu(\theta_i, \tilde{G}_{\alpha}) \subseteq \setauto$ (resp.\ $\overline{\mu}_i := \setauto \setminus \mu_i$) of
\emph{admissible}
(resp.\ \emph{non-admissible}) \emph{automorphisms}, where
\[
\tilde{G}_{\alpha} \, \, := \, \, \{ \, \theta \in \Mon(\setauto) \mid \theta \tilde{u}^{(\alpha)} \in \{ \ld(\tilde{f}_1), \ldots, \ld(\tilde{f}_s) \} \, \}\,.
\]
We call $\{ \tilde{f}_1 \!=\! 0, \ldots, \tilde{f}_s \!=\! 0 \}$ or $T := \{ (\tilde{f}_1, \mu_1), \ldots, (\tilde{f}_s, \mu_s) \}$
\emph{Janet complete}
if each $\tilde{G}_{\alpha}$ equals its Janet completion, $\alpha = 1$, \ldots, $m$.

\medskip

Let $\tilde{r} \in \differencering$.
If some $v \in \Mon(\setauto) \tilde{\mathbf{u}}$ occurs in $\tilde{r}$ for which there exists
$(\tilde{f}, \mu) \in T$ such that $v = \theta \ld(\tilde{f})$ for some $\theta \in \Mon(\mu)$
and $\deg_v(\tilde{r}) \ge \deg_v(\theta \tilde{f})$, then
$\tilde{r}$ is \emph{Janet reducible modulo $T$}.
In this case, $(\tilde{f}, \mu)$ is called a
\emph{Janet divisor} of $\tilde{r}$. If $\tilde{r}$ is not Janet reducible modulo $T$,
then $\tilde{r}$ is also said to be \emph{Janet reduced modulo $T$}.
Iterated pseudo-re\-duc\-tions of $\tilde{r}$ modulo $T$
yield its \emph{Janet normal form} $\NF(\tilde{r}, T, \ranking)$, which is the Janet reduced
difference polynomial $\tilde{r}'$ returned by Algorithm~\ref{alg:janetreducenonlin}.

\begin{algorithm}
\DontPrintSemicolon
\KwInput{$\tilde{r} \in \differencering$, $T = \{ \, (\tilde{f}_1, \mu_1), (\tilde{f}_2, \mu_2), \ldots, (\tilde{f}_s, \mu_s) \, \}$, and a ranking $\differenceranking$ on $\differencering$, where $T$ is Janet complete (with respect to $\ranking$)}
\KwOutput{$(\tilde{r}', \tilde{b}) \in \differencering \times \differencering$ such that
(1) if $\tilde{r} \in \differencefield$ or $T = \emptyset$, then $\tilde{r}' = \tilde{r}$, $\tilde{b} = 1$,
(2) otherwise $\tilde{r}'$ is Janet reduced modulo $T$ and
$$
\tilde{r}' + \diffidealgen{\tilde{f}_1, \ldots, \tilde{f}_s} \, \, = \, \, \tilde{b} \cdot \tilde{r} + \diffidealgen{\tilde{f}_1, \ldots, \tilde{f}_s}\,,
$$
where $\tilde{b}$ is in the multiplicatively closed set generated by
$$
\bigcup_{i=1}^s \, \{ \, \theta \init(\tilde{f}_i) \mid \theta \in \Mon(\setauto), \, \ld(\tilde{r}) \differenceranking \theta \ld(\tilde{f}_i) \, \} \cup \{ 1 \}
$$}
$\tilde{r}' \gets \tilde{r}$; \, $\tilde{b} \gets 1$\;
\If{$\tilde{r}' \not\in \differencefield$}
{
  $v \gets \ld(\tilde{r}')$\;
  \While{$\tilde{r}' \not\in \differencefield$ and there exist $(\tilde{f}, \mu) \in T$ and $\theta \in \Mon(\mu)$ such that $v = \theta \ld(\tilde{f})$ and $\deg_{v}(\tilde{r}') \ge \deg_{v}(\theta \tilde{f})$}
  {
    $\tilde{r}' \! \gets \! \init(\theta \tilde{f}) \, \tilde{r}' \! - \! \init(\tilde{r}') \, v^d \, \theta \tilde{f}$, $d \! := \! \deg_{v}(\tilde{r}') \! - \! \deg_{v}(\theta \tilde{f})$\;
    $\tilde{b} \gets \init(\theta \tilde{f}) \cdot \tilde{b}$\;
  }
  \For{each coefficient $\tilde{c}$ of $\tilde{r}'$ (as a polynomial in $v$)}
  {
    $(\tilde{r}'', \tilde{b}') \gets$ \algref{Janet-reduce}($\tilde{c}$, $T$, $\differenceranking$)\;
    replace the coefficient $\tilde{b}' \cdot \tilde{c}$ in $\tilde{b}' \cdot \tilde{r}'$ with $\tilde{r}''$ and replace $\tilde{r}'$ with this result\;
    $\tilde{b} \gets \tilde{b}' \cdot \tilde{b}$\;
  }
}
\Return $(\tilde{r}', \tilde{b})$\;
\caption{\algref{Janet-reduce} for difference algebra\label{alg:janetreducenonlin}}
\end{algorithm}

\begin{definition}\label{de:differencepassive}
Let $T = \{ \, (\tilde{f}_1, \mu_1), \ldots, (\tilde{f}_s, \mu_s) \, \}$
be Janet complete.
The difference system $\{ \, \tilde{f}_1 = 0, \ldots, \tilde{f}_s = 0 \, \}$ or $T$ is said to be \emph{passive}, if
the following Janet passivity conditions hold:
\begin{equation}\label{JanetPassivity}
\NF(\automorphism \tilde{f}_i, T, \differenceranking) = 0 \qquad \mbox{for all} \quad
\automorphism \in \overline{\mu}_i = \setauto \setminus \mu_i\,, \quad i = 1, \ldots, s\,.
\end{equation}
\end{definition}
Note that Eqs.~\eqref{JanetPassivity}, in general, form a proper subset of the
Gr\"obner passivity conditions~\eqref{PassivityConditions} (cf.~\cite{G'05}).

\begin{definition}\label{de:differencesimple}
Let $\differenceranking$ be a ranking on $\differencering$, and fix a total ordering on $\setauto$
with respect to which Janet division is defined.
A difference system $\tilde{S}$ as in (\ref{eq:polypartialdifference})
is said to be \emph{simple} (resp., \emph{quasi-simple}) if the following three conditions are satisfied.
\begin{enumerate}
\item $\tilde{S}$ is simple (resp., quasi-simple) as an algebraic system
(in the finitely many indeterminates which occur in the equations and inequations of $\tilde{S}$,
totally ordered by $\differenceranking$; cf.\ Definitions~\ref{de:algebraicsimple} and \ref{de:algebraicquasisimple}).\label{de:differencesimple_a}
\item The difference system $\{ \, \tilde{f}_1 = 0, \, \ldots, \, \tilde{f}_s = 0 \, \}$ is passive.\label{de:differencesimple_b}
\item The left hand sides $\tilde{f}_{s+1}$, \ldots, $\tilde{f}_{s+t}$
are Janet reduced modulo
the passive difference system $\{ \, \tilde{f}_1 = 0, \, \ldots, \, \tilde{f}_s = 0 \, \}$.\label{de:differencesimple_c}
\end{enumerate}
\end{definition}

\begin{theorem}\label{thm:differencesimplemembership}
Let $\tilde{S}$ be a quasi-simple difference system over $\differencering$ as in (\ref{eq:polypartialdifference}).
Let $\tilde{E}$ be the difference ideal of $\differencering$
generated by $\tilde{f}_1$, \ldots, $\tilde{f}_s$
and let $\tilde{Q}$ be the smallest subset of $\differencering$ which is multiplicatively closed,
closed under $\automorphism_1$, \ldots, $\automorphism_n$ and contains the initials $\tilde{q}_i := \init(\tilde{f}_i)$
for all $i = 1$, \ldots, $s$. Then a difference polynomial $\tilde{f} \in \differencering$ is an
element of
\[
\begin{array}{rcl}
\tilde{E} : \tilde{Q} \! & \! = \! & \! \{ \, \tilde{f} \in \differencering \mid (\theta_1(\tilde{q}_1))^{r_1} \ldots (\theta_s(\tilde{q}_s))^{r_s} \, \tilde{f} \in \tilde{E}\\[0.2em]
& &\quad \mbox{for some } \theta_1, \ldots, \theta_s \in \Mon(\setauto), \, r_1, \ldots, r_s \in \Z_{\ge 0} \, \}
\end{array}
\]
if and only if the Janet normal form
of $\tilde{f}$ modulo $\{ \, \tilde{f}_1, \ldots, \tilde{f}_s \, \}$ is zero.
\end{theorem}

\begin{proof}
By definition of $\tilde{E} : \tilde{Q}$, every element $\tilde{f} \in \differencering$
for which Algorithm~\ref{alg:janetreducenonlin}
yields
Janet normal form zero is an element of $\tilde{E} : \tilde{Q}$.

Let $\tilde{f} \in \tilde{E} : \tilde{Q}$, $\tilde{f} \neq 0$.
Then there exist $\tilde{q} \in \tilde{Q}$
and $k_1$, \ldots, $k_s \in \Z_{\ge 0}$
and $\tilde{c}_{i,j} \in \differencering \setminus \{ 0 \}$, $\alpha_{i,j} \in \Mon(\setauto)$, $j = 1$, \ldots, $k_i$,
$i = 1$, \ldots, $s$, such that
\begin{equation}\label{eq:linearcombalpha}
\tilde{q} \, \tilde{f} \, \, = \, \, \sum_{i=1}^s \sum_{j=1}^{k_i} \tilde{c}_{i,j} \, \alpha_{i,j}(\tilde{f}_i)\,.
\end{equation}
Among all pairs $(i, j)$ for which $\alpha_{i,j}$ involves a
non-admissible automorphism for $\tilde{f}_i = 0$ let the pair $(i^{\star}, j^{\star})$
be such that $\alpha_{i^{\star}, j^{\star}}(\ld(\tilde{f}_{i^{\star}}))$ is maximal
with respect to the ranking $\differenceranking$. Let $\automorphism$ be a non-admissible automorphism
for $\tilde{f}_{i^{\star}} = 0$ which divides the monomial $\alpha_{i^{\star},j^{\star}}$.
Since $\{ \tilde{f}_1 = 0, \ldots, \tilde{f}_s = 0 \}$ is a passive difference system, there exist $\tilde{b} \in \tilde{Q}$ and
$l_1$, \ldots, $l_s \in \Z_{\ge 0}$
and $\tilde{d}_{i,j} \in \differencering \setminus \{ 0 \}$
as well as $\beta_{i,j} \in \Mon(\setauto)$, $j = 1$, \ldots, $l_i$, $i = 1$, \ldots, $s$, such that
\[
\tilde{b} \cdot (\automorphism \, \tilde{f}_{i^{\star}}) \, \, = \, \, \sum_{i=1}^s \sum_{j=1}^{l_i}\tilde{d}_{i,j} \, \beta_{i,j}(\tilde{f}_i)\,,
\]
where each $\beta_{i,j}$ involves only admissible automorphisms for $\tilde{f}_i = 0$.
Let $\gamma_{i^{\star}, j^{\star}} := \alpha_{i^{\star}, j^{\star}} / \automorphism$ and
multiply \eqref{eq:linearcombalpha} by $\gamma_{i^{\star}, j^{\star}}(\tilde{b})$ to obtain
\[
\gamma_{i^{\star}, j^{\star}}(\tilde{b}) \cdot \tilde{q} \, \tilde{f} \, \, = \, \,
\sum_{i=1}^s \sum_{j=1}^{k_i} \tilde{c}_{i,j} \cdot \gamma_{i^{\star}, j^{\star}}(\tilde{b}) \cdot \alpha_{i,j}(\tilde{f}_i)\,.
\]
In this equation we replace
\[
\gamma_{i^{\star}, j^{\star}}(\tilde{b}) \cdot \alpha_{i^{\star}, j^{\star}}(\tilde{f}_{i^{\star}}) \, \, = \, \,
\gamma_{i^{\star}, j^{\star}}(\tilde{b} \cdot \automorphism(\tilde{f}_{i^{\star}}))
\]
by
\[
\gamma_{i^{\star}, j^{\star}}\left(
\sum_{i=1}^s \sum_{j=1}^{l_i} \tilde{d}_{i,j} \, \beta_{i,j}(\tilde{f}_i) \right).
\]
Since $\gamma_{i^{\star}, j^{\star}} \, \beta_{i^{\star}, j^{\star}}$
involves fewer non-admissible automorphisms for $\tilde{f}_i = 0$ than $\alpha_{i^{\star}, j^{\star}}$,
iteration of this substitution process will rewrite equation \eqref{eq:linearcombalpha}
in such a way that every $\alpha_{i,j}(\ld(\tilde{f}_i))$ involving non-ad\-missible automorphisms
for $\tilde{f}_i = 0$ will be less than $\alpha_{i^{\star}, j^{\star}}(\ld(\tilde{f}_{i^{\star}}))$
with respect to
$\differenceranking$. A further iteration of this substitution process
will therefore produce an equation as \eqref{eq:linearcombalpha} with no $\alpha_{i,j}$
involving any non-admissible automorphisms for $\tilde{f}_i = 0$.

This shows that for every $\tilde{f} \in (\tilde{E} : \tilde{Q}) \setminus \{ 0 \}$ there
exists a Janet divisor of $\ld(\tilde{f})$ in the passive set defined
by $\tilde{f}_1 = 0$, \ldots, $\tilde{f}_s = 0$.\qed
\end{proof}

\begin{corollary}\label{cor:differencesimpleradical}
In the situation of Theorem~\ref{thm:differencesimplemembership}
let $\tilde{S}$ be simple. Then the difference ideal $\tilde{E} : \tilde{Q}$ is radical.
\end{corollary}

\begin{proof}
Let $\tilde{f} \in \differencering$ and $r \in \N$ be such that $\tilde{f}^r \in \tilde{E} : \tilde{Q}$.
We will show that $\tilde{f} \in \tilde{E} : \tilde{Q}$.
Since $\tilde{f}^r \in \tilde{E} : \tilde{Q}$ and the difference system $\tilde{S}$ is (quasi-) simple, there exist
$\tilde{q} \in \tilde{Q}$, $k_1$, \ldots, $k_s \in \Z_{\ge 0}$, $\tilde{c}_{i,j} \in \differencering$,
$\alpha_{i,j} \in \Mon(\setauto)$, $j = 1$, \ldots, $k_i$, $i = 1$, \ldots, $s$,
such that
\begin{equation}\label{eq:proofradical}
\tilde{q} \, \tilde{f}^r \, \, = \, \, \sum_{i=1}^s \sum_{j=1}^{k_i} \tilde{c}_{i,j} \, \alpha_{i,j}(\tilde{f}_i)\,,
\end{equation}
where each $\alpha_{i,j}$ only involves admissible automorphisms for $\tilde{f}_i = 0$
and where $\tilde{q}$ is a product of powers of $\init(\alpha_{i,j}(\tilde{f}_i))$,
$i = 1$, \ldots, $s$, $j = 1$, \ldots, $k_i$.

Let $V \subset \Mon(\setauto) \tilde{\mathbf{u}}$ be minimal such that
the (non-difference) polynomial algebra $\differencefield[V] \subset \differencering$
contains all indeterminates occurring in \eqref{eq:proofradical}.
Note that $V$ is finite and recall that $\tilde{S}$ is simple as an
algebraic system (cf.\ Definition~\ref{de:algebraicsimple}).
Now define the algebraic system (over $\differencefield[V]$ with $V$ totally ordered by $\differenceranking$)
\[
\tilde{S}' \, \, = \, \,
\{ \, \alpha_{i,j}(\tilde{f}_i) = 0 \mid i = 1, \ldots, s, \, j = 1, \ldots, k_i \, \} \cup
\{ \, \tilde{f}_{s+1} \neq 0, \, \ldots, \, \tilde{f}_{s+t} \neq 0 \, \}\,.
\]
Then $\tilde{S'}$ is simple. In fact, the leaders of all equations and inequations
in $\tilde{S}'$ are pairwise distinct, because the cones of monomials
in $\automorphism_1$, \ldots, $\automorphism_n$ defined by applying admissible
automorphisms to the leaders of $\tilde{f}_1$, \ldots, $\tilde{f}_s$ are disjoint
(cf.\ the discussion before Definition~\ref{de:admissiblederivations}),
and vanishing of $\init(\alpha_{i,j}(\tilde{f}_i)) = \alpha_{i,j}(\init(\tilde{f}_i))$
or $\disc(\alpha_{i,j}(\tilde{f}_i)) = \alpha_{i,j}(\disc(\tilde{f}_i))$
on the solution set of the algebraic system $\tilde{S}'$ is prevented by the
simplicity of $\tilde{S}$.

Let $I_V$ be the (algebraic) ideal of $\differencefield[V]$ that is generated by
the equations in $\tilde{S}'$,
and let $\tilde{q}'$ be the product of the initials of all equations in $\tilde{S}'$.
Then equation~\eqref{eq:proofradical} shows that $\tilde{f}^r \in I_V : (\tilde{q}')^{\infty}$.
Since the algebraic system $\tilde{S}'$ is simple, \cite[Prop.~2.2.7]{Robertz'14} shows
that the ideal $I_V : (\tilde{q}')^{\infty}$ is radical. Hence,
$\tilde{f} \in I_V : (\tilde{q}')^{\infty} \subset \tilde{E} : \tilde{Q}$.\qed
\end{proof}

Let $\Omega \subseteq \R^n$ be open and connected and fix $\mathbf{z} \in \Omega$.
Denoting the grid in \eqref{grid} by $\Gamma_{\mathbf{z},\mathbf{h}}$, we define
\[
\begin{array}{l}
\mathcal{F}_{\Omega,\mathbf{z},\mathbf{h}} \, \, := \, \, \{ \, \tilde{u}\colon \Gamma_{\mathbf{z},\mathbf{h}} \cap \Omega \to \C \mid
\mbox{$\tilde{u}$ is the restriction to $\Gamma_{\mathbf{z},\mathbf{h}} \cap \Omega$ of}\\[0.25em]
\qquad \qquad \qquad \qquad \qquad \qquad \mbox{some locally analytic function $u$ on $\Omega$} \, \}\,,
\end{array}
\]
and for a system $\tilde{S}$ of partial difference equations and inequations as in \eqref{eq:polypartialdifference}
we define the solution set
\[
\begin{array}{l}
\Sol_{\Omega,\mathbf{z},\mathbf{h}}(\tilde{S}) \, \, := \, \, \{ \, \tilde{u} \in \mathcal{F}_{\Omega,\mathbf{z},\mathbf{h}} \mid
\tilde{f}_i(\tilde{u}) = 0, \, \tilde{f}_{s+j}(\tilde{u}) \neq 0 \mbox{ for all}\\[0.25em]
\qquad \qquad \qquad \qquad \qquad \qquad \qquad \qquad i = 1, \ldots, s, \, j = 1, \ldots, t \, \}\,.
\end{array}
\]

\begin{definition}\label{de:differencedecomposition}
Let $\tilde{S}$ be a finite difference system over $\differencering$ and $\differenceranking$ a ranking on $\tilde{\cR}$.
A \emph{difference decomposition} of $\tilde{S}$ (with respect to $\differenceranking$)
is a finite collection of quasi-simple
difference systems $\tilde{S}_1$, \ldots, $\tilde{S}_r$ over $\differencering$ such that
$\Sol_{\Omega,\mathbf{z},\mathbf{h}}(\tilde{S}) = \Sol_{\Omega,\mathbf{z},\mathbf{h}}(\tilde{S}_1) \uplus \ldots \uplus \Sol_{\Omega,\mathbf{z},\mathbf{h}}(\tilde{S}_r)$.
\end{definition}

Given a finite difference system $\tilde{S}$ over $\differencering$,
Algorithm~\ref{alg:differencedecomp}, presented below, constructs a difference decomposition of $\tilde{S}$ in finitely many steps.
In step~11 \algref{Decompose} refers to an algorithm which computes a smallest superset of $\tilde{G} = \{ \tilde{f}_1, \ldots, \tilde{f}_s \}$ in $\tilde{\cR}$ that is Janet complete as defined on page~\pageref{alg:autoreducenonlin} (cf., for example, \cite[Algorithm~2.1.6]{Robertz'14}).

\begin{algorithm}
\DontPrintSemicolon\SetCommentSty{textit}
\KwInput{A finite difference system $\tilde{S}$ over $\differencering$, a ranking $\differenceranking$ on $\differencering$, and a total ordering on $\setauto$ (used by \algref{Decompose})}
\KwOutput{A difference decomposition of $\tilde{S}$}
$Q \gets \{ \tilde{S} \}$; \, $T \gets \emptyset$\;
\Repeat{$Q = \emptyset$}{
  choose $\tilde{L} \in Q$ and remove $\tilde{L}$ from $Q$\;
  compute a decomposition $\{ A_1, \ldots, A_r \}$ of $\tilde{L}$, considered as an algebraic system, into quasi-simple systems (cf.~Definition~\ref{de:algebraicquasisimple})\;
  \For{$i = 1$, \ldots, $r$}{
     \If(\tcp*[f]{no equation and no inequation}){$A_i = \emptyset$}{
      \Return $\{ \emptyset \}$\;
    }
    \Else{
      $(a, \tilde{G}) \gets$ \algref{Auto-reduce}($A_i^{=}$, $\differenceranking$)\tcp*[r]{Algorithm~\ref{alg:autoreducenonlin}}
      \If{$a = \text{\bf true}$}{
        $J \gets$ \algref{Decompose}($\tilde{G}$)\;
        $P \!\gets\! \{ \NF(\automorphism \tilde{f}, J, \differenceranking) \mid (\tilde{f}, \mu) \in J, \, \automorphism \in \overline{\mu} \}$\tcp*[r]{Algorithm~\ref{alg:janetreducenonlin}}
        \If(\tcp*[f]{$J$ is passive}){$P \subseteq \{ 0 \}$}{
          replace each $\tilde{g} \! \neq \! 0$ in $A_i$ with $\NF(\tilde{g}, J, \differenceranking) \! \neq \! 0$\;
          \If{$0 \not\in A_i^{\neq}$}{
            insert $\{ \tilde{f} = 0 \mid (\tilde{f}, \mu) \in J \} \cup \{ \tilde{g} \neq 0 \mid \tilde{g} \in A_i^{\neq} \}$ into $T$\;
          }
        }
        \ElseIf{$P \cap \differencefield \subseteq \{ 0 \}$}{
          insert $\{ \tilde{f} = 0 \mid (\tilde{f}, \mu) \in J \} \cup \{ \tilde{f} = 0 \mid \tilde{f} \in P \setminus \{ 0 \} \} \cup \{ \tilde{g} \neq 0 \mid \tilde{g} \in A_i^{\neq} \}$ into $Q$\;
        }
      }
      \Else{
        insert\! $\{ \tilde{f} = 0 \mid \tilde{f} \in G \} \!\cup\! \{ \tilde{g} \neq 0 \mid \tilde{g} \in A_i^{\neq} \}$\! into $Q$\;
      }
    }
  }
}
\Return $T$\;
\caption{\algref{DifferenceDecomposition}\label{alg:differencedecomp}}
\end{algorithm}

\begin{theorem}
Algorithm~\ref{alg:differencedecomp} terminates and is correct.
\end{theorem}

\begin{proof}
Algorithm~\ref{alg:differencedecomp} maintains a set $Q$ of difference systems that still have to be dealt with.
Given that termination of all subalgorithms has been proved, termination of Algorithm~\ref{alg:differencedecomp}
is equivalent to the condition that $Q = \emptyset$ holds after finitely many steps.

Apart from step~1, new systems are inserted into $Q$ in steps~18 and 20.
We consider the systems that are at some point an element of $Q$ as the vertices of a tree.
The root of this tree is the input system $\tilde{S}$. The systems which are inserted into $Q$ in
steps~18 and 20 are the vertices of the tree whose ancestor is the system $\tilde{L}$ that was
extracted from $Q$ in step~3 which in the following steps produced these new systems.
Since the for loop beginning in step~5 terminates, the degree of each vertex in the tree
is finite. We claim that every branch of the tree is finite, i.e., that the tree has finite
height, hence, that the tree has only finitely many vertices.

In case of step~20 the new system contains an equation which resulted from a non-trivial
difference reduction in step~9. When this new system will be extracted from $Q$ in a later
round, a decomposition into quasi-simple algebraic systems
will be computed in step~4. This may
produce new branches of the tree, but along any of these branches, after finitely many steps
the condition $a = $ true in step~10 will hold, because the order of the shifts in
leaders of the arising equations is bounded by the maximum order of shifts in leaders
of the ancestor system $\tilde{L}$.

In case of step~18 we are going to show that after finitely many steps a difference equation
is obtained whose leader has not shown up as a leader of an equation in any preceding
system in the current branch of the tree. First of all,
the passivity check (step~12) yielded an equation $\tilde{f} = 0$,
$\tilde{f} \in P \setminus \differencefield$, which is Janet reduced modulo $J$. Hence, either
$\ld(\tilde{f})$ is not contained in the multiple-closed set generated by $\ld(\tilde{G})$, or
there exists $(\tilde{f}', \mu') \in J$ such that $\ld(\tilde{f}')$ is a Janet divisor of $\ld(\tilde{f})$,
but the degree of $\tilde{f}$ in $\ld(\tilde{f})$ is smaller than the degree of $\tilde{f}'$ in $\ld(\tilde{f}')$.
In the first case the above claim holds.
The second case cannot repeat indefinitely: First of all, if $\ld(\tilde{f}) = \ld(\tilde{f}')$, then
in a later round, either a pseu\-do-re\-duc\-tion of $\tilde{f}'$ modulo $\tilde{f}$ will be performed if
the initial of $\tilde{f}$ does not vanish, or $\init(\tilde{f}) = 0$ has been added as a new equation
(with lower ranked leader).
Since this leads to a sequence in $\Mon(\setauto)$ which strictly decreases, infinite
chains are excluded in this situation.
If case $\ld(\tilde{f}) \neq \ld(\tilde{f}')$ occurs repeatedly, then a sequence $((\theta_i \, \tilde{u}^{(\alpha)})^{e_i})_{i = 1, 2, 3, \ldots}$
of leaders of newly inserted equations arises,
where $\theta_i \in \Mon(\setauto)$, $\alpha \in \{ 1, \ldots, m \}$,
$e_i \in \Z_{\ge 0}$, such that $e_{i+1} < e_i$ holds (and where also $\theta_i \mid \theta_{i+1}$).
Any such sequence is finite. Hence, the first case arises after finitely many steps.
Therefore, termination follows from Dickson's Lemma.

In order to prove correctness, we note that a difference system is only inserted into $T$ if step~12 confirmed passivity. Such a system is quasi-simple as an algebraic system because (up to auto-reduction in step~9 and Janet completion in step~11) it was returned as one system $A_i$ in step~4.
Condition~(\ref{de:differencesimple_c}) in Definition~\ref{de:differencesimple} is ensured by step~14.
Hence, all difference systems in $T$ are quasi-simple.
Splitting of a system only arises in step~4 by adding an equation $\init(\tilde{f}) = 0$ and the corresponding inequation $\init(\tilde{f}) \neq 0$, respectively, to the two new systems replacing the given one.
Since no solutions are lost or gained, this leads to a partition as required by Definition~\ref{de:differencedecomposition}.\qed
\end{proof}

%
%

\begin{algorithm}
\DontPrintSemicolon\SetCommentSty{textit}
\KwInput{A simple differential system $S$ over $\differentialring$, a differential ranking $\differentialranking$ on $\differentialring$, a difference ranking $>$ on $\differencering$, a total ordering on $\setauto$ (used by \algref{Decompose}) and a difference system $\tilde{S}$ consisting of equations that are w-consistent with $S$}
\KwOutput{$\tilde{L} = \{ (\tilde{L}_1, b_1), \ldots, (\tilde{L}_r, b_r) \}$, where $\tilde{L}_i$ is s-consistent (resp.\ w-consistent) with $L_i\xleftarrow[|\mathbf{h}|\rightarrow 0]{}  \tilde{L}_i$ if $b_i = \text{\bf true}$ (resp.\ {\bf false})}
$\tilde{L}=\{ \tilde{L}_1, \ldots, \tilde{L}_k \} \gets \algref{DifferenceDecomposition}(\tilde{S}, >)$\;
\For{$i = 1$, \ldots, $k$}{
    \If(\tcp*[f]{Definition~\ref{Implication}}){$\exists \tilde{f}\in \tilde{L}_i^{\neq}$ such that $\tilde{f}\rhd \mathcal{F}$ with $\mathcal{F} \cap \llbracket S^{=} \rrbracket \neq \emptyset$}{
      $\tilde{L} \gets \tilde{L}\setminus \{ \tilde{L}_i \}$\;
    }
    \Else{
      $b_i \gets \text{\bf true}$\;
      \For{$\tilde{f}\in \tilde{L}^{=}$}{
         compute $\mathcal{F} \subset \differentialring$ such that $\tilde{f} \rhd \mathcal{F}$\tcp*[r]{Remark~\ref{ContLim}}
         \If{$\exists f \in \mathcal{F}$ such that $\NF(f,S^{=},\differentialranking)\neq 0$}{
             $b_i \gets \text{\bf false}$; \, {\bf break}\;
         }
      }
    }
}
\Return $\{ \, (\tilde{L}_i, b_i) \mid \tilde{L}_i \in \tilde{L} \, \}$\;
\caption{\algref{S-ConsistencyCheck}\label{alg:discretization}}
\end{algorithm}

Given a simple differential system and its w-consistent discretization on the regular grid~\eqref{grid},  Algorithm 5 allows to verify strong consistency of the latter.

{\it Correctness} of the algorithm follows from Definition~\ref{def-wcons} (extended to inequations),
Definition~\ref{def-scon}
and passivity of the subsystems returned by Algorithm~\ref{alg:differencedecomp}. Their solution spaces partition the solution space of the input FDA. Thereby,
any subsystem $\tilde{L}_i$ in the output with $b_i=$ {\bf true} is s-consistent with
$L_i$, where
\[
\tilde{L}_i\xrightarrow[|\mathbf{h}|\rightarrow 0]{}  {L}_i
\]
and w-consistent if $b_i =$ {\bf false}.
If $b_i =$ {\bf true} for all $i$, then $\tilde{S}$ is s-consistent with $S$.
{\it Termination} follows from that of the subalgorithms.

%
%
\section{Examples of quasi-linear systems}\label{sec:examples}

In this section we consider two systems of quasi-linear PDEs for unknown functions of two independent variables.

\begin{example}\label{ex:Example1}
Let us consider the overdetermined PDE system
\begin{equation}\label{eq:PDEsystem1}
\left\{ \begin{array}{rcl}
\displaystyle u_x - u^2 & = & 0\,,\\[0.5em]
\displaystyle u_y + u^2 & = & 0\,,
\end{array} \right.
\qquad \qquad u = u(x,y)\,.
\end{equation}
Since the cross-derivative
$\partial_y (u_{x} - u^2) - \partial_x (u_{y} + u^2)$
reduces to zero modulo the given equations, the differential system \eqref{eq:PDEsystem1} is simple
(with respect to any ranking).
The exact general solution of \eqref{eq:PDEsystem1} can easily be found with Maple:
\begin{equation}\label{exact:sol}
u(x, y) = \frac{1}{y - x + C}\,,
\end{equation}
where $C$ is an arbitrary constant (the corresponding
counting polynomial \cite{LangeHegermann} being $\infty$).
For the numerical comparison of the following finite difference approximations we consider the
domain $[0, 10] \times [0, 10]$ with Cartesian grid defined
by $h_1 = h_2 = h = 1/5$ and we shall let $C = 12$. The error will be
computed as
\[
\left(e_{g}\right)_{j,k} := \frac{|g_{j,k}-g(x_j,y_k)|}{1+|g(x_j,y_k)|}\,, \quad \text{maximum error}:=\max_{j,k}\left(e_{g}\right)_{j,k}\,, \
\]
where $g$ is the exact solution~\eqref{exact:sol}.

\medskip

We investigate the system of difference equations
\begin{equation}\label{eq:FDA1}
\left\{ \begin{array}{rcl}
\displaystyle \frac{\tilde{u}_{i+1,j} - \tilde{u}_{i,j}}{h_1} - \tilde{u}_{i,j}^2 & = & 0\,,\\[1em]
\displaystyle \frac{\tilde{u}_{i,j+1} - \tilde{u}_{i,j}}{h_2} + \tilde{u}_{i,j}^2 & = & 0\,,
\end{array} \right.
\end{equation}
which is obtained as discretization of (\ref{eq:PDEsystem1}) by replacing
$u_x$ and $u_y$ by the corresponding forward differences,
with step sizes $h_1$ and $h_2$, respectively.
For simplicity we ignore case distinctions and pursue the generic case only.
We fix an orderly ranking on the difference polynomial ring $\differencering = \differencefield\{ \tilde{u} \}$
with automorphisms $\automorphism_1$, $\automorphism_2$.

Denote by $\tilde{f}_1$ and $\tilde{f_2}$
the left hand sides in \eqref{eq:FDA1}.
Then \eqref{eq:FDA1} is simple as an algebraic system, but the
passivity check (cf.~Definition~\ref{df:PassivityConditions}) reveals the consequence
\[
\begin{array}{rcl}
\displaystyle h_2 \, \automorphism_1 \tilde{f}_2 - h_1 \, \automorphism_2 \tilde{f}_1
& - & \displaystyle
\left( h_2 \, \tilde{u}_{i+1,j} + h_1 \, h_2 \, \tilde{u}_{i,j}^2 + h_2 \, \tilde{u}_{i,j} - 1 \right) h_1 \, \tilde{f}_1\\[0.5em]
& - & \displaystyle \left( h_1 \, \tilde{u}_{i,j+1} - h_1 \, h_2 \, \tilde{u}_{i,j}^2 + h_1 \, \tilde{u}_{i,j} + 1 \right) h_2 \, \tilde{f}_2\\[0.5em]
\multicolumn{3}{l}{\, \, = \, \, h_1 \, h_2 \left( h_1 + h_2 \right) \tilde{u}_{i,j}^4\,.}
\end{array}
\]
Note that $h_1 \, h_2 \, (h_1 + h_2)$ is non-zero. By adding $\tilde{f}_3 := \tilde{u}_{i,j}^4$ to system (\ref{eq:FDA1})
we obtain the quasi-simple difference system $\{ \tilde{f}_1 = 0, \, \tilde{f}_2 = 0, \, \tilde{f}_3 = 0 \}$.
The continuous limit of $\tilde{f}_3$ for $h_1 \to 0$, $h_2 \to 0$
is the differential polynomial $u^4$, which is not in the radical differential ideal corresponding to \eqref{eq:PDEsystem1}.
Hence, \eqref{eq:FDA1} is not s-consistent with \eqref{eq:PDEsystem1}.

Figure~\ref{fig4} shows the error computed for FDA~\eqref{eq:FDA1} relative to the exact solution~\eqref{exact:sol}. Hereafter, in our numerical computation we chose the grid spacings as $h_1=h_2=h=0.2$. The computed error takes maximum value at the point $(x=10,y=0)$, which is closest to the pole in~\eqref{exact:sol}.  Because of s-inconsistency of this discretization we did not compute the associated \emph{modified PDE} (cf.\ \cite{ZhangGerdtBlinkov}). In general terms a modified PDE for a given FDA is one that a numerical solution to FDA satisfies to a higher accuracy than the initial PDE (see, for example, the textbooks~\cite{Moin:2010}, Sections~5.5--5.6 and~\cite{Thomas:ME}, Section~7.7). The method of modified equation provides a useful tool for evaluating such important properties of finite difference schemes as order of accuracy, consistency, stability, dissipation and dispersion.

\begin{figure}
\begin{center}
\includegraphics[width=1.0\textwidth]{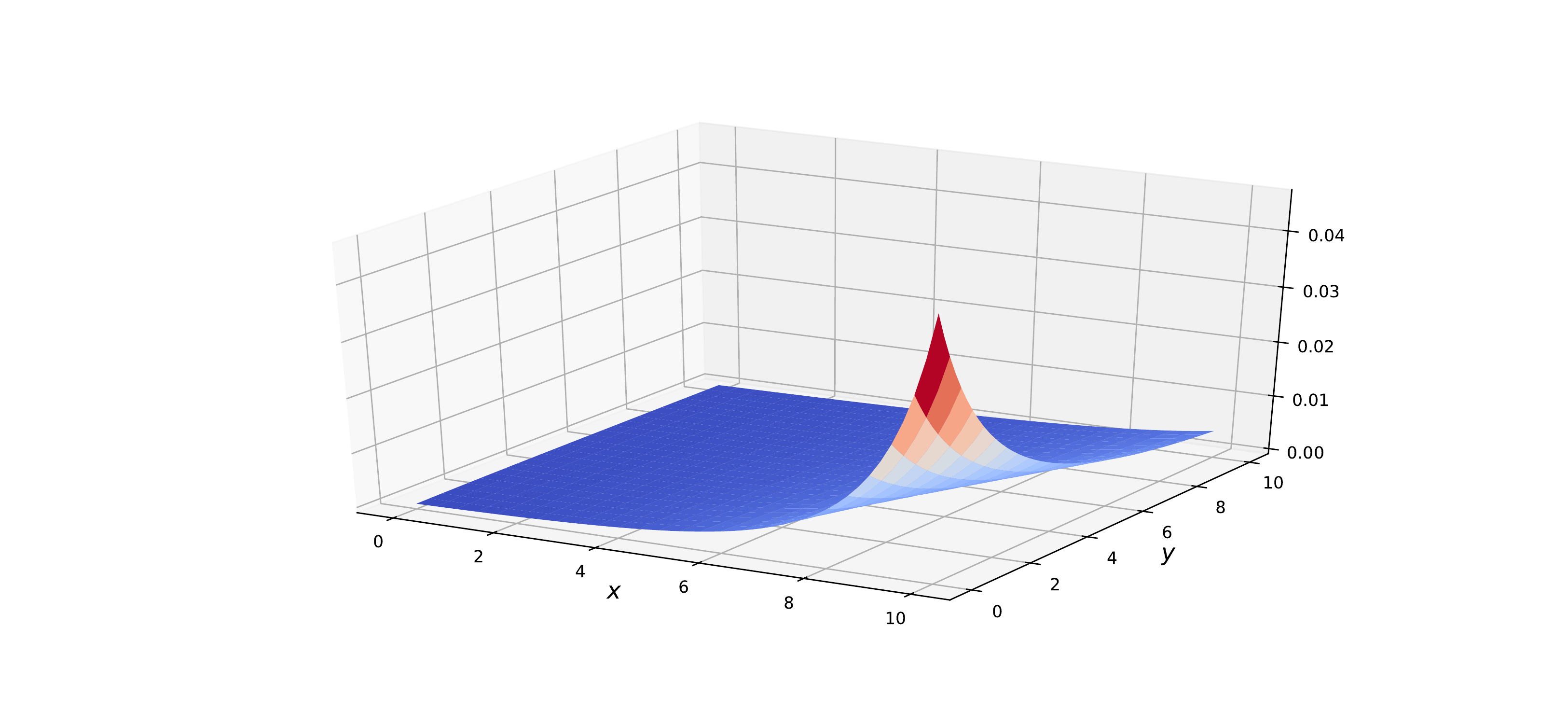}
\end{center}
\caption{Computed error for discretization \eqref{eq:FDA1}\label{fig4}, maximum error = 0.046222672401600905}
\end{figure}

\medskip

Next we consider the discretization obtained by replacing
$u_x$ by the forward difference as before and $u_y$ by the backward difference:
\begin{equation}\label{eq:FDA2}
\left\{
\begin{array}{lcl}
\displaystyle \frac{\tilde{u}_{i+1,j} - \tilde{u}_{i,j}}{h_1} - \tilde{u}_{i,j}^2 & = & 0\,,\\[1em]
\displaystyle \frac{\tilde{u}_{i,j+1} - \tilde{u}_{i,j}}{h_2} + \tilde{u}_{i,j+1}^2 & = & 0\,,
\end{array}
\right.
\end{equation}
again with step sizes $h_1$ and $h_2$, respectively. Denote by $\tilde{f}'_1$ and $\tilde{f}'_2$
the left hand sides in (\ref{eq:FDA2}).
The passivity check reveals the consequence (with underlined leader)
\[
\begin{array}{rcl}
\!\!\!\!\!\! & & \displaystyle
h_2^3 \left( \automorphism_1 \tilde{f}'_2 - h_1 \tilde{u}_{i+1,j+1} \, \automorphism_2 \tilde{f}'_1 \right)
- \left( h_1 h_2 \tilde{u}_{i,j+1}^2 + h_2 \, \tilde{u}_{i,j+1} + 1 \right) h_1 \, h_2^2 \, \automorphism_2 \tilde{f}'_1 \\[0.5em]
\!\!\!\!\!\! & & - \left( h_1^2 \, h_2^2 \, \tilde{u}_{i,j+1}^2 - h_1 \, h_2 \, (h_1 - 2 \, h_2) \, \tilde{u}_{i,j+1} +
h_1^2 \, h_2 \, \tilde{u}_{i,j} + h_1^2 - h_1 \, h_2 + h_2^2 \right) h_2 \, \tilde{f}'_2\\[0.5em]
\!\!\!\!\!\! & = & - h_1 \left( h_1 - h_2 \right)
\left( \left( 2 \, h_2 \, \tilde{u}_{i,j} + 1 \right)
\underline{\tilde{u}_{i,j+1}} - h_2 \, \tilde{u}_{i,j}^2 - \tilde{u}_{i,j} \right).
\end{array}
\]
The continuous limit of this difference polynomial $\tilde{f}'_3$ is given by
\[
\left( u_y + u^2 \right) h_1 \, h_2^2 - \left( u_y + u^2 \right) h_1^2 \, h_2\,.
\]
Now a pseudo-reduction of $\tilde{f}'_2$ modulo $\tilde{f}'_3$ yields
\[
\begin{array}{rlcl}
& \displaystyle
h_1 \, h_2 \left( h_1 - h_2 \right) \left( 2 \, h_2 \, \tilde{u}_{i,j} + 1 \right) \tilde{f}'_2\\[0.5em]
+ & \displaystyle
\left( h_2 \, (2 \, h_2 \, \tilde{u}_{i,j} + 1) \, \tilde{u}_{i,j+1} +
h_2^2 \, \tilde{u}_{i,j}^2 + 3 \, h_2 \, \tilde{u}_{i,j} + 1 \right) \tilde{f}'_3
& = & \displaystyle
h_1 \, h_2^3 \left( h_1 - h_2 \right) \tilde{u}_{i,j}^4\,.
\end{array}
\]
For $h_1 \neq h_2$ we define $\tilde{f}'_4 := \tilde{u}_{i,j}^4$ and obtain
the quasi-simple difference system $\{ \tilde{f}'_1 = 0, \, \tilde{f}'_2 = 0, \, \tilde{f}'_4 = 0 \}$,
and we conclude that (\ref{eq:FDA2}) is not s-consistent with (\ref{eq:PDEsystem1}).
However, if $h_1 = h_2$, then (\ref{eq:FDA2}) is a simple difference system
and it is s-con\-sis\-tent with \eqref{eq:PDEsystem1}.

\begin{figure}
\begin{center}
\includegraphics[width=1.0\textwidth]{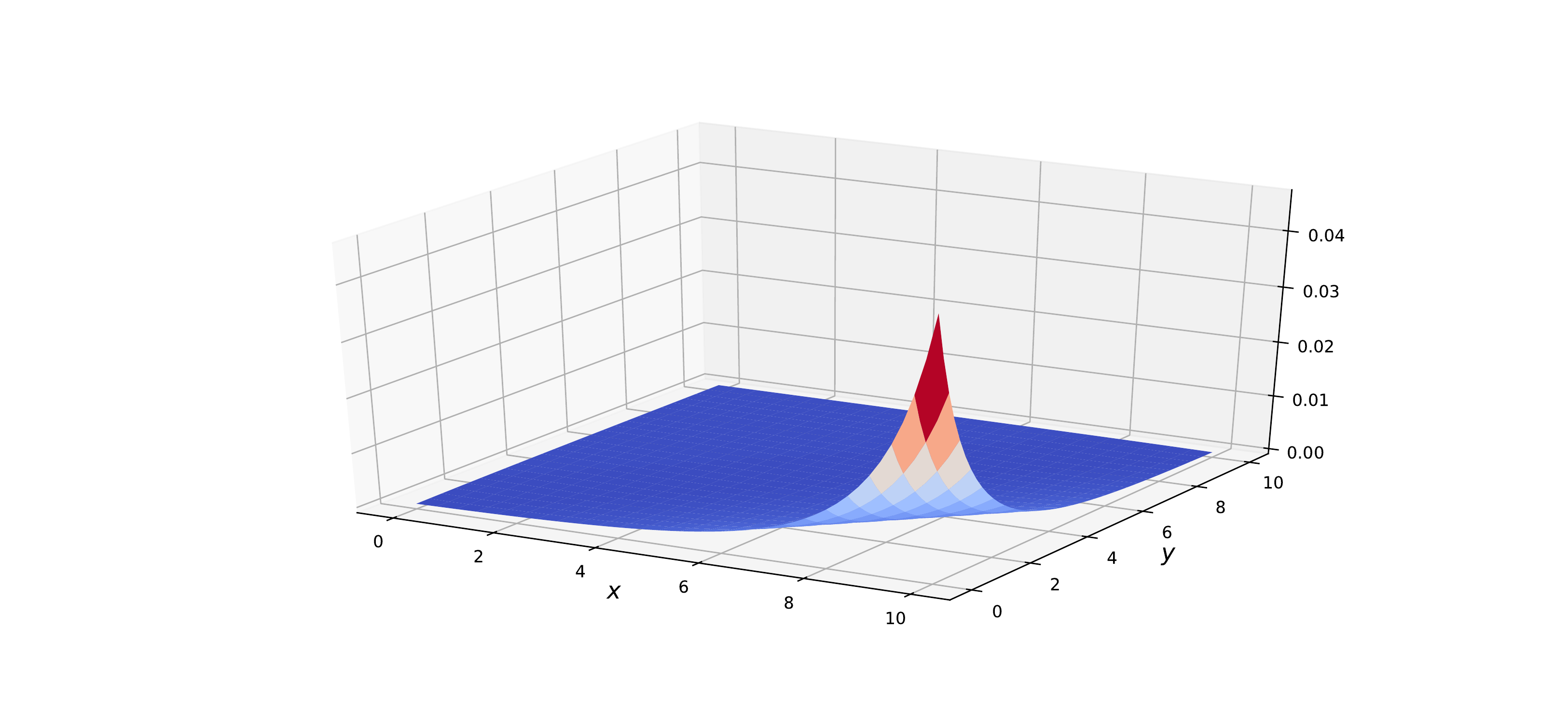}
\end{center}
\caption{Computed error for discretization~\eqref{eq:FDA2}\label{fig5}, maximum error = 0.038621319610305495}
\end{figure}





Now we perform the Taylor expansion of the left-hand sides in Eqs.~\eqref{eq:FDA2} with $h_1=h_2=h$ and explicitly write the first-order terms
\begin{equation}\label{TaylorExpansion:56}
\left\{
\begin{array}{l}
f_1:=\displaystyle u_x-u^2 + \frac{u_{x,x}}{2}h + \mathcal{O}(h^2)=0           \,,\\[1em]
f_2:=\displaystyle u_y +u^2 +\left(2uu_y+\frac{u_{y,y}}{2}\right)h + \mathcal{O}(h^2)=0\,.
\end{array}
\right.
\end{equation}
From Eqs.~\eqref{TaylorExpansion:56} we obtain the modified PDE for scheme~\eqref{eq:FDA2}   
\begin{equation}\label{ModifiedEquation:56}
\left\{
\begin{array}{l}
f_1- \displaystyle \left(\frac{1}{2}(f_1)_x + uf_1\right)h  = \displaystyle u_x-u^2 + u^3h + \mathcal{O}(h^2)=0           \,,\\[1em]
f_2- \displaystyle \left(\frac{1}{2}(f_2)_y + uf_2\right)h  = u_y +u^2 - u^3h + \mathcal{O}(h^2)=0\,,
\end{array}
\right.
\end{equation}
which shows that scheme~\eqref{eq:FDA2} has first order accuracy. Furthermore, Eqs.~\eqref{ModifiedEquation:56} allow an obvious modification (cf.~\cite{Samarskii'01}, p.~80) of FDA~\eqref{eq:FDA2} to one with second order accuracy given by
\begin{equation}\label{eq:FDA2_mod}
\left\{
\begin{array}{lcl}
\displaystyle \frac{\tilde{u}_{i+1,j} - \tilde{u}_{i,j}}{h} - \tilde{u}_{i,j}^2 -h\tilde{u}_{i,j}^3 & = & 0\,,\\[1em]
\displaystyle \frac{\tilde{u}_{i,j+1} - \tilde{u}_{i,j}}{h} + \tilde{u}_{i,j+1}^2 +h\tilde{u}_{i,j}^3 & = & 0\,.
\end{array}
\right.
\end{equation}

The corresponding decrease of numerical error for FDA~\eqref{eq:FDA2_mod} in comparison with FDA~\eqref{eq:FDA2} (Fig.~\ref{fig5}) is shown in Fig.\,\ref{fig6}.
\begin{figure}
\begin{center}
\includegraphics[width=1.0\textwidth]{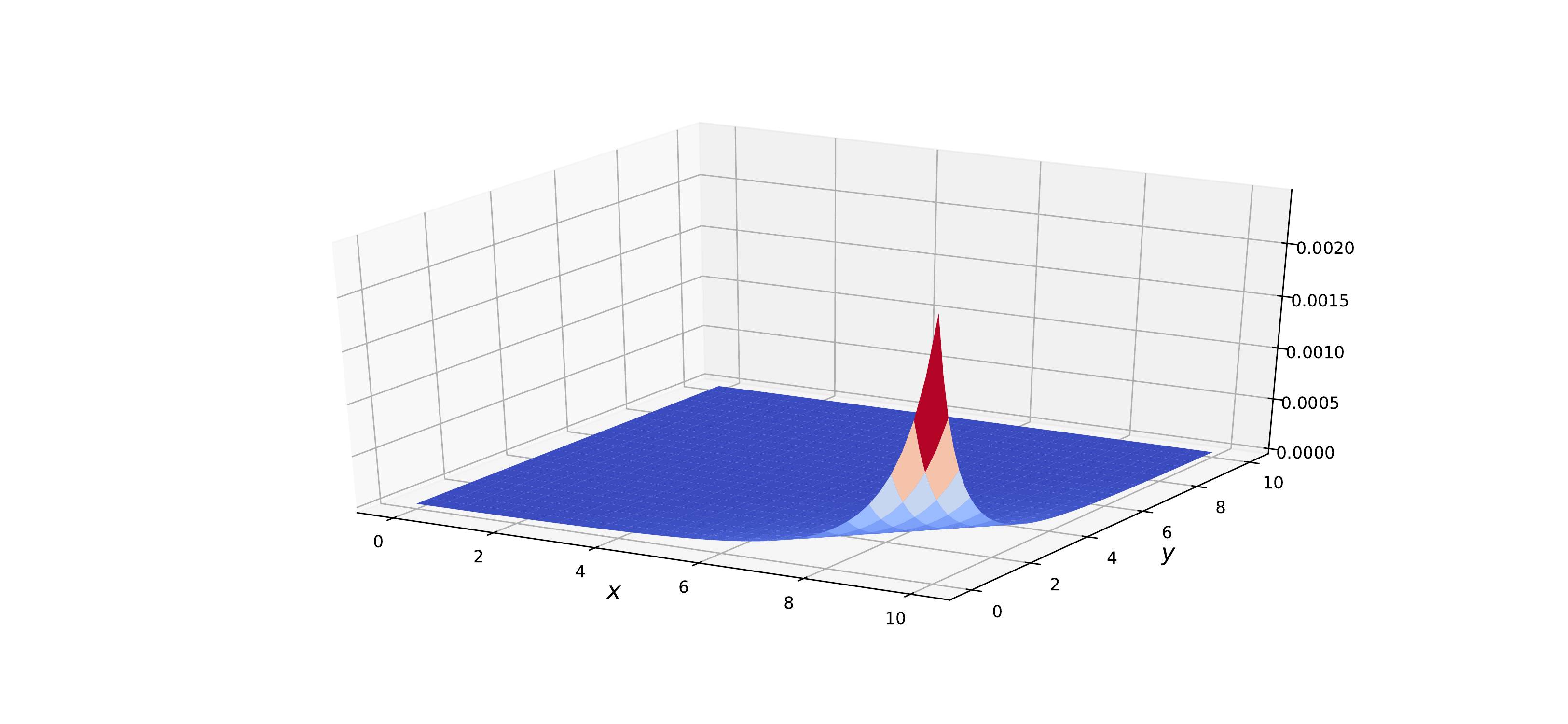}
\end{center}
\caption{Computed error for discretization~\eqref{eq:FDA2_mod}\label{fig6}, maximum error = 0.002446111801807538}
\end{figure}

Now we consider the discretization obtained by replacing both
$u_{x}$ and $u_y$ by central differences:
\begin{equation}\label{eq:FDA3}
\left\{
\begin{array}{rcl}
\displaystyle \frac{\tilde{u}_{i+2,j} - \tilde{u}_{i,j}}{2 \, h_1} - \tilde{u}_{i+1,j}^2 & = & 0\,,\\[1em]
\displaystyle \frac{\tilde{u}_{i,j+2} - \tilde{u}_{i,j}}{2 \, h_2} + \tilde{u}_{i,j+1}^2 & = & 0\,,
\end{array}
\right.
\end{equation}
with step sizes $h_1$ and $h_2$, respectively. Denote by $\tilde{f}''_1$ and $\tilde{f}''_2$
the left hand sides in \eqref{eq:FDA3}.
The passivity check yields the consequence (with underlined leaders)
\[
\begin{array}{rcl}
& & \! \displaystyle
\Big( \, h_2 \, \automorphism_1^2 \tilde{f}''_2 - h_1 \, \automorphism_2^2 \tilde{f}''_1
- 2 \, h_1 \, h_2 \left( \tilde{u}_{i+2,j+1} + 2 \, h_1 \, \tilde{u}_{i+1,j+1}^2 + \tilde{u}_{i,j+1} \right) \automorphism_2 \tilde{f}''_1\\[0.5em]
& & \quad \quad
- 2 \, h_1 \, h_2 \left( \tilde{u}_{i+1,j+2} - 2 \, h_2 \, \tilde{u}_{i+1,j+1}^2 + \tilde{u}_{i+1,j} \right) \automorphism_1 \tilde{f}''_2
+ h_1 \, \tilde{f}''_1 - h_2 \, \tilde{f}''_2 \, \Big) \, / \, 4\\[0.5em]
& = &
h_1 \, h_2 \, \underline{\tilde{u}_{i+1,j+1}^2} \left( (h_1 + h_2) \,
\underline{\tilde{u}_{i+1,j+1}^2} - \tilde{u}_{i+1,j} + \tilde{u}_{i,j+1} \right).
\end{array}
\]
The continuous limit of this difference polynomial $\tilde{f}''_3$ is given by
\[
u^2 \left( u_y + u^2 \right) h_1 \, h_2^2 - u^2 \left( u_x - u^2 \right) h_1^2 \, h_2\,.
\]
Further passivity checks yield the difference polynomial $\tilde{f}''_4$ defined by
\[
\begin{array}{rcl}
\displaystyle \automorphism_1 \tilde{f}''_3 & - & \displaystyle
2 \, h_1^2 \, h_2 \, \Big[ \, (h_1 + h_2) \, \Big( \,
\tilde{u}_{i+2,j+1}^2 + 4 \, h_1^2 \, \tilde{u}_{i+1,j+1}^4 +
4 \, h_1 \, \tilde{u}_{i,j+1} \, \tilde{u}_{i+1,j+1}^2\\[0.75em]
& + & \displaystyle
\tilde{u}_{i,j+1}^2 \, \Big)
+ \tilde{u}_{i+1,j+1} - \tilde{u}_{i+2,j} \, \Big]
\, (\tilde{u}_{i+2,j+1} + 2 \, h_1 \, \tilde{u}_{i+1,j+1}^2 + \tilde{u}_{i,j+1}) \,
\automorphism_2 \tilde{f}''_1\\[0.75em]
& + & \displaystyle
2 \, h_1^2 \, h_2 \left( 2 \, h_1 \tilde{u}_{i+1,j+1}^2 + \tilde{u}_{i,j+1} \right)^2 \tilde{f}''_1
\, - \, 4 \, h_1^2 \left( 2 \, h_1 \, \tilde{u}_{i+1,j+1}^2 + \tilde{u}_{i,j+1} \right)^2 \tilde{f}''_3\\[0.75em]
\multicolumn{3}{l}{\, \, = \, \, h_1 \, h_2
\left( 2 \, h_1 \, \underline{\tilde{u}_{i+1,j+1}^2} + \tilde{u}_{i,j+1} \right)^2
\Big( \, 4 \, h_1 \, (h_1 \, \tilde{u}_{i+1,j} + h_2 \, \tilde{u}_{i,j+1}) \, \underline{\tilde{u}_{i+1,j+1}^2} \, +}\\[0.75em]
\multicolumn{3}{l}{\qquad \qquad \qquad \qquad \qquad \qquad \quad
\underline{\tilde{u}_{i+1,j+1}} + (h_1 + h_2) \, \tilde{u}_{i,j+1}^2 - 2 \, h_1 \, \tilde{u}_{i+1,j}^2 - \tilde{u}_{i,j} \, \Big)\,,}
\end{array}
\]
as well as the difference polynomial $\tilde{f}''_5$ defined by
\[
\begin{array}{rcl}
\displaystyle \automorphism_2 \tilde{f}''_3
& - & \displaystyle
2 \, h_1 \, h_2^2 \, \Big[ \, (h_1 + h_2) \, \Big( \,
\tilde{u}_{i+1,j+2}^2 + \, 4 \, h_2^2 \, \tilde{u}_{i+1,j+1}^4 -
4 \, h_2 \, \tilde{u}_{i+1,j} \, \tilde{u}_{i+1,j+1}^2\\[0.75em]
& + & \displaystyle
\tilde{u}_{i+1,j}^2 \, \Big)
- \tilde{u}_{i+1,j+1} + \tilde{u}_{i,j+2} \, \Big]
\, (\tilde{u}_{i+1,j+2} - 2 \, h_2 \, \tilde{u}_{i+1,j+1}^2 + \tilde{u}_{i+1,j}) \,
\automorphism_1 \tilde{f}''_2\\[0.75em]
& - & \displaystyle
2 \, h_1 \, h_2^2 \left( 2 \, h_2 \, \tilde{u}_{i+1,j+1}^2 - \tilde{u}_{i+1,j} \right)^2 \tilde{f}''_2
\, - \, 4 \, h_2^2 \left( 2 \, h_2 \, \tilde{u}_{i+1,j+1}^2 - \tilde{u}_{i+1,j} \right)^2 \tilde{f}''_3\\[0.75em]
\multicolumn{3}{l}{\, \, = \, \, h_1 \, h_2
\left( 2 \, h_2 \, \underline{\tilde{u}_{i+1,j+1}^2} - \tilde{u}_{i+1,j} \right)^2
\Big( \, 4 \, h_2 \, (h_1 \, \tilde{u}_{i+1,j} + h_2 \, \tilde{u}_{i,j+1}) \, \underline{\tilde{u}_{i+1,j+1}^2} \, +}\\[0.75em]
\multicolumn{3}{l}{\qquad \qquad \qquad \qquad \qquad \qquad \quad
\underline{\tilde{u}_{i+1,j+1}} - (h_1 + h_2) \, \tilde{u}_{i+1,j}^2 + 2 \, h_2 \, \tilde{u}_{i,j+1}^2 - \tilde{u}_{i,j} \, \Big)\,.}
\end{array}
\]
The continuous limit of the difference polynomial $\tilde{f}''_4 + \tilde{f}''_5$
is given by
\[
2 \, u \, u_y \, (u_y + 2 \, u^2) \, h_1 h_2^3 +
4 \, u^3 \, (u_x + u_y) \, h_1^2 h_2^2
-2 \, u \, u_x \, (u_y - 2 \, u^2) \, h_1^3 h_2\,,
\]
whose Janet normal form modulo \eqref{eq:PDEsystem1} is
\begin{equation}\label{EqualSpacings}
2 \, h_1 \, h_2 \, (h_1 - h_2) \, (h_1 + h_2) \, u^5\,.
\end{equation}
Since the differential polynomial $u^5$ is not in the
radical differential ideal corresponding to \eqref{eq:PDEsystem1},
we conclude that \eqref{eq:FDA3} is not s-consistent with \eqref{eq:PDEsystem1}
unless $h_1 = h_2$.

\medskip

Let $h_1 = h_2 = h$ in \eqref{eq:FDA3}.
If one handles this case along the same lines as above, one may encounter
an enormous growth of expressions. We demonstrate here how to benefit from
applying inverse shifts to difference polynomials when possible,
i.e., when no indeterminates with negative shifts are introduced by this process.
Note that the perfect closure of the difference ideal $\differenceideal$ generated by (\ref{eq:FDA3})
contains the reflexive closure of $\differenceideal$, i.e., all $\tilde{f} \in \differencering$ such
that $\automorphism \tilde{f} \in \differenceideal$ for some $\automorphism \in \Mon(\setauto)$.

Denote that left hand sides of \eqref{eq:FDA3}, for $h_1 = h_2 = h$,
again by $\tilde{f}''_1$ and $\tilde{f}''_2$. Similarly to the previous discussion,
the passivity check yields a difference polynomial
\[
\tilde{f}''_3 \, \, := \, \, \underline{\tilde{u}_{i+1,j+1}^2} \left( 2 \, h \,
\underline{\tilde{u}_{i+1,j+1}^2} - \tilde{u}_{i+1,j} + \tilde{u}_{i,j+1} \right).
\]
The difference polynomial
\[
\tilde{f}''_3 + 2 \, h \, \tilde{u}_{i+1,j+1}^2 \, (\automorphism_2 \tilde{f}''_1 - \automorphism_1 \tilde{f}''_2)
\, = \, \tilde{u}_{i+1,j+1}^2 \, (\tilde{u}_{i+2,j+1} - \tilde{u}_{i+1,j+2} - 2 \, h \, \tilde{u}_{i+1,j+1}^2)
\]
can be shifted back by one step in each of the two grid directions, producing
\[
\tilde{f}''_6 \, \, := \, \,
\tilde{u}_{i,j}^2 \, (\underline{\tilde{u}_{i+1,j}} - \tilde{u}_{i,j+1} - 2 \, h \, \tilde{u}_{i,j}^2)\,.
\]
The continuous limit of this difference polynomial is given by
\[
u^2 \left( u_x - u_y - 2 \, u^2 \right) h\,.
\]
Passivity checks yield
\[
\tilde{f}''_7 \, \, := \, \,
\automorphism_1 \tilde{f}''_6 - 2 \, h \, \tilde{u}_{i+1,j}^2 \, \tilde{f}''_1 \, \, = \, \,
-\tilde{u}_{i+1,j}^2 \, (\underline{\tilde{u}_{i+1,j+1}} - \tilde{u}_{i,j})
\]
and
\[
\tilde{f}''_8 \, \, := \, \,
\automorphism_2 \tilde{f}''_6 + 2 \, h \, \tilde{u}_{i,j+1}^2 \, \tilde{f}''_2 \, \, = \, \,
\tilde{u}_{i,j+1}^2 \, (\underline{\tilde{u}_{i+1,j+1}} - \tilde{u}_{i,j})\,,
\]
whose continuous limits are given by
\[
\pm u^2 \, (u_x + u_y) \, h\,.
\]
We obtain the decomposition into simple difference systems
\[
\left\{ \begin{array}{rcl}
\tilde{f}''_8 \, \, = \, \, \tilde{u}_{i,j+1}^2 \, (\underline{\tilde{u}_{i+1,j+1}} - \tilde{u}_{i,j}) & = & 0\,,\\[0.5em]
\tilde{f}''_6 \, \, = \, \, \tilde{u}_{i,j}^2 \, (\underline{\tilde{u}_{i+1,j}} - \tilde{u}_{i,j+1} - 2 \, h \, \tilde{u}_{i,j}^2) & = & 0\,,\\[0.5em]
\underline{\tilde{u}_{i,j}} & \neq & 0\,,
\end{array} \right.
\qquad \vee \qquad
\left\{ \begin{array}{rcl}
\phantom{x}\\[0.2em]
\tilde{u}_{i,j} & = & 0\,,\\[0.2em]
\phantom{x}
\end{array}\right.
\]
the first one confirming s-consistency of (\ref{eq:FDA3}) with (\ref{eq:PDEsystem1})
provided $h_1 = h_2 = h$.

Now we present our numerical experiments with difference equations in FDA \eqref{eq:FDA3} provided $h_1=h_2=h$. Again we perform the Taylor expansion of their left-hand sides up to terms of order $h^2$
\begin{equation}\label{TaylorExpansion:60}
\left\{
\begin{array}{lcl}
g_1:&=&\displaystyle u_x-u^2 + \left(u_{x,x}+u_{x,y}-2uu_x-2uu_y\right)h  \\[0.5em]
 && -\left((u_{x,x}+2u_{x,y}+u_{y,y})u +(u_x+u_y)^2 \right)h^2  \\[0.6em]
 && +\left(\frac{2}{3}u_{x,x,x} + u_{x,x,y} +\frac{1}{2}u_{x,y,y} \right)h^2      +  \mathcal{O}(h^3)=0           \,,\\[1em]
g_2:&=&\displaystyle u_y+u^2 + \left(u_{y,y}+u_{x,y}+2uu_x+2uu_y\right)h  \\[0.5em]
 && +\left((u_{x,x}+2u_{x,y}+u_{y,y})u +(u_x+u_y)^2 \right)h^2  \\[0.6em]
 && +\left(\frac{1}{2}u_{x,x,y} + u_{x,y,y} +\frac{2}{3}u_{y,y,y} \right)h^2      +  \mathcal{O}(h^3)=0          \,.
\end{array}
\right.
\end{equation}
Based on Eqs.~\eqref{TaylorExpansion:60}, the modified PDE for scheme~\eqref{eq:FDA3} with $h_1=h_2=h$ reads
\begin{equation}\label{ModifiedEquation:60}
\left\{
\begin{array}{lll}
g_1  &-& \displaystyle  (g_1)_xh + \left(\frac{1}{3}(g_1)_{x,x} - \frac{1}{3}(g_1)_xu
- \frac{1}{3}g_1u_x - g_1u^2h^2 \right)h^2 \\[1.em]
 &=& u_x-u^2 + u^4h^2 + \mathcal{O}(h^4)=0           \,,\\[1em]
g_2  &-& \displaystyle  (g_2)_yh + \left(\frac{1}{3}(g_2)_{y,y} + \frac{1}{3}(g_2)_yu
  + \frac{1}{3}g_2u_y - g_2u^2h^2 \right)h^2 \\[1.em]
 &=& u_y+u^2 - u^4h^2 + \mathcal{O}(h^4)=0\,.
\end{array}
\right.
\end{equation}
Thus, the scheme~\eqref{eq:FDA3} has second order accuracy, and the last can be increased to fourth order as follows:
\begin{equation}\label{eq:FDA3_mod}
\left\{
\begin{array}{rcl}
\displaystyle \frac{\tilde{u}_{i+2,j} - \tilde{u}_{i,j}}{2 \, h_1} - \tilde{u}_{i+1,j}^2 - h^2\tilde{u}_{i+1,j}^4 & = & 0\,,\\[1em]
\displaystyle \frac{\tilde{u}_{i,j+2} - \tilde{u}_{i,j}}{2 \, h_2} + \tilde{u}_{i,j+1}^2 + h^2\tilde{u}_{i,j+1}^4 & = & 0\,.
\end{array}
\right.
\end{equation}
The numerical behavior of schemes~\eqref{eq:FDA3} and \eqref{eq:FDA3_mod} in the above described initial value problem for~\eqref{eq:PDEsystem1} with the initial data defined by the exact solution~\eqref{exact:sol} is shown in Fig.~\ref{fig7} and Fig.~\ref{fig8}, respectively. One can see that the experimental numerical accuracy is scaled in accordance with the theoretical accuracy $h^2$ for~\eqref{eq:FDA3} and $h^4$ for \eqref{eq:FDA3_mod}.

\begin{figure}
\begin{center}
\includegraphics[width=1.0\textwidth]{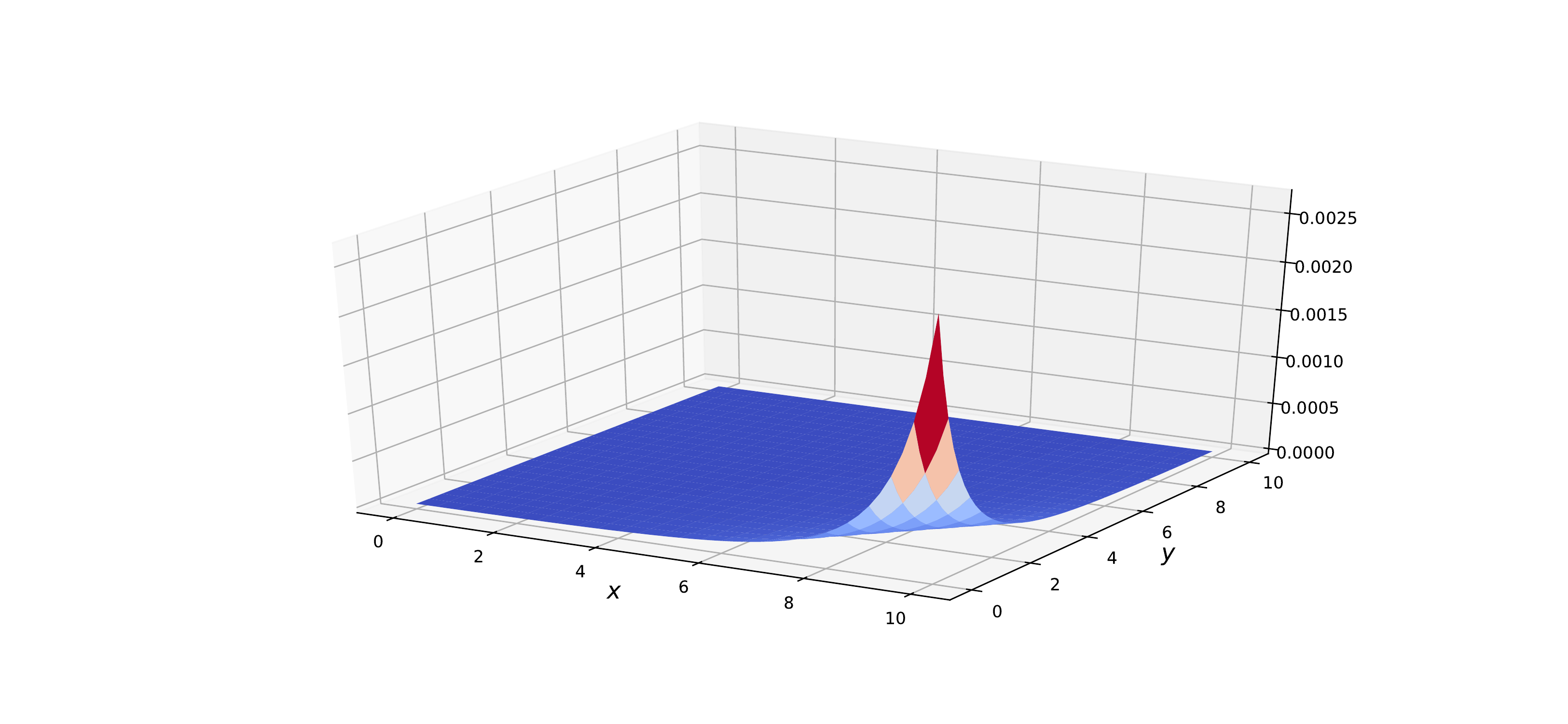}
\end{center}
\caption{Computed error for discretization \eqref{eq:FDA3}\label{fig7}, maximum error = 0.0026833687620488877}
\end{figure}

\begin{figure}
\begin{center}
\includegraphics[width=1.0\textwidth]{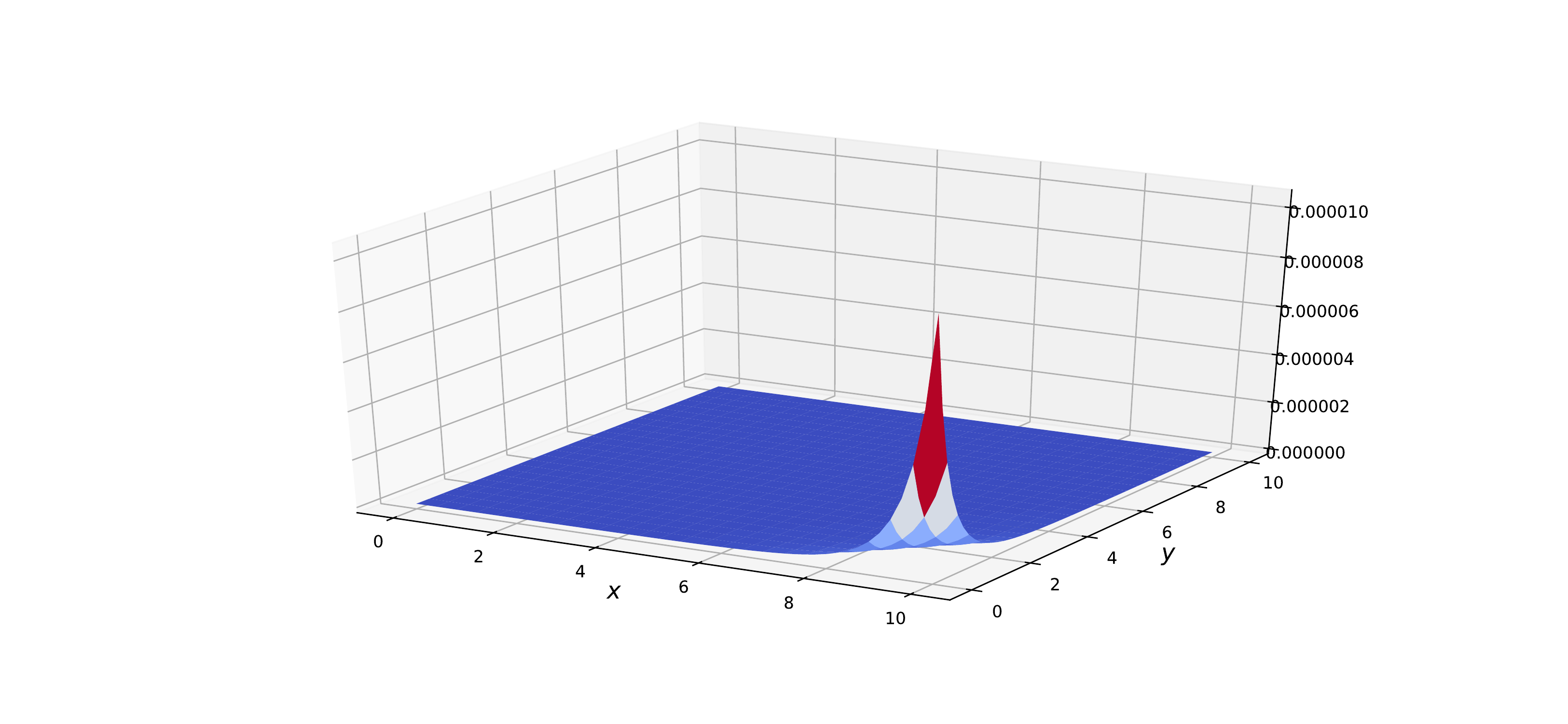}
\end{center}
\caption{Computed error for discretization \eqref{eq:FDA3_mod}\label{fig8}, maximum error = 1.0482200407964845e-05}
\end{figure}

\end{example}


\begin{example}\label{ex:Example2}
Let us consider another quasi-linear system of PDEs
\begin{equation}\label{eq:PDEsystem2}
\left\{ \begin{array}{rcl}
\displaystyle u_x - u \, v & = & 0\,,\\[1em]
\displaystyle u_y + u \, v & = & 0\,,
\end{array} \right.
\qquad \qquad u = u(x,y), \quad v = v(x,y)\,.
\end{equation}
We define the differential polynomial ring $\differentialring = \differentialfield\{ u, v \}$ with commuting
derivations $\partial_x$ and $\partial_y$, and we
use the elimination ranking $\differentialranking$ on $\differentialring$ satisfying
\[
\ldots \differentialranking v_{yy}
\differentialranking v_{xy}
\differentialranking v_{xx}
\differentialranking v_y
\differentialranking v_x
\differentialranking v
\differentialranking \ldots
\differentialranking u_{yy}
\differentialranking u_{xy}
\differentialranking u_{xx}
\differentialranking u_y
\differentialranking u_x
\differentialranking u
\]
The passivity check involves a pseudo-reduction of the second equation in \eqref{eq:PDEsystem2},
multiplied by $u$, modulo the derivative of the first equation with respect to $y$.
Hence, the case distinction $u = 0 \vee u \neq 0$ is made.
If $u = 0$, then \eqref{eq:PDEsystem2} reduces to
\begin{equation}\label{1stSubsystem}
\left\{ \begin{array}{rcl}
u & = & 0\,,\\[0.5em]
v_y & = & 0\,,
\end{array} \right.
\end{equation}
which is a simple differential system. If $u \neq 0$, then pseudo-reduction yields
\[
u \, (v_y + u \, v) + \partial_y \, (u_x - u \, v) \, \, = \, \,
(u^2 - u_y) \, \underline{v} + u_{x,y}\,.
\]
Another pseudo-reduction of the last differential polynomial modulo the
first equation in (\ref{eq:PDEsystem2}) gives the simple differential system
(with underlined leaders)
\begin{equation}\label{eq:simpledifferentialsystem}
\left\{ \begin{array}{rcl}
\displaystyle \underline{u} & \neq & 0\,,\\[0.5em]
\displaystyle u_x - u \, \underline{v} & = & 0\,,\\[0.5em]
\displaystyle u \, \underline{u_{x,y}} + (u^2 - u_y) \, u_x & = & 0\,.
\end{array} \right.
\end{equation}

We investigate the system of difference equations
\begin{equation}\label{eq:FDA_1}
\left\{ \begin{array}{rcl}
\displaystyle \frac{\tilde{u}_{i+1,j} - \tilde{u}_{i,j}}{h_1} - \tilde{u}_{i,j} \, \tilde{v}_{i,j} & = & 0\,,\\[1em]
\displaystyle \frac{\tilde{v}_{i,j+1} - \tilde{v}_{i,j}}{h_2} + \tilde{u}_{i,j} \, \tilde{v}_{i,j} & = & 0\,,
\end{array} \right.
\end{equation}
which is obtained as discretization of (\ref{eq:PDEsystem2}) by replacing
$\partial_{x} u$ and $\partial_{y} v$ by the corresponding forward differences,
with step sizes $h_1$ and $h_2$, respectively. Denote by $\tilde{f}_1$ and $\tilde{f_2}$
the left hand sides in (\ref{eq:FDA_1}).
We use the difference ranking $\differenceranking$ (cf.\ Definition~\ref{DifferenceRanking})
on $\differencering = \differencefield\{ \tilde{u}, \tilde{v} \}$
that corresponds to the differential one, namely
\[
\ldots \differenceranking \tilde{v}_{i+2,j}
\differenceranking \tilde{v}_{i,j+1}
\differenceranking \tilde{v}_{i+1,j}
\differenceranking \tilde{v}_{i,j}
\differenceranking \ldots
\differenceranking \tilde{u}_{i+2,j}
\differenceranking \tilde{u}_{i,j+1}
\differenceranking \tilde{u}_{i+1,j}
\differenceranking \tilde{u}_{i,j}\,.
\]
The passivity check yields the consequence
\[
h_2 \, \tilde{u}_{i,j+1} \, \tilde{f}_2 + \automorphism_2 \tilde{f}_1 \, \, = \, \,
(h_2 \, \tilde{u}_{i,j} - 1) \, \tilde{u}_{i,j+1} \, \underline{\tilde{v}_{i,j}} +
\frac{\tilde{u}_{i+1,j+1} - \tilde{u}_{i,j+1}}{h_1}\,.
\]
The continuous limit of this difference polynomial $\tilde{f}_3$ is given by
\[
u_x - u \, v\,.
\]
The above pseudo-reduction assumed that $\tilde{u}_{i,j+1}$ does not vanish.
If $\tilde{u}_{i,j+1} = 0$, then $\tilde{u}_{i,j} = 0$, and we obtain the simple difference system
\[
\left\{ \begin{array}{rcl}
\tilde{u}_{i,j} & = & 0\,,\\[0.5em]
\tilde{v}_{i,j+1} - \tilde{v}_{i,j} & = & 0\,.
\end{array} \right.
\]
Otherwise, we continue with the generic case by applying pseudo-reduction to $\tilde{f}_3$ modulo $\tilde{f}_1$,
which yields the remainder
\[
\begin{array}{rcl}
& & \tilde{u}_{i,j} \, \tilde{f}_3 + (h_2 \, \tilde{u}_{i,j} - 1) \, \tilde{u}_{i,j+1} \, \tilde{f}_1\\[0.5em]
& = &
(\tilde{u}_{i,j} \, \underline{\tilde{u}_{i+1,j+1}} + (h_2 \, \tilde{u}_{i,j} - 1) \, \tilde{u}_{i,j+1} \, \tilde{u}_{i+1,j} - h_2 \, \tilde{u}_{i,j}^2 \, \tilde{u}_{i,j+1}) / h_1\,.
\end{array}
\]
The continuous limit of this difference polynomial $\tilde{f}_4$ is given by
\[
(u \, \underline{u_{x,y}} + (u^2 - u_y) \, u_x) \, h_2\,.
\]
We obtain the simple difference system
\[
\left\{ \begin{array}{rcl}
h_1 \, \tilde{f}_1 \, \, = \, \, \tilde{u}_{i+1,j} - \tilde{u}_{i,j} - h_1 \, \tilde{u}_{i,j} \, \underline{\tilde{v}_{i,j}} & = & 0\,,\\[0.5em]
h_1 \, \tilde{f}_4 \, \, = \, \, \tilde{u}_{i,j} \, \underline{\tilde{u}_{i+1,j+1}} + (h_2 \, \tilde{u}_{i,j} - 1) \, \tilde{u}_{i,j+1} \, \tilde{u}_{i+1,j} - h_2 \, \tilde{u}_{i,j}^2 \, \tilde{u}_{i,j+1} & = & 0\,,\\[0.5em]
\underline{\tilde{u}_{i,j}} & \neq & 0\,,
\end{array} \right.
\]
which is s-consistent with (\ref{eq:simpledifferentialsystem}).

\medskip

Alternatively, the same difference system (\ref{eq:FDA_1}) can be checked for s-con\-sis\-ten\-cy with (\ref{eq:simpledifferentialsystem})
by using difference Gr\"obner bases.
If we choose the monomial ordering $\succeq \, \, = \, \, \ge^{\rm TOP}_{\rm degrevlex}$ (cf.\ Appendix~\ref{sec:GBtermorder}),
then the leading monomial of
the equations in (\ref{eq:FDA_1}) are the underlined ones in
\[
\tilde{f}_1 \, \, = \, \, (\underline{\tilde{u}_{i+1,j}} - \tilde{u}_{i,j}) / h_1 - \tilde{u}_{i,j} \, \tilde{v}_{i,j}\,,
\qquad
\tilde{f}_2 \, \, = \, \, (\underline{\tilde{v}_{i,j+1}} - \tilde{v}_{i,j}) / h_2 + \tilde{u}_{i,j} \, \tilde{v}_{i,j}\,.
\]
Reduction of the S-polynomial of $\tilde{f}_1$ and $\tilde{f}_2$ modulo (\ref{eq:FDA_1}) yields
\[
\tilde{v}_{i,j+1} \, h_1 \, \tilde{f}_1 - \tilde{u}_{i+1,j} \, h_2 \, \tilde{f}_2
+ h_1 \, h_2 \, \tilde{u}_{i,j} \, \tilde{v}_{i,j} \, (\tilde{f}_1 + \tilde{f}_2)
- h_1 \, \tilde{v}_{i,j} \, (\tilde{f}_1 - \tilde{f}_2) \, \, = \, \, 0\,.
\]
Hence, by Proposition~\ref{B-criterion},
$\tilde{f}_1$ and $\tilde{f}_2$ form a difference Gr\"obner basis (cf.~Definition~\ref{def_SB}) of the
difference ideal defined by (\ref{eq:FDA_1}), confirming s-consistency of the difference approximation~(\ref{eq:FDA_1})
with the PDE system (\ref{eq:PDEsystem2}) again (cf.\ Theorem~\ref{GB:s-constistency}).

\medskip

Next we consider the discretization obtained by replacing
$\partial_{x} u$ by the forward difference as before and $\partial_{y} v$ by the backward difference:
\begin{equation}\label{eq:FDA_2}
\left\{
\begin{array}{rcl}
\displaystyle \frac{\tilde{u}_{i+1,j} - \tilde{u}_{i,j}}{h_1} - \tilde{u}_{i,j} \, \tilde{v}_{i,j} & = & 0\,,\\[1em]
\displaystyle \frac{\tilde{v}_{i,j+1} - \tilde{v}_{i,j}}{h_2} + \tilde{u}_{i,j+1} \, \tilde{v}_{i,j+1} & = & 0\,,
\end{array}
\right.
\end{equation}
again with step sizes $h_1$ and $h_2$, respectively. Denote by $\tilde{f}'_1$ and $\tilde{f}'_2$
the left hand sides in (\ref{eq:FDA_2}).
We use the degree-reverse lexicographical ranking $\differenceranking$
with $\tilde{v} \differenceranking \tilde{u}$
on the difference polynomial ring $\differencering = \differencefield\{ \tilde{u}, \tilde{v} \}$, namely
\[
\ldots \differenceranking \tilde{v}_{i+2,j}
\differenceranking \tilde{u}_{i+2,j}
\differenceranking \tilde{v}_{i,j+1}
\differenceranking \tilde{u}_{i,j+1}
\differenceranking \tilde{v}_{i+1,j}
\differenceranking \tilde{u}_{i+1,j}
\differenceranking \tilde{v}_{i,j}
\differenceranking \tilde{u}_{i,j}\,.
\]
The leaders of $\tilde{f}'_1$ and $\tilde{f}'_2$
are $\tilde{u}_{i+1,j}$ and $\tilde{v}_{i,j+1}$, respectively.
Since these involve different indeterminates $\tilde{u}$ and $\tilde{v}$, passivity is
ensured, but a case distinction regarding the initial of $\tilde{f}'_2$
leads to a splitting. Thus a difference decomposition of (\ref{eq:FDA_2}) is
\[
\left\{ \begin{array}{rcl}
\displaystyle (h_2 \, \tilde{u}_{i,j+1} + 1) \, \underline{\tilde{v}_{i,j+1}} - \tilde{v}_{i,j} & = & 0\,,\\[1em]
\displaystyle \underline{\tilde{u}_{i+1,j}} - \tilde{u}_{i,j} - h_1 \, \tilde{u}_{i,j} \, \tilde{v}_{i,j} & = & 0\,,\\[1em]
\displaystyle h_2 \, \underline{\tilde{u}_{i,j+1}} + 1 & \neq & 0\,,
\end{array} \right.
\qquad \vee \qquad
\left\{ \begin{array}{rcl}
\displaystyle \tilde{v}_{i,j} & = & 0\,,\\[1em]
\displaystyle h_2 \, \tilde{u}_{i,j} + 1 & = & 0\,.
\end{array} \right.
\]
Since the continuous limits of the equations in the first simple system
are in the radical differential ideal corresponding to (\ref{eq:PDEsystem2})
due to w-consistency of (\ref{eq:FDA_2}), we conclude that (\ref{eq:FDA_2})
is s-consistent with (\ref{eq:PDEsystem2}). For $h_2 \to 0$ the second simple
system yields the contradiction $1 = 0$, so that this case and the inequation
in the first simple system can be ignored.

\medskip

Alternatively, using difference Gr\"obner bases, we may choose the
monomial ordering $\succeq \, \, = \, \, \ge^{\rm TOP}_{\rm degrevlex}$ with $\tilde{u} \differenceranking \tilde{v}$
(cf.\ Appendix~\ref{sec:GBtermorder}).
Then the leading monomials of $\tilde{f}'_1$ and $\tilde{f}'_2$ are $\tilde{u}_{i+1,j}$ and $\tilde{u}_{i,j+1} \, \tilde{v}_{i,j+1}$,
respectively. The passivity check yields
\[
\begin{array}{rcl}
& & \sigma_1 \tilde{f}'_2 - h_1 \, \tilde{v}_{i+1,j+1} \, (\sigma_2 \tilde{f}'_1 + \tilde{f}'_2)\\[0.75em]
& = & \displaystyle
\underline{\tilde{u}_{i,j+1} \, \tilde{v}_{i+1,j+1}}
- \frac{h_1 \, \tilde{v}_{i,j+1} \, \tilde{v}_{i+1,j+1} - \tilde{v}_{i+1,j+1} - h_1 \, \tilde{v}_{i,j} \, \tilde{v}_{i+1,j+1} + \tilde{v}_{i+1,j}}{h_2}\,.
\end{array}
\]
The continuous limit of this difference polynomial $\tilde{f}'_3$ is given by
$v_y + u \, v$. A further reduction yields
\[
\begin{array}{rcl}
& & \tilde{v}_{i,j+1} \, \tilde{f}'_3 - \tilde{v}_{i+1,j+1} \, \tilde{f}'_2\\[0.5em]
& = &
-(h_1 \, \underline{\tilde{v}_{i,j+1}^2 \, \tilde{v}_{i+1,j+1}}
- h_1 \, \tilde{v}_{i,j} \, \tilde{v}_{i,j+1} \, \tilde{v}_{i+1,j+1}
- \tilde{v}_{i,j} \, \tilde{v}_{i+1,j+1} + \tilde{v}_{i,j+1} \, \tilde{v}_{i+1,j}) / h_2\,.
\end{array}
\]
The continuous limit of this difference polynomial $\tilde{f}'_4$ is given by
\[
\left( v \, v_{x,y} - v_y \, (v_x + v^2) \right) h_1\,.
\]
Note that the coefficient of $h_1$ is the linear combination
\[
(v \, \partial_{x} - v_x - v^2) \, (v_y + u \, v) - v^2 \, (u_{x} - u \, v)
\]
of the original equations in (\ref{eq:PDEsystem2}).
The reduction of the final S-polynomial is
\[
\begin{array}{l}
h_1 \, \tilde{v}_{i,j+1}^2 \, \tilde{f}'_3 + h_2 \, \tilde{u}_{i,j+1} \, \tilde{f}'_4
- (h_1 \, \tilde{v}_{i,j} \, \tilde{v}_{i,j+1} + \tilde{v}_{i,j}) \, \tilde{f}'_3 \, +\\[0.5em]
\qquad \qquad \qquad \qquad \qquad \qquad
(h_1 \, \tilde{v}_{i,j} - h_1 \, \tilde{v}_{i,j+1} + 1) \, \tilde{f}'_4
+ \tilde{v}_{i+1,j} \, \tilde{f}'_2 \, \, = \, \, 0\,.
\end{array}
\]
Hence, we again conclude that (\ref{eq:FDA_2})
is s-consistent with (\ref{eq:PDEsystem2}).
\end{example}

%
%
\section{Conclusion}\label{sec:conclusion}

We extended the notion of s(trong)-consistency for FDA~\eqref{fda}, introduced and studied  in~\cite{GR'12,G'12,GR19} for the Cartesian grids $(h_1=h_2=\cdots =h_n)$, to the regular ones, where the grid spacings $h_i$ may be pairwise different. This notion for a finite difference discretization~\eqref{fda} of PDE~\eqref{pde} satisfying the condition~\eqref{def-wcons} in Definition 5.1 means that any element $\tilde{f}$ in the difference ideal $[\tilde{F}]$, as well as any in its perfect closure  $\llbracket \tilde{F}\rrbracket$, after appropriate normalization
(cf.\ \eqref{s-cond-consequence}), approximates an element $f\in \llbracket F\rrbracket$ in the radical differential ideal. Thereby,
the algebraic properties of discrete (finite difference) equations, characterized by the perfect difference ideal they generate, {\em mimic} the algebraic properties of differential equations characterized by the radical differential ideal generated by these equations.

By using the method of difference \Gr bases we derived a new s-con\-sis\-tent and conservative FDA~\eqref{DNSE1}--\eqref{DIntCond} to the three-dimensional incompressible Navier-Stokes equations. This discretization allows to solve numerically the last equations in the  pressure-Poisson formulation when the pressure is determined from the Poisson pressure equation and the velocities from the momentum equations. Our numerical experiments with the two-di\-men\-sio\-nal analogue~\eqref{FDA1} of the new scheme have clearly demonstrated its superiority over the other two-dimensional schemes~\eqref{FDA2}--\eqref{FDA4}. In particular, the scheme reveals, at the discrete level, a
surprisingly high accuracy preservation of the mass conservation law (continuity equation).
This law is satisfied by the initial condition, but is not employed in the subsequent construction of the numerical solution.

In general, the techniques of difference \Gr bases cannot be applied to the s-consistency analysis of FDA to nonlinear PDE systems, since termination of Algorithm~\ref{StandardBasis} is not guaranteed. Instead, the fully algorithmic triangular difference Thomas decomposition, designed last year in our conference paper~\cite{GR19} and described in Section~\ref{DifferenceThomasDecomposition}, can be applied. Before its application, we suggest first to apply the differential Thomas decomposition to the input PDE~\eqref{pde}.
Each output subsystem is simple (cf.\ Definition~\ref{de:differentialsimple}),
which, e.g., clarifies the arbitrariness of power series solutions to the system~\cite{LangeHegermann} and thus facilitates formulating well-posed initial value problems in the sense of Hadamard~\cite{Hadamard'1902} (cf.\ \cite[Example~14]{GLHR'18} for
the case of the three-dimensional Navier-Stokes equations).
Disjointness of the decomposition, i.e., partition of the solution space by the output subsystems, allows to confine the investigation to the unique simple system admitting a solution of interest. After a discretization of the input simple differential system providing its w-consistency (cf.\ Definition~\ref{def-wcons}) we apply Algorithm~\ref{alg:differencedecomp} to the FDA. For different simple systems different ways to discretize may be chosen.
Algorithm~\ref{alg:differencedecomp} is the main one, it provides the difference Thomas decomposition into quasi-simple subsystems (cf.\,Definition~\ref{de:differencesimple}). It is based on two subalgorithms: Algorithm~\ref{alg:autoreducenonlin} performing difference auto-re\-duc\-tion and Algorithm~\ref{alg:janetreducenonlin} computing Janet normal forms of difference polynomials. Finally, Algorithm~\ref{alg:discretization} performs the s-con\-sis\-ten\-cy analysis for the input FDA with the simple differential system. Since the difference Thomas decomposition partitions the solution space of the FDA, s-con\-sis\-ten\-cy holds if and only if every difference equation in each output subsystem approximates an element in the radical differential ideal generated by the elements in the input simple differential system. In the recent paper~\cite{Li'20} it is argued that if a differential (or difference) decomposition algorithm terminates on every input, then one can provide a computable upper bound for the size of its output in terms of the input, i.e., an upper bound for number of output subsystems, their order and degree. Because of the termination of both decomposition algorithms, the upper bound estimation approach of paper~\cite{Li'20} is applicable to differential and difference Thomas decompositions. 

For illustration we applied both methods, the one based on difference Gr\"ob\-ner bases and the one based on difference Thomas decomposition, to the s-con\-sisten\-cy analysis of finite difference discretizations of two first-or\-der quasi-li\-near PDE systems (Section 7) with two independent variables. The first PDE system~\eqref{eq:PDEsystem1} is overdetermined and has a consequence of the conservation law form
\begin{equation}\label{CLForm}
\partial_xu+\partial_yu \, \, = \, \, 0\,.
\end{equation}
If one approximates the partial derivatives in Eqs.~\eqref{eq:PDEsystem1} by forward differences, then its difference $S$-polynomial in the continuous limit yields the equation $u^4=0$, which does not follow from Eqs.~\eqref{eq:PDEsystem1}. Therefore, the FDA~\eqref{eq:FDA1} is s-in\-con\-sis\-tent. Another FDA~\eqref{eq:FDA2} combining the forward and backward differences for derivatives yields an $S$-polynomial that shows that s-con\-sis\-tency is equivalent to $h_1 = h_2$.
The third discretization~\eqref{eq:FDA3} based on approximation of both partial derivatives by central differences also has a passivity condition, whose continuous limit \eqref{EqualSpacings} allows to conclude that $h_1 = h_2$ is a necessary condition for s-con\-sis\-ten\-cy.
In case $h_1=h_2=h$ subsequent passivity checks produce rather large expressions. However, by applying backward shifts to intermediate difference polynomials when possible, the s-consistency analysis is drastically simplified, yielding differential polynomials as
continuous limits which occur in the left-hand sides of Eqs.~\eqref{eq:PDEsystem1} or their sum~\eqref{CLForm} or their difference.
Hence, for Cartesian grid, i.e., for equisized grid spacings, both difference approximations~\eqref{eq:FDA2} and~\eqref{eq:FDA3} are s-con\-sis\-tent. This example shows that s-con\-sis\-tency may place constraints on the grid spacings. Furthermore, for both s-consistent FDA~\eqref{eq:FDA2} and~\eqref{eq:FDA3} we constructed modified equations and applied them to analyze the actual accuracy of those FDA and to increase their accuracy. Additionally, we used the exact solution~\eqref{exact:sol} to~\eqref{eq:PDEsystem1} for the numerical construction of this solution, verifying experimentally the theoretically predicted accuracy.

The second quasi-linear PDE system~\eqref{eq:PDEsystem2} of Section~\ref{sec:examples} has two dependent variables. First, we apply to this system the differential Thomas decomposition which splits
Eqs.~\eqref{eq:PDEsystem2} into two simple systems~\eqref{1stSubsystem} and~\eqref{eq:simpledifferentialsystem}. It is easy to see that
any (w-consistent) FDA to system~\eqref{1stSubsystem} is s-consistent due to the lack of passivity conditions. As to FDA for~\eqref{eq:simpledifferentialsystem} one can replace partial derivatives by the corresponding forward differences.  This produces a simple difference system providing an s-consistent approximation to~\eqref{eq:simpledifferentialsystem}.
Moreover, we also established the compatibility of the differential and the difference Thomas decomposition by
starting with \eqref{eq:PDEsystem2} and discretizing its equations by forward differences.
The difference Thomas decomposition again produces the same discrete version of~\eqref{eq:simpledifferentialsystem}.
Alternatively, we applied the method of difference Gr\"ob\-ner bases to verify our results. Next we considered the discretization of~\eqref{eq:PDEsystem2} by using the forward difference to approximate $u_x$ and the backward difference to approximate $v_y$. We established the s-consistency by means of the Gr\"ob\-ner basis method using one of the monomial orderings described in Appendix~\ref{sec:GBtermorder}.

Concerning implementation of (nonlinear) difference Gr\"ob\-ner basis construction, the only one is realized in~\cite{LaScala'15,GLS'15}. There the problem of computation in a difference polynomial ring is reduced to computation in the ring of commutative polynomials whose set of variables is extended with their shifts obtained by the action of the elements in~\eqref{MonoidOfShifts}. In this case, under the assumption of an admissible monomial ordering compatible with the {\em order function} defined in \cite[Def.\,4.1]{GLS'15} and for difference ideals that admit finite Gr\"ob\-ner bases one can use the algorithm designed in \cite[Alg.\,4.1]{GLS'15} to compute such a basis in a finite number of steps. This algorithm, implemented in Maple \cite{GLS'15}, may cause a quite considerable growth of the number of variables involved in the computation, and thus restricts applicability to rather small problems. The difference Thomas decomposition has not yet been implemented. All computations with difference polynomials presented in the paper were done ``by hand'' using Maple for simplification of intermediate expressions.

\bigskip

\noindent
\textbf{Acknowledgements}
The work of V.P.\,Gerdt (Sections 1,3-5,9; Section 6, except Fig.1-5; Algorithm 5 of Section 7) is supported by the Russian Science Foundation under grant No. 20-11-202574.

\appendix
\section{Monomial ordering for difference Gr\"obner basis}\label{sec:GBtermorder}

Using the identification $\tilde{u}^{(r)}_{a_1,a_2} = \automorphism^{\mathbf{a}} \tilde{u}^{(r)}$
for $\mathbf{a} \in (\Z_{\ge 0})^2$, $1 \le r \le m$,
we define a total ordering $\differencemonomialordering$ on the set of monomials in the infinitely many indeterminates
$\automorphism^{\mathbf{a}} \tilde{u}^{(r)}$, where $\mathbf{a} \in (\Z_{\ge 0})^2$, $r \in \{ 1, 2 \}$, as follows:
\[
\begin{array}{l}
\displaystyle
\prod_{i=1}^{d} \automorphism^{\mathbf{a}_i} \tilde{u}^{(r_i)} \, \, \differencemonomialordering \, \,
\prod_{i=1}^{e} \automorphism^{\mathbf{b}_i} \tilde{u}^{(s_i)}
\qquad :\Longleftrightarrow \\[2em]
\left\{ \begin{array}{l}
\displaystyle
\quad
\sum_{i=1}^{d} |\mathbf{a}_i| \, \, > \, \,
\sum_{i=1}^{e} |\mathbf{b}_i| \qquad \mbox{or}
\qquad
\Big( \,
\sum_{i=1}^{d} |\mathbf{a}_i| \, \, = \, \,
\sum_{i=1}^{e} |\mathbf{b}_i| \qquad \mbox{and}\\[1.5em]
\displaystyle
\quad \qquad \quad
( \, (\mathbf{a}_{j_1}, r_{j_1}), \, \ldots, \, (\mathbf{a}_{j_d}, r_{j_d}) \, ) \, \, >_{{\rm lex}} \, \,
( \, (\mathbf{b}_{k_1}, s_{k_1}), \, \ldots, \, (\mathbf{b}_{k_e}, s_{k_e}) \, ) \, \Big)\,,
\end{array}
\right.
\end{array}
\]
where
\[
(\mathbf{a}_{j_1}, r_{j_1}) \, \, \ge \, \,
(\mathbf{a}_{j_2}, r_{j_2}) \, \, \ge \, \,
\ldots \, \, \ge \, \, (\mathbf{a}_{j_d}, r_{j_d})
\]
and
\[
(\mathbf{b}_{k_1}, s_{k_1}) \, \, \ge \, \,
(\mathbf{b}_{k_2}, s_{k_2}) \, \, \ge \, \,
\ldots \, \, \ge \, \, (\mathbf{b}_{k_e}, s_{k_e})
\]
are the tuples $(\mathbf{a}_i, r_i)$ and $(\mathbf{b}_i, s_i)$ arranged
in decreasing order with respect to the ordering $\ge$ used for breaking ties,
and where $>_{{\rm lex}}$ is the lexicographic ordering comparing tuple entries
with respect to $\ge$. The ordering $(\mathbf{a}_{j_i}, r_{j_i}) \ge (\mathbf{b}_{k_i}, s_{k_i})$
is assumed to respect addition of a multi-index $\mathbf{c}$
to $\mathbf{a}_{j_i}$ and $\mathbf{b}_{k_i}$.
In Section~\ref{sec:examples}, where we let $\tilde{u} = \tilde{u}^{(1)}$ and $\tilde{v} = \tilde{u}^{(2)}$,
we choose
$\ge \, \, = \, \, \ge^{\rm TOP}_{\rm degrevlex}$, which is defined by
\[
\begin{array}{l}
((a_1, a_2), r) \, \, \ge^{\rm TOP}_{\rm degrevlex} \, \, ((a'_1, a'_2), r') \qquad
:\Longleftrightarrow\\[1.5em]
\left\{ \begin{array}{l}
a_1 + a_2 \, \, > \, \, a'_1 + a'_2 \qquad \mbox{or} \qquad
\Big( \, a_1 + a_2 \, \, = \, \, a'_1 + a'_2 \quad \mbox{and} \quad
\big( \, a_2 \, \, < \, \, a'_2\, \quad \mbox{or}\\[0.5em]
\qquad \qquad \qquad \qquad \qquad \qquad \qquad \qquad \qquad \qquad
a_2 \, \, = \, \, a'_2 \quad \mbox{and} \quad r \, \, \le \, \, r' \, \big) \Big)\,.
\end{array} \right.
\end{array}
\]
For elimination purposes one may choose $\ge \, \, = \, \, \ge^{\rm POT}_{\rm degrevlex}$, which is defined by
\[
\begin{array}{l}
((a_1, a_2), r) \, \, \ge^{\rm POT}_{\rm degrevlex} \, \, ((a'_1, a'_2), r') \qquad
:\Longleftrightarrow\\[1.5em]
\left\{ \begin{array}{l}
r \, \, > \, \, r' \quad \mbox{or} \quad
\Big( \, r \, \, = \, \, r' \quad \mbox{and} \quad
\big( \, a_1 + a_2 \, \, > \, \, a'_1 + a'_2\, \quad \mbox{or}\\[0.5em]
\qquad \qquad \qquad \qquad \qquad \qquad \qquad
a_1 + a_2 \, \, = \, \, a'_1 + a'_2 \quad \mbox{and} \quad a_2 \, \, \le \, \, a'_2 \, \big) \Big)\,.
\end{array} \right.
\end{array}
\]
It is clear that we have $\tilde{t} \differencemonomialordering 1$ for very difference monomial $\tilde{t} \neq 1$.
Suppose that the difference monomials $\tilde{v}$ and $\tilde{w}$
satisfy $\tilde{v} \differencemonomialordering \tilde{w}$
and let $\tilde{t}$ be another difference monomial and $\theta \in \setauto$.
Then either the sum of shifts in $\tilde{v}$ is greater than the sum of shifts in $\tilde{w}$, in which case
the same statement holds for $\tilde{t} \cdot \theta \circ \tilde{v}$ compared to $\tilde{t} \cdot \theta \circ \tilde{w}$,
or the sums of shifts are equal and, in the above notation, either the lexicographic ordering
identifies an index $i$ such that $(\mathbf{a}_{j_i}, r_{j_i}) > (\mathbf{b}_{k_i}, s_{k_i})$
with respect to the ordering used for breaking ties, or
$( \, (\mathbf{b}_{k_1}, s_{k_1}), \, \ldots, \, (\mathbf{b}_{k_e}, s_{k_e}) \, )$
is a proper prefix of
$( \, (\mathbf{a}_{j_1}, r_{j_1}), \, \ldots, \, (\mathbf{a}_{j_d}, r_{j_d}) \, )$.
In the latter situations, application of $\theta$ is respected by $\ge$, whereas
multiplication by $\tilde{t}$ leads to insertion of the pair corresponding to $\tilde{t}$
at appropriate positions in the above tuples, which is respected by the lexicographic ordering.
Hence, we conclude $\tilde{t} \cdot \theta \circ \tilde{v} \differencemonomialordering \tilde{t} \cdot \theta \circ \tilde{w}$ in any case.
Therefore, according to Definition~\ref{DifferenceMonomialOrdering}, $\differencemonomialordering$
is an admissible difference monomial ordering.

%
%



\end{document}